\documentclass[11pt]{article}
\usepackage{hyperref}
\usepackage{graphicx}
\usepackage{graphics}
\usepackage{dsfont}
\usepackage{epsfig}
\usepackage{amsmath,amssymb,amsthm,amscd}
\allowdisplaybreaks
\usepackage[ansinew]{inputenc}
\usepackage{rotating} 

\newcommand{\be}{\begin{equation}}
\newcommand{\ee}{\end{equation}}

\newtheorem{conj}{Conjecture}

\newtheorem{prop}{Proposition}
\newtheorem{lem}{Lemma}
\newtheorem{cor}{Corollary}

\numberwithin{equation}{section}

\begin{document}
\begin{titlepage}
{}~ \hfill\vbox{ \hbox{} }\break

\rightline{USTC-ICTS-15-02} 
\rightline{BONN-TH-2015-01}

\vskip 3 cm

\centerline{
\Large \bf  
Topological String on elliptic CY 3-folds and the ring of Jacobi forms}   
\vskip 0.5 cm

\renewcommand{\thefootnote}{\fnsymbol{footnote}}
\vskip 30pt \centerline{ {\large \rm Min-xin Huang
\footnote{minxin@ustc.edu.cn}, Sheldon Katz\footnote{katz@math.uiuc.edu} and Albrecht Klemm
\footnote{aklemm@th.physik.uni-bonn.de}  } } \vskip .5cm \vskip 30pt

\begin{center}
{\em $^*$ Interdisciplinary Center for Theoretical Study,  \\ \vskip 0.2 cm  University of Science and Technology of China,  Hefei, Anhui 230026, China} \\ [5 mm]
{\em $^\dagger$   Department of Mathematics  \\ \vskip 0.2 cm 
University of Illinois at Urbana-Champaign, 1409 W.\ Green St., Urbana, IL 6180   } \\ [5 mm]
{\em $^\ddagger$ Bethe Center for Theoretical Physics (BCTP),   \\ \vskip 0.2 cm 
Physikalisches Institut, Universit\"at Bonn,  53115 Bonn, Germany}
\end{center}

\setcounter{footnote}{0}
\renewcommand{\thefootnote}{\arabic{footnote}}
\vskip 60pt
\begin{abstract}
We give evidence  that the all genus amplitudes of topological string 
theory on compact elliptically fibered Calabi-Yau manifolds can be written 
in terms of meromorphic  Jacobi forms whose weight grows linearly and whose index  
grows quadratically with the base degree. The denominators of these forms 
have a simple universal form with the property that the poles of the meromorphic 
form  lie only at torsion points. The modular parameter corresponds to the 
fibre class while the r\^ole of the string coupling is played by the elliptic 
parameter. This leads to very strong all genus results on these geometries, 
which are checked against results from curve counting. The structure can be 
viewed as an indication that an $N=2$ analog of the reciprocal of the Igusa 
cusp form exists that might govern the topological string theory on these 
Calabi-Yau manifolds completely.   
\end{abstract}

\end{titlepage}
\vfill \eject


\newpage

\baselineskip=16pt

\tableofcontents

\section{Introduction and Summary}

Topological strings on non-compact Calabi-Yau geometries are solvable and have 
a very interesting structure with a wealth of connections to gauge theories, 
integrable models, large N-dualities, Chern-Simons theories, supersymmetric 
localisation  and matrix models. In fact topological string techniques 
had considerable influence on these areas and solved important problems, 
see~\cite{Drukker:2010nc} for a recent example.  As a benchmark problem in 
topological  string theory with a similar wealth of connections to quantum gravity problems,
remains the solution of the topological string on compact Calabi-Yau manifolds.
In this paper we report on some progress in this matter.  

According to~\cite{BCOV} the genus $g$ topological string amplitudes $F_g(S^{ij},S^j,S;\underline{z})$ 
are inhomogeneous polynomials  of weighted degree $3g -3$ in the anholomorphic propagators  
$S^{ij}, S^j,S$, which have respective weights $(1,2,3)$, with rational functions of the 
moduli $\underline{z}$ as coefficients. The holomorphic anomaly determines all of 
$F_g$, but the weight zero piece $f_g$, whose  determination  is the key conceptual obstacle to solving the topological 
string on compact Calabi-Yau spaces.  A careful analysis of the conifold gap 
condition~\cite{HKQ}, regularity at the orbifold and Castelnuovo's criterium for 
the vanishing of higher genus curves reveals that the $f_g$ can be determined 
up genus $51$\footnote{That condition has been derived under the assumption 
that the ring of modular generators on the quintic has no relations at high 
weight, which has  been  checked up to the weight needed for 
genus 30.} for the quintic~\cite{HKQ}. While for non-compact models the same conditions lead to a 
complete solution of the closed topological string on these 
geometries~\cite{Haghighat:2008gw}~\footnote{In fact our key example  
$X_{18}(1,1,1,6,9)$ is an elliptic CY manifold over $\mathbb{P}^2$ and decompactifying 
the elliptic fibre, leads to a local model solved in~\cite{Haghighat:2008gw}.}, 
on compact Calabi-Yau manifolds one needs  additional boundary condition.  

One purpose of this article is to  investigate the additional symmetries 
and boundary conditions that are specific for elliptic  fibered Calabi-Yau 
threefolds\footnote{Note that the transitions between Calabi-Yau manifolds~\cite{Reid} 
and mirror symmetry make it possible to predict from the solution of the 
topological string on one Calabi-Yau 3 fold with fibration structure the 
ones of such without fibration structure, if the complex moduli space of the 
latter is embedded into the one of the former, see~\cite{Hosono:1993} for many examples.}.
Beside the motivation to solve the topological string in general the elliptically 
fibered cases are very interesting in their own right as the topological string 
amplitudes  predict terms in the low energy effective  actions of F-theory and 
calculate indices of 6d conformal theories in particular the $E$-string and its generalizations 
with gauge theories~\cite{Haghighat:2014vxa}.           

One possible constraint could come from the  $SL(2,\mathbb{Z})$  modularity 
of the amplitudes in the K\"ahler parameter of the elliptic fibre that has been 
argued  in~\cite{Klemm:2012}~\cite{Alim:2012} based on a monodromy analysis 
in one global case~\cite{Candelas:1994hw} and calculations for more general  
elliptic threefolds\cite{Klemm:2004km}. It also  extends observations for 
local  elliptic surfaces inside Calabi-Yau spaces~\cite{Klemm:1996hh} 
and~\cite{HST,Hosono:2001,Hosono:2002xj}.

We find that the constraint from the fibre modularity is equivalent to the constraint
that comes from a $\mathbb{Z}_2$ involution symmetry $I$, that acts on the 
moduli space of elliptic fibrations as an R symmetry, i.e. it is not realized  
as a symplectic transformation like the monodromies of the model. This symmetry 
$I$  acts also on the $F_g$ by $(-1)^{(g-1)}$  and on the propagators in a way described in section \ref{involutiononthepropagators}. 
Knowing these actions we can  restrict the possible constants in the rational ambiguity 
$f_g$ to roughly one fourth. Using this information, the conifold gap and the  regularity at 
the other points in the moduli space, especially the orbifold point we can solve the topological 
string for all classes to genus $9$ on an elliptic fibration over $\mathbb{P}^2$, with 
one section.  This is already the strongest  available result for compact multiparameter 
Calabi-Yau manifolds and needs only very mild results on the vanishing of the BPS numbers.

Our main result is however as follows. Consider the topological string partition 
function $Z=\exp(F)=\exp(\sum_{g=0}^\infty \lambda^{2g-2} F_g)$ 
expanded in the base degree(s). E.g. for the $\mathbb{P}^2$ case one has~\footnote{He we dropped 
the base degree zero contributions, which is rather trivial and discussed separately.}   
\be 
Z=1+ \sum_{d_B =1}^{\infty} Z_{d_B} (\tau,\lambda) Q^{d_B}\ ,
\label{expansion} 
\ee
where $\tau$ is the complexified K\"ahler class of the fiber $\lambda$ is string coupling and 
$Q=e^{2 \pi i t_B}$ with $t_{d_B}$ the suitable defined 
complexified K\"ahler class of the base, see (\ref{redefinitionti}).  Then $Z_{d_B}$ is a quotient 
of even weak Jacobi forms of the following form\footnote{Here we use the standard 
notation of the elliptic argument $z$ which is identified with the string coupling. We hope no confusion with 
the B-model moduli, which are also traditionally denoted by $z$, occurs.}
\be 
Z_{d_B}(\tau,z)=\frac{\varphi_{d_B}(\tau,z)}{\eta^{36 d_B}(\tau)\prod_{k=1}^{d_B} \varphi_{-2,1}(\tau,k z)},
\label{ansatz}
\ee
where $\eta(\tau)$ is the Dedekind function and $\varphi_{d_B}(\tau,z)$ is an even weak Jacobi form of index  $\frac{1}{3}d_B(d_B-1)(d_B+4)$ 
and weight $16 d_B$. The number of coefficients of $\varphi_{d_B}(\tau,z)$ expressed in the basis of the 
ring of even weak forms $P:=E_4,Q:=E_6, A:=\varphi_{-2,1}$ and  $B:=\varphi_{0,1}$  
still grows very fast with $d_B^6$, however we can argue that the Castelnuovo bounds 
discussed in section \ref{GVfromgeometry} bring down the growth to $d_B^4$. 
Combined with the restriction from the constraints from the conifold gap and the 
regularity of the $F_g$ in the interior of the moduli space and with the 
independence of these conditions assumed\footnote{This holds at least for as far as we could actually 
calculate i.e. up to $d_B\leq 5$, $d_E,g$ arbitrary or  $g\leq 8$, $d_B,d_E$ arbitrary.},   
this allows to solve the model up to genus $g=189$ for arbitrary base and fibre classes, 
or up to $d_B\leq 20$ for arbitrary genus and fibre classes.

We will show in section \ref{weakjacobiHA} that the Jacobi forms fulfill a heat 
kernel type form of the holomorphic anomaly equation for $Z_{d_B}$ that is equivalent to a limit 
of the wave function equation proposed by Witten~\cite{Witten:1993ed} for the 
topological string partition function in interpreting the holomorphic anomaly equation of~\cite{BCOV} as consequence of a  
background independent quantization of $H_3(M,\mathbb{C})$. This heat kernel type 
form of the holomorphic anomaly equation further summarizes all 
the holomorphic anomaly equations for the fibre modulus, which 
were proven in~\cite{Klemm:2012}, i.e. the ansatz  (\ref{ansatz}) 
identically satisfies the holomorphic anomaly equation for the fibre 
modulus. 

We argue that the special form of the denominator in our ansatz naturally 
is crucial to satisfy the Castelnuovo bounds. The scaling of the $z$ argument 
in the denominator can interpreted as coming from multiple string windings of 
the base similar as in the elliptic genus of E-strings~\cite{Haghighat:2014,Cai:2014,Kim:2014} 
or more general strings of 6d SCFTs~\cite{Haghighat:2014vxa} with gauge symmetries. 
On the other hand, the paper~\cite{Kim:2014} constructs a 2d quiver gauge theory 
for E-strings and~\cite{Haghighat:2014vxa} for the $\hat D_4$ string 
that can in principle  compute the elliptic genus of any finite number of 
E-strings with the techniques of~\cite{Benini:2013}\footnote{It would be very 
interesting to see if such 2d quiver theory duals exist also for the 
compact threefold case, i.e. 6d theories that do include gravity.}. 

The prediction from the B-model calculations and from (\ref{expansion},\ref{ansatz}) 
are sucessfully checked in section \ref{GVfromgeometry} against results 
from enumerative geometry. Since the moduli space of the BPS states 
tends to becomes smoother for fixed degree and high $g$ we can 
make an infinite number of nontrivial checks on the form  of $Z_{d_B}$ 
and the determination of its coefficients following the methods 
proposed in~\cite{KKV}. 

A remarkable mathematical characterization of the form 
(\ref{ansatz}) is is that, again due to the form of the denominator, 
the poles of  $Z_{d_B}$ lie only at the points $z\in \mathbb{Q}\tau+ \mathbb{Q}$, 
i.e. at torsion  points of the elliptic curve. As such they have a 
modified theta expansion that involves mock modular forms~\cite{Zwegers}.
These structures played r\^ole in the understanding of 
Wall-crossing for $N=4$ BPS states~\cite{Dabholkar:2012} 
based on the reciprocal of the Igusa cusp and related 
meromophic forms.

The general structure (\ref{expansion},\ref{ansatz}) can be viewed as 
an extension of the only compact case 3-fold case $K3\times T2$ 
for which all  BPS states associated~\footnote{In order to make 
this interpretation  one has to lift zero modes that come from the 
$T^2$ and make the $N=2$ invariants trivially zero.} to moduli 
space of curves are known~\cite{KKV}. Due to the  $SU(2)$ holonomy the 
Type II compactifications  has $N=4$ supersymmetry. The prediction is 
based on the formula for the elliptic genus of symmetric products of K3~\cite{Dijkgraaf:1996xw} 
and M-theory/Type IIB duality. It has been recently proven in~\cite{Oberdieck:2014aga}. 
In this case $Z$ is related to the reciprocal of the Igusa cusp form of a genus two curve with weight
$10$ which has a Fourier expansion in $p=\exp(2 \pi i \sigma)$ 
as~\footnote{We follow the notation of~\cite{Dabholkar:2012}.}        
\begin{equation}
\label{reciprocalphi10} 
\hat Z(\Omega) =\frac{1}{\Phi_{10}(\Omega)}=
\sum_{m=-1}\psi_{-10,m}(\tau,z) p^m,\ {\rm with}\ \Omega=\left(\begin{array}{cc}\tau& z\\ z &\sigma\end{array}\right) \ . 
\end{equation}
Here $\Omega$ parametrizes the Siegel upper half space  $\mathbb{H}_2$ and $\Phi_{2,m}=\Delta (\tau) \psi_{-10,m}$ are 
meromophic Jacobi-Forms of weight two and index $m$. The expansion can written 
\begin{equation} 
\label{borcherdsproduct} 
\hat Z(\Omega)=\sum_{m=-1}\chi( {\rm Sym}^{m+1}(K3),q,y) p^m=\frac{1}{p} \prod_{{n>0,m\geq 0,l}\atop {l,m,n \in \mathbb{Z}}}\frac{1}{(1-p^n q^m y^l)^{c(m\cdot n,l)}}
\end{equation}
and summed to the total free energy $F(\lambda,Q)=\log(Z)=\sum_{j=1}^\infty \frac{1}{j} F^{(j)}$ by~\cite{KKV}    
\be
\label{FK3T2}
F^{(j)}=\left(2 \sin\frac{j\lambda}{2}\right)^{-2}\prod_{{k>0,n>0}\atop { m\geq 0,l \in \mathbb{Z}}} \frac{(1-q^k)^4}{(1-y^j q^k)^2 (1-y^{-j} q^k)^2} 
\frac{1}{(1-p^n q^m y^{lj})^{c(m\cdot n,l)}}\ ,
\ee 
where $c(m,r)$ are the expansion coefficients of the elliptic genus 
of K3 (\ref{ellipticgenusK3}) and we identified $\lambda=2 \pi z$ 
and defined $y=\exp(i\lambda)$.

Of course a key question for a full $N=2$ extension 
of the discussed  structure is whether the formulas 
(\ref{expansion},\ref{ansatz})  come from an underlying 
modular object as (\ref{reciprocalphi10},\ref{FK3T2}). One 
obvious difference is the quadratic versus linear growth 
of the index in the weak Jacobi forms. Nevertheless we 
like  to discuss in section \ref{projective} the formulation 
of projective special K\"ahler manifolds in $N=2$ 
supergravity that might give a hint for the construction of 
such an object. We hope to report on progress on that 
question in~\cite{HKKZ}.  

In this article we focus for simplicity on elliptic fibrations with 
only $I_1$ singular fibers in the Kodaira classification, i.e.\ no gauge 
symmetries,  but eventually flavor symmetries, in the $F$ theory 
interpretation. We note however that for the gauge theory case, i.e.\ 
elliptic fibrations with higher Kodaira fibre singularities, the 
work~\cite{Haghighat:2014vxa} gives already the crucial hint for 
the generalization of the numerator in (\ref{ansatz}). In this case, 
we will find additional Jacobi theta functions in the numerator, 
whose different elliptic arguments correspond to the additional  
K\"ahler classes of the resolution. This Jacobi theta function 
will introduce additional zeros in the numerator, which signal 
the gauge symmetry enhancements, and should be fixed by this 
structure. The formulas we find also suggest a possible 
refinement, which we shortly indicate in (\ref{refinedBPSinvariants}).

\section{The elliptic Calabi-Yau manifolds and their mirrors} 
\label{toricellipticcy}  

In this section we describe the toric construction of elliptically 
fibered Calabi-Yau spaces with only $I_1$ Kodaira fibers. Mirror symmetry is 
manifest in this formalism.  We collect the key data of our main 
example and describe shortly the general structure of the genus zero sector.    

\subsection{Construction of toric hypersurface Calabi-Yau spaces}  
Our construction of mirror pairs of Calabi-Yau 
n-folds as hypersurfaces in toric ambient spaces 
$\mathbb{P}^{n+1}_\Delta$ follows Batyrev's 
construction which relies on dual pairs of 
$n+1$ dimensional reflexive pairs of lattice 
polyhedra $(\Delta,\Delta^*)$. Let 
$(\Gamma,\Gamma^*)$ be dual lattices with pairing 
$\langle .,.\rangle$ and real completion 
$(\Gamma_{\mathbb{R}}, \Gamma^*_{\mathbb{R}})$.  
A lattice polyhedron $\Delta$ ($\Delta^*$)  is the convex hull of 
integer points in $\Gamma$ ($\Gamma^*$). The dual polyhedron $\Delta^*$ 
to $\Delta$ is defined by $\Delta^*=\{y\in \Gamma^*_{\mathbb{R}}| \langle y,x\rangle \ge -1, \forall\  x 
\in \Delta\}$.  $\Delta$ is reflexive if $\Delta^*$ is a lattice polyhedron. 
Note $(\Delta^*)^*=\Delta$ so if $\Delta$ is reflexive  $(\Delta,\Delta^*)$  
is a reflexive pair and each polyhedron contains the origin as its unique inner point.

The Calabi-Yau n-fold $M_n$ is given by the canonical hypersurface 
\be 
P=\sum_{\nu_i \in \hat \Delta\cap \Gamma} 
a_i^* \prod_{\nu^*_j\in \hat \Delta^*\cap \Gamma^*} y_j^{\langle \nu_i,\nu_j^*\rangle+1}=0
\label{P} 
\ee
in the toric ambient space $\mathbb{P}_\Delta$ with coordinate ring $y_j$, 
while the mirror  Calabi-Yau n-fold $W_n$ is given by the constraint $P^*=0$ in 
$\mathbb{P}_{\Delta^*}$ with coordinate ring $x_i$, where $P^*=0$ is defined 
analogously to (\ref{P}) with  $\Delta$ and $\Delta^*$ exchanged. 
The notation $\hat \Delta$ ($\hat \Delta^*$) means the polyhedra 
$\Delta$ ($\Delta^*$) with the integer points interior to codimension 
one faces omitted.

To give the ambient space $\mathbb{P}_\Delta$ a fibration structure\footnote{See exercise in \cite{Fulton} p. 49, 
where the statement is made in the language of the fans associated to $\Delta$.}     
such that the embedded Calabi-Yau $n$-fold defined as hypersurface 
has a fibration by a Calabi-Yau $m$-fold, we combine a base polyhedron 
$\Delta^{B*}$ and a reflexive fibre polyhedron $\Delta^{F*}$ into 
an $n+1$ dimensional polyhedron $\Delta^*$ 
as follows 
\begin{equation} 
  \footnotesize 
  \begin{array}{|ccc|ccc|} 
   \multicolumn{3}{c}{ \nu^*_i\in \Delta^*} &\multicolumn{3}{c}{ \nu_j\in \Delta}  \\ 
    &            &\nu_i^{F*} &             &\nu_j^{F} & \\
    & \Delta^{B*}_{n-m} & \vdots                &s_{ij}\Delta^{B}_{n-m}&\vdots & \\
    &            &\nu_i^{F*} &                   & \nu_j^{F}& \\
    & 0 \ldots 0      &                       & 0\dots 0              &                          & \\
    & \vdots     &\Delta^{*F}_{m+1}             &   \vdots          & \Delta^{F}_{m+1}            & \\
    & 0 \ldots 0      &                       & 0\ldots 0              &                          & \\
  \end{array} \, .
\label{polyhedrafrombaseandfibre}
\end{equation} 
If  $\Delta^{*F}$ and $\Delta^{*B}$ are reflexive then $(\Delta,\Delta^*)$, given by the 
complex hull of the indicated points, is a reflexive pair. Reflexivity of $\Delta^{*F}$ 
is required by the CY condition on the fibre. For $\Delta^{*B}_{n-m}$ it is not a 
necessary condition, see~\cite{Huang:2013yta} for more details on this construction. 
We defined $s_{ij}= \langle \nu_i^F,\nu_j^{F*} \rangle+1\in \mathbb{N}$ and 
scaled $\Delta^B \rightarrow s_{ij}\Delta^{B}$. Here we indicated the 
dimensions of some polyhedra by subscripts; elliptic fibrations correspond to $m=1$.

For $n=3$ and $m=1$ we get many examples by choosing any of the 
$16$ reflexive polyhedra in 2d as $\Delta^{*F}$  and $\Delta^{*B}$ 
respectively and specifying in addition $\nu_i^{F*}\in  \Delta^{*F}$ 
as well as twisting parameters~\cite{BKL}, which not indicated 
in (\ref{polyhedrafrombaseandfibre}).        

\subsection{The $X_{18}(1,1,1,6,9)$ 3-fold, an elliptic fibration over $\mathbb{P}^2$}
\label{X18}
Our main example is the smooth elliptic fibration over $\mathbb{P}^2$.
This case\footnote{It can be also written as the zero locus of a degree 
18 polynomial  in the weighted projective space $\mathbb{P}^4(1,1,1,6,9)$ called $X_{18}(1,1,1,6,9)$.} 
is a two parameter model discussed in~\cite{Hosono:1993}  and in greater detail in~\cite{Candelas:1994hw} 
and illustrates the general construction described above. 

Pick for the base $\mathbb{P}^2$, whose toric polyhedron is the convex hull 
of the points $\Delta^{*B}={\rm conv}(\{(1,0),(0,1),(-1,-1)\})$,
for the fibre polynomial $\Delta^{*F}={\rm conv}(\{(1,0),(0,1),(-2,-3)\})$  
and for $\nu_3^{*F}=(-2,-3)$. Then $\nu_3^{F}=(-1,-1)$ and $s_{33}=6$. 

We list the points which give rise to the coordinate ring of 
$\mathbb{P}=\mathbb{P}_{\Delta}/(\mathbb{Z}_{18}
\times \mathbb{Z}_{6})$, all points ${\nu}^*_i\in \hat \Delta^*$ and the 
two vectors of linear relations among them, which correspond to the Mori cone of 
$\mathbb{P}_{\Delta^*}$, as well the toric divisors $D_{x_i}=\{x_i=0\}$      
\begin{equation}  
 \begin{array}{cc|rrrr|rrrr|rr|} 
    \multicolumn{2}{c}{Div.} &\multicolumn{4}{c}{\nu_i}&\multicolumn{4}{c}{{\bar \nu}^*_i}     &l^{(E)}& l^{(B)} \\ 
    D_{x_0}  &&       0& 0& 0&0&          0&    0&    0&    0&         -6&   0    \\ 
    D_{x_1}=L   &&  12&-6&-1&-1&         1&    0&   -2&    -3&         0&   1    \\ 
    D_{x_2}=L   &&  -6&12&-1&-1&         0&    1&   -2&    -3&         0&   1     \\  
    D_{x_3}=L   &&  -6&-6&-1&-1&        -1&  -1&   -2&    -3&         0&   1     \\ 
    D_z=E       &&   0&0&-1&-1&          0&    0&   -2&    -3&         1&  -3     \\ 
    D_x=2H       &&   0&0&2& -1&          0&    0&    1&    0&          2&   0   \\
    D_y=3H       &&   0&0&-1 &1&          0&    0&    0&    1&          3&   0  \\  
  \end{array} \ .
  \label{F1case} 
\end{equation}
The classical topological data of the 3-fold $M$ are easily calculable from the 
toric construction. The  Euler is $\chi(M)=-540$, the two independent Hodge numbers are 
$h^{1,1}(M)=2$, $h^{2,1}(M)=272$, the classical triple intersection numbers are  
given by\footnote{In the notation of \cite{Hosono:1993} these intersections 
are encoded in the ring ${\cal R}=9 J_E^3+ 3 J_E^2 J_B + J_E J_B^2$. $H,L$ 
are the notations for the divisors used in~\cite{Candelas:1994hw}. }   
\begin{equation}
C^0_{111}=H^3=9,\quad C^0_{112}=H^2\cdot  L =3, \quad C^0_{122}= H\cdot L^2=1, \quad C^0_{222}= L^3=0  
\end{equation} 
where  $H$ and $L$ are the divisors  dual to the curves defined by the Mori 
vectors $l^{(E)}$ and $l^{(B)}$ and the K\"ahler classes $J_E$ and $J_B$.   
The intersection with the second Chern class $c_2$ of $M$ are
\begin{equation}
\int_M c_2 \wedge J_E= [c_2]\cdot H= 36, \quad \int_M c_2\wedge J_B=  [c_2]\cdot L=102   \ .
\label{Chern2}
\end{equation}    

The mirror manifold is given by the zero locus      
\begin{equation} 
\label{mirrorcurvep11169}  
P^*=x_0( z^6 (x_1^{18}+x_2^{18}+x_3^{18}- b (x_1 x_2 x_3)^6)- 2^\frac{1}{3}\sqrt{3} a z x_1x_2x_3 x_4 x_5+ x^3+y^2) =0 
\end{equation}
in the space $\mathbb{P}$. $z:=x_4,x:=x_5=$ and $y:=x_6$ are the conventional 
names of variables in the Weierstrass form of the elliptic fibre. 
In  $\mathbb{P}$ there are toric $\mathbb{C}^*$ actions on the coordinates 
$x_i$, $i=1,\ldots,6$, which can be used to eliminate all $a_i$, 
but the two complex structure variables $(a,b)$ of $W$. This is because two $\mathbb{C}^*$ actions  
\begin{equation} 
x_i\rightarrow \mu_r^{l^{(r)}_i} x_i,\quad {\rm with } \ \  \mu_r\in \mathbb{C}^* \ ,
\label{xrescaling}
\end{equation}
are divided out from the coordinate ring of $\mathbb{P}$. 
One can introduce manifestly $\mathbb{C}^*$ invariant combinations the $a_i$ as 
complex structure variables of $W$, namely  
\be 
z_i=(-1)^{l^{(i)}_0}\prod_{k=1} a_k^{l^{(i)}_k}, \quad i=1,\ldots, h_{21}(W_3)=h_{11}(M_3) \ . 
\ee
In the case at hand $z_1:=z_E=\frac{a_4 a_5^2 a_6^3}{a_0^6}$ corresponds 
to the elliptic fibre and $z_2:=z_B=\frac{a_1a_2a_3}{a_4^3}$ to the base class. 
Using the $\mathbb{C}^*$ actions on the period integrals $\Pi(z)=\int_{\gamma_3} \Omega$ 
with ($a= a_0$)
\be
\label{Omega} 
\Omega=\oint_{\gamma_\epsilon} \frac{a \mu}{P^*},
\ee 
given by a residuum integral around $P^*=0$ with the measure  
$\mu =\sum_{i=1}^5 (-1)^i w_i d \tilde x_1\wedge \ldots \widehat{d x_i}\ldots \wedge d \tilde x_5$,
one can derive two Picard-Fuchs (PF) differential equations~\cite{Hosono:1993}   
\begin{eqnarray}  \label{PF2.1}
\mathcal{L}_1 &=& \theta_1(\theta_1-3\theta_2) -12 z_1 (6\theta_1+1) (6\theta_1+5), \nonumber \\
\mathcal{L}_2 &=& \theta_2^3 +z_2 \prod_{i=0}^2 (3\theta_2-\theta_1+i)
\end{eqnarray}
determining the periods from $\mathcal{L}_i\Pi(z)=0$, $i=1,2$. Here $\theta_i =z_i \frac{\partial}{\partial z_i}$. 
The discriminants of the operators are 
\begin{eqnarray}
\Delta_1 &=& (1-432 z_1)^3 -27 z_2 (432 z_1)^3, \nonumber \\
\Delta_2 &=& 1+27 z_2
\end{eqnarray} 

The 3-point couplings can be computed from the PF operators~\cite{Hosono:1993} 
\begin{eqnarray} 
\label{3point2.3}
&& C_{111}  = \frac{9}{z_1^3 \Delta_1},   ~~~~~ 
C_{112}  =C_{121}  =C_{211}  = \frac{3\Delta_3 }{z_1^2z_2  \Delta_1},   \nonumber \\ 
&& C_{122}  = C_{212}  =C_{221}  = \frac{\Delta_3^2 }{z_1 z_2^2 \Delta_1},   ~~~
C_{222}  = \frac{ 9(\Delta_3^3 + (432z_1)^3 ) }{z_2^2 \Delta_1 \Delta_2 },   
\end{eqnarray}
where for convenience we can define the factor $\Delta_3 = 1-432z_1$. 

\subsection{Integral symplectic basis and genus zero topological string amplitudes}
\label{symplecticbasis}
  
The PF operators determine the 3-point couplings. Its solutions in a special basis 
determine the metric on the moduli space, which is a K\"ahler manifold 
$G_{i\bar\jmath}=\partial_{t_i}\partial_{t_{\bar \jmath }} K$ ,
whose K\"ahler potential $K$ is given by
\be 
e^{-K}=i \int_W \Omega \wedge \bar \Omega =  i \Pi^\dagger \eta \Pi\ ,
\ee
where ($k=h^{1,1}(W)+1$)
\begin{equation}
\label{eta}
\eta= \left(\begin{array}{cc} 
        0 & {\bf 1}_{k \times k}\\ 
        -{\bf 1}_{k \times k}  & 0 
        \end{array}\right)\ . 
\end{equation}
is the symplectic pairing on $H_3(W,\mathbb{Z})$  and $\Pi(z)=(X^I=\int_{A^I} \Omega,F_I=\int_{B_I}\Omega)^T$  
is the period vector w.r.t. to the corresponding symplectic basis $(A^I,B_I)$, $I=0,\ldots,k-1$ of  $H_3(W,\mathbb{Z})$.  
The two structures are related by special geometry, which implies the existence of a 
prepotential, the genus zero amplitude ${\cal F}^{(0)}$, with 
\be
\Pi=\left(\begin{array}{c} 
X^0  \\ 
X^i \\
{F_0 }\\
{F_i}
\end{array}\right)=X^0\left(\begin{array}{c} 
1   \\
t^i \\
2 {\cal F}^{(0)}- t^i  \partial_i {\cal F }^{(0)}\\
{\partial {\cal F}^{(0)}\over \partial t^i}\end{array}\right)=X^0\left(\begin{array}{c}
1\\
t^i\\
-{C^0_{ijk}\over 3!} t^i t^j t^k+\frac{\int c_2\wedge J_i}{24}   t^i-i{\chi \zeta(3)\over (2 \pi  )^3}
+f(q)\\
{C^0_{ijk}\over 2} t^i t^j+{n_{ij}} t^j+ \frac{\int c_2\wedge J_i}{24} +\partial_i f({\underline q})
\end{array}\right)\ .
\label{periodbasis3fold} 
\ee
Here  
\be
\label{prepot} 
{\cal F}^{(0)}=\left[-{C_{ijk} t^i t^j t^k\over 3!}+n_{ij} {t^i t^j \over 2}+ \frac{\int c_2\wedge J_i}{24}t^i-i{\chi \zeta(3)\over 2 (2 \pi)^3} +f({\underline q})\right]
\ee
and $X^I$ are homogeneous coordinates and  $t^i(z)$ are inhomogeneous coordinates    
\begin{equation} 
t^i(z)=\frac{X^i({\underline z})}{X^0({\underline z})}=\frac{1}{2 \pi i} \left(\log(z_i)+ \Sigma^i({\underline z})\right), \quad i=1,\ldots, h^{1,1}(W) \ ,    
\label{mirrormap} 
\end{equation}
which serve as a mirror map. In particular the complexified volumes of the curves in the Mori cone $t_i$ 
are the flat coordinates at large radius near $z_i=0$, $\forall i=1,\ldots ,h^{11}(M)$ ,
 where the third equal sign in (\ref{periodbasis3fold}) 
is valid, $X^0$  is the unique holomorphic period normalized to $\omega_0(z)=1+{\cal O}(z) $. The $C^{(0)}_{ijk}$ 
are the classical intersection numbers.  There is a freedom in choosing the $n_{ij}$, but because of the odd 
intersections and the requirement of integer monodromy around $z_i=0$ , they cannot be set to zero 
as they must be half integral.  Following~\cite{Candelas:1994hw} or the topological  description the quadratic 
terms in ${\cal F}^{(0)}$ in the appendix of~\cite{HKQ} we may take them to $n_{11}=\frac{9}{2}$ and $n_{12}=\frac{3}{2}$ 
for the $X_{18}(1,1,1,6,9)$ model. Like for all toric hypersurfaces in this example the 
Picard-Fuchs equations are of the generalized hyperelliptic type and all logarithmic solutions 
can be derived from 
\begin{equation} 
\omega_0({\underline{z}},{\underline{\rho}}) = \sum_{{\underline{n}} } c({\underline{n}},{\underline{\rho}}) z^{{\underline{n}}+{\underline{\rho}}}\ \ 
{\rm where} \ \  c({\underline{n}},{\underline{\rho}})=\frac{\prod_j \Gamma(\sum_\alpha l_{0j}^{(\alpha)} ( n_\alpha+\rho_\alpha)+1)}{
\prod_i \Gamma(\sum_{\alpha} l_i^{(\alpha)}( n_\alpha+\rho_\alpha)+1)}\ ,   
\label{solution} 
\end{equation}
with $X^0({\underline z})=\omega_0({\underline z},0)$, by taking derivatives with respect to $\rho_\alpha$, 
in particular 
\be 
X^i({\underline z})=\frac{\partial_{\rho_i}}{2 \pi i} \omega_0({\underline{z}},{\underline{\rho}})|_{{\underline \rho}={\underline 0}}\ , 
\label{logsol} 
\ee 
see the Appendix in~\cite{Hosono:1994ax} for the explicit structure of 
the higher derivatives.

\section{Involution symmetry and BCOV formalism}  
\label{involutionsymmetry} 

In this section we mainly discuss the involution symmetry, its 
realization on the topological string amplitudes and its 
consequences. In the course of the exposition we also review the 
relevant aspects of the BCOV formalism. Many of the actual 
calculations are quite technical and relegated to the 
appendices. We comment on the relation of the involution symmetry 
to the monodromies in subsection \ref{monodromiesversusinvolution}  
and on the local limit in subsection \ref{locallimit}. The section 
\ref{ridigspecialgeometry} illustrates the idea of the BCOV ring 
and contains many modular properties needed in section 
\ref{fibremodularity}. Subsection \ref{projective} is more 
speculative and tries to give a perspective on a possible
global modular object related to the all genus amplitude.

\subsection{The involution symmetry} 
Fibrations with fibre polytope $\Delta^{*F}$ have an involution 
symmetry $I$  acting on the moduli space, which is independent of the 
chosen base $B$.  The action of $I$ on the $X_{18}(1,1,1,6,9)$ moduli space
that we describe below appeared in \cite{Candelas:1994hw}.  

$I$ acts non-trivial on the higher genus amplitudes $F_g(S^{ij},S^i,S)$ and 
the anholomorphic propagators  $S^{ij},S^j,S$. The $F_g$ are inhomogeneous 
polynomials of degree $3g-3$ with weights $1,2,3$ and rational functions 
in the moduli.  This follows from the Feynman rules for the genus $g$ 
amplitudes in~\cite{BCOV}. The rational function $f_g$ that does not multiply 
any propagator is not determined by the  holomorphic anomaly equation, that 
otherwise fix the $F_g$ recursively in $g$. The function  $f_g$ is 
restricted by the pole behavior of the $F_g$ at the boundary divisors 
of the moduli space to contain only finite parameters for each $g$.  This 
 is called the holomorphic ambiguity and its determination is a main problem to solve 
the topological string on compact Calabi-Yau manifolds.  If we know 
the action of $I$ on $F_g$ and the propagators we can further restrict 
$f_g$. In fact this requirement cuts down the possible parameters in $f_g$ to 
roughly  one fourth. Moreover we find that the symmetry $I$ is equivalent 
to the  constraints that the fibre modularity imposes on $f_g$.                             

To see the involution symmetry explicitly 
define the monomial $m=(\prod_i x_i)^6$, where $x_i$ 
are the coordinates of the base. Then 
\begin{equation} 
P^*= g(z,\underline x,{\underline a}_B) - b z^6 m^6- 2^\frac{1}{3}\sqrt{3} a z m x y  + x^3+y^2 =0, 
\end{equation}
where $g(z,\underline x,{\underline a}_B)$ is a polynomial compatible with 
the scaling, which does not contain $m$. Now requiring that  
\begin{equation}
\begin{array}{rl}
x&\rightarrow x + c_1 z^2 m^2\\ 
y&\rightarrow y + c_2 x z m + c_3 z^3 m^3 \\
\label{transformation}
\end{array}
\end{equation}
leaves $P^*$ invariant fixes $c_1=2^\frac{1}{3} a^2$, 
$c_2=\frac{(1-i)}{2^\frac{1}{3}} \sqrt{3} a $  and $c_3= 
\sqrt{3} a^3$ and acts on the parameters as
\begin{equation} 
I:(a,b)\rightarrow  (i a,b+a^6) \ .  
\end{equation}
This involution operation acts on the $(3,0)$ form $\Omega$ given in (\ref{Omega}) by  
\begin{equation}
I:\Omega \rightarrow  i \Omega ,
\end{equation}
because $P^*$ as well as the measure $\mu$ are invariant.
Since $\Omega$ defines vacuum line bundle ${\cal L}$ and the higher 
genus  amplitudes ${\cal F}_g$ transforms as section ${\cal F}_g\in {\cal L}^{2g-2}$ 
we conclude that the involution symmetry maps ${\cal F}^{(g)}$ to $\tilde {\cal F}^{(g)}$ with   
\begin{eqnarray} 
\label{involutionFg} 
\tilde{\mathcal{F}}^{(g)} = (-1)^{g-1}  \mathcal{F}^{(g)} . 
\end{eqnarray} 
In the the $z_i$ coordinates the involution acts as 
\begin{eqnarray} \label{invo2.4}
I: ~~~ (z_1,z_2) \rightarrow (x_1,x_2) = \big{(} \frac{1}{432}-z_1, -\frac{(432z_1)^3 z_2}{(1-432z_1)^3} \big{)}.  
\end{eqnarray} 
The PF operators (\ref{PF2.1}) are invariant up to some trivial factors under the involution.  
On the other hand, the involution exchanges the two discriminants up to some factors
as 
\begin{equation} 
I(\Delta_1) = (432z_1)^3 \Delta_2, \qquad I(\Delta_2) = \frac{\Delta_1}{(1-432z_1)^3} \ .
\label{involutionexchange} 
\end{equation}

\subsection{Monodromy group versus involution symmetry}
\label{monodromiesversusinvolution} 
The involution symmetry multiplies the periods w.r.t.\ $\Omega$ with  $i$, so it 
is in particular not a symplectic transformation (in fact $M_I^T \eta M_I =-\eta$)   and 
cannot be related to an actual monodromy action.      

The   monodromy group is most quickly described as  follows. 
There are Neveu-Schwarz B-field  shifts that leave the instanton 
action invariant $t_i\rightarrow t_i+ 1$, $i=1,..,h^{1,1}(M)$. 
Their monodromy is  completely fixed by   (\ref{periodbasis3fold}) after 
specifying the topological data as in subsections \ref{X18},\ref{symplecticbasis}
for the main example.  We call these monodromies and the corresponding 
$6\times 6$ matrices  $T_1(=T_E)$ and $T_2(=T_B)$. 

Further there is  a cycle $\nu_1$, which corresponds to the $B_3\sim S^3$  
base of the Strominger-Yau-Zaslow $T^3\rightarrow B_3$  fibration, that 
vanishes at the  conifold $\Delta_1=0$. In the $X_6(1,2,3)$ 
elliptic fibrations  such as the $X_{18}(1,1,1,6,9)$ one has a 
second conifold discriminant $\Delta_2=0$, where $\Delta_1$ 
and  are exchanged, up to irrelevant factors, by the involution 
symmetry as discussed in the last section. Let us call the corresponding 
vanishing cycle $\nu_2$ and  use the integral symplectic basis of cycles that corresponds to the 
period vector  (\ref{periodbasis3fold}). Using some basic analytic 
continuation one calculates the vanishing cycles in this basis as~\footnote{In fact $\nu_1$ 
corresponds  always to $F_0$, while $\nu_2$  corresponds to 
$F_E-X_0-\sum_{i=1}^{h^{11}(B)} a^i F_{B,i}$, where  the $a_i$ 
defined above (\ref{redefinitionti}) .}   
\begin{equation}
\begin{array}{rl}
\nu_1&=(\phantom{-}0,0,0,1,0,\phantom{-}0)\ ,\\
\nu_2&=(-1,0,0,0,1,-3)\ .
\end{array}
\label{vanishingcycles}
\end{equation}  
The Lefschetz monodromy theorem for 3-folds states that the monodromy 
around a conifold divisor, where by definition an $S^3$ sphere $\nu$  
vanishes, on each cycle $\gamma\in H_3(M,\mathbb{Z})$  is given by 
the symplectic reflection 
\be        
 S_{\nu}(\gamma)=\gamma-\langle \gamma, \nu\rangle \nu \ ,  
\ee 
where $\langle, \rangle$ denotes the symplectic pairing. By 
this formula and (\ref{vanishingcycles}) we can calculate the 
monodromy in the basis  (\ref{periodbasis3fold}). We call the 
corresponding  monodromies  and $6 \times 6$ matrices 
$C_1$ and $C_2$ respectively. The $T_i$ and $C_i$ 
generate the monodromy group of the $X_{18}(1,1,1,6,9)$ 
model in fact redundantly, as can be seen by the Van Kampen
relations analyzed for this case in great detail in~\cite{Candelas:1994hw}. 
In particular one has that an order 18 element $A$ is given by\footnote{We use 
the same notation $A$ for this element as in~\cite{Candelas:1994hw}. The 
other notions ate related by $C_1=T$, $C_2=B$, $T_1=T_\infty$ and $T_2=D_\infty$.}    
\begin{equation} 
 A^{-1}=C_1 C_2 T_2              
\end{equation}
and the two conifold monodromies  $C_1$ and $C_2$ are conjugated to each other  
by the order 6 element $A^3$
\be                    
 C_1 =A^3 C_2 A^{-3}\ ,
\label{conjugation}
\ee
which exchanges  $\nu_1\leftrightarrow \nu_2$ and is explicitly
\be 
A^3=\left(
\begin{array}{cccccc}
1 & 1& 0& 0& 0& 0\\ 
-1& 0& 0& 0& 0& 0\\ 
3& 3& 1& 0& 0& -1\\ 
-1& 0& 0& 0& 1& -3\\ 
0& 10& 3& -1& 1& -3\\ 
0& 3& 1& 0& 0& 0
\end{array}\right)\ .
\ee 
One can conclude that the monodromy group is generated by $A$ and 
$C_1$~\cite{Candelas:1994hw}. Note that $T_1$  and $A^3$ acts on 
$t_1$ as~\cite{Candelas:1994hw}
\begin{equation}
T_1:t_1\mapsto t_1+1,\qquad A^3:t_1\mapsto  - \frac{1}{t_1+1}
\end{equation}
These operations  generate  an ${\rm Sl}(2, \mathbb{Z})$ action 
on the elliptic fibre parameter, which we call therefore often $\tau$ 
in the following.  As we mentioned already invariance under the 
 involution symmetry $I$   is equivalent to the fibre modularity. 
This could be guessed from the fact that it is  $A^3$  that really  
imposes the nontrivial part of fibre modularity --- the shift symmetry 
is present for all K\"ahler moduli --- and conjugates the conifold monodromies just as  
$I$ exchanges the two conifold discriminants. The equivalence of the restrictions 
imposed  the  by invariance under $I$ and the fibre modularity  will be 
strictly proven in appendix \ref{sec:FI}.     
 
\subsubsection{The local limit}
\label{locallimit} 
As we have seen there is an ${\rm SL}(2,\mathbb{Z})$ action on the 
elliptic fibre modulus. This symmetry governs the amplitudes order by 
order in the  exponential  $Q_B=\exp(2 \pi i t_B)$ of the flat coordinate 
representing the base curve in $\mathbb{P}^2$. There is an other action of the congruent 
subgroup $\Gamma_0(3)=\left\{\left(\begin{array}{cc} a& b\\ c& d 
\end{array} \right)\in  {\rm SL}(2,\mathbb{Z})\Bigl| c=0 \ {\rm mod}\ 3\right\} $   
of ${\rm SL}(2,\mathbb{Z})$ on the base modulus. To see it we must make 
the volume of the fibre large, $\lim(t_E) \rightarrow i \infty$, which 
corresponds to $z_E\sim q=\exp(2 \pi i t_E)\sim z_E\sim 0$. This is known as the 
local limit. In the $A$-model language it focuses on the 
${\cal O}(-3)\rightarrow \mathbb{P}^2$ geometry by decompactifying 
the elliptic fibre of $X_{18}(1,1,1,6,9)$. The periods of the local geometry 
are given by integrals of a meromorphic differential $\lambda$  over cycles on 
an  elliptic curve ${\cal C}_B$ that is the the mirror geometry. These periods  
$\Pi_{loc}=\left(\int_{i} \lambda, i=0,a,b,\right) =\left( 1, t_B, 
-3 \partial_{t_B} F_{loc}^{(0)}\right)$ fulfill local rigid special 
geometry~\cite{Huang:2011}.  
The occurrence of the $1$ indicates that the CY  $(3,0)$ form $\Omega$ 
becomes in the local limit the meromorphic one form $\lambda$ which has  
a non-vanishing residuum at a pole $ \lambda $ on ${\cal C}_B$. From 
(\ref{PF2.1}) one sees that the periods in the local limit $\Pi_{loc}$ are governed 
by a specialization of the second differential operator $\mathcal{L}= 
\theta_2^3 +z_2 \prod_{i=0}^2 (3\theta_2+i)$, which is another way 
to see the constant solution. 

However we  want to understand precisely, which periods $\Pi=\int_{\underline \Gamma}\Omega$ 
on the compact CY 3-fold become the local periods  in the limit or equivalently  
how the  $\Gamma_0(3)$  action is embedded  in the action of  ${\rm Sp}(6,\mathbb{Z})$. 
To see this notice that we have to make a linear change in the $\partial_{t_i} F^{(0)}$ , 
$i=1,2$ part  of  the basis $\Pi$ (\ref{periodbasis3fold}) to keep with 
\be
\begin{array}{rl}
\tilde F_1&=F_B=-\frac{3}{2} t_E^2-t_E t_B+\frac{3 t_E}{2}+\frac{3}{2} \\ 
\tilde F_2&=3 F_B-F_E=\frac{t_B^2}{2}-\frac{3 t_B}{2}+\frac{1}{4}\ 
\end{array}
\ee
in the limit $z_E=0$ three fine periods $\left( X_0=1,t_B, \tilde F_2\right)$. 
If we now conjugate the monodromies  that we found in the last section by 
the corresponding element $C=\left(\begin{array}{ccc} 
{\bf 1}_{4\times 4}&0&0\\
0                          &0&1\\
0                          &-1&3\\ \end{array}\right)$, we 
obtain the following  monodromies $M_{i}=C T_2 C^{-1}$, $M_{c}=C C_2 C^{-1}$ and 
$M_{o}=(M_{c} M_{i})^{-1}$   
{\footnotesize{
\be
M_{i}=\left (\begin{array}{cccccc}
{\underline 1}& 0& 0& 0& 0& 0\\ 
0& 1& 0& 0& 0& 0\\ 
{\underline 1}& 0& {\bf 1}& 0& 0& {\bf 0}\\ 
3& 2& 0&1& -1& 0\\ 
0& -1& 0& 0& 1& 0\\ 
-{\underline 1}& 0& {\bf 1}& 0& 0& {\bf 1}
\end{array}\right), 
M_{c}=\left (\begin{array}{cccccc}
{\underline 1}& 0& 0& 0& 0& 0\\ 
1& 1& 0& 0& 0& 1\\ 
-{\underline 3}& 0& {\bf 1}& 0& 0& {\bf -3}\\ 
1& 0& 0& 1& 0& 1\\ 
0& 0& 0& 0& 1& 0\\ 
{\underline 0}& 0& {\bf 0}& 0& 0& {\bf 1}
\end{array}\right),  
M_{o}=\left (\begin{array}{cccccc}
{\underline 1}& 0& 0& 0& 0& 0\\ 
-1& 1& 0& 0& 0& -1\\ 
{\underline 2}& 0& {\bf 1}& 0& 0&  {\bf 3}\\ 
-3& -1& 0& 1& 1& 0\\ 
-1& 1& 0& 0& 1& -1\\ 
-{\underline 1}& 0& {\bf -1}& 0& 0& {\bf -2}
\end{array}\right) \ . \ee}} 
Here we have printed the elements $\left(\begin{array}{cc} d&c\\ a& b\end{array}\right)$ 
of the $\Gamma_0(3)$ subgroup, see~\cite{Aganagic:2006wq}, in bold face and underlined 
the shifts due to the non-vanishing residua of $\lambda$. A consequence of this symmetry is the fact 
that in the strict local limit the amplitudes $F^{(g)}(\hat E_2(\tau_B), 
G_2(\tau_B), G_4(\tau_B),G_6(\tau_B))$~\cite{Aganagic:2006wq}, with 
$\tau_B= \frac{-3\partial_u\partial_{t_B} F^{(0)}}{\partial_u t_B}$, can 
be expressed in the indicated generators of the ring of almost holomorphic 
ring of  forms  of $\Gamma_0(3)$ and has been completely solved~\cite{Haghighat:2008gw}.
Here $u$ is the complex modulus of the local geometry~\cite{Haghighat:2008gw}, 
which is  identified with the $b$ parameter in (\ref{mirrorcurvep11169}).  Of course this 
base modularity that holds in the strict large fibre limit, and acts not directly on the 
periods of the 3-fold geometry, extends much less trivially to the whole two 
parameter family than the fibre modularity, that we discuss further in section~\ref{fibremodularity}.

\subsection{Involution symmetry at genus one}
\label{involutionatgenusone}
First we consider the genus one amplitude $\mathcal{F}^{(1)}$  which fulfills the
holomorphic anomaly equation~\cite{BCOV}
\be 
\partial_i \bar \partial_{\bar \jmath} F^{(1)}=\frac{1}{2} {\rm Tr} C_i \bar C_{\bar \jmath} - \frac{\chi(M)}{24} G_{i\bar \jmath},        
\ee
where $\langle j|\bar \phi_{\bar \imath} |k\rangle =\bar C_{\bar \imath\bar \jmath\bar k}=C_{\bar \imath}^{j'k'} \eta_{j'j} \eta_{k k'}$ and  
$\langle j | \phi_i| k\rangle=C_{ijk}$ are the 3 point functions discussed above, 
with the relation $C_{\bar \imath}^{jk}=e^{2K} G^{j \bar \jmath } G^{k \bar k} \bar C_{\bar \imath \bar \jmath \bar k}$.  The solution has the  well known form
\begin{eqnarray} \label{genusoneamplitude}
\mathcal{F}^{(1)} = \frac{1}{2}(3+h^{1,1} -\frac{\chi}{12}) K + \frac{1}{2} \log\det G^{-1} - \frac{1}{24}\sum_{i=1}^2 s_i \log z_i -\frac{1}{12} \sum_{a=1}^2 \log \Delta_a.  
\end{eqnarray} 
The topological data for the $X_{18}(1,1,1,6,9)$ model are given in section \ref{X18}.  The leading asymptotic of $ \mathcal{F}^{(1)}$ anholomorphic
near large volume limit is 
\be 
\label{F1}
\mathcal{F}^{(1)}=-\frac{1}{24} \sum_{i=1}^2 t_i \int_M c_2 J_i +{\cal O}(Q) \ .
\ee
The values of the second Chern-Classes
 (\ref{Chern2})  determine the constants $s_1= 114, s_2=48$.  In the holomorphic limit, the determinant of the 
Kahler metric transforms as $\det(\tilde{G}) =\det(\frac{\partial z_i}{\partial x_j}) \det(G)$.   A simple calculation 
shows that the invariance of the genus one amplitude $\mathcal{F}^{(1)}$ under the involution 
transformation (\ref{invo2.4}) impose the constrain on the constants $s_1=-30+3s_2$, which is 
consistent with the evaluation of the second Chern class on $J_i$. 

\subsection{Involution symmetry at higher genus}
The higher genus amplitudes are defined recursively from the holomorphic anomaly 
equations~\cite{BCOV}
\begin{equation} 
\label{generalholomorphicanomaly}
\partial_{\bar \imath} F^{(g)}=\frac{1}{2} \bar C^{kl}_{\bar \imath}\left(D_k D_l F^{(g-1)}+\sum_{r=1}^{g-1} D_k F^{(r)} D_l F^{(g-r)}\right)\ .           
\end{equation}
The $F^{(g)}$  can be integrated using an anholomorphic potential $S$ for $\bar C^{jk}_{\bar \imath}$, 
whose existence is  consequence of special geometry or more generally the $tt^*$ geometry~\cite{BCOV}. 
One has 
\be 
\partial_{\bar \imath} S=G_{\bar \imath j}S^j, \quad \partial_{\bar \imath} S^j=
G_{\bar \imath k} S^{jk},\qquad \partial_{\bar \imath}  S^{ij}= \bar C^{ij}_{\bar \imath}\ . 
\ee                
The $F^{(g)}$ were  obtained in~\cite{BCOV} either using the special geometry commutator 
$[D_i,\partial_{\bar \jmath}]_k^l=G_{k\bar \jmath}\delta_i^l+G_{i\bar \jmath}\delta_k^l 
-C_{ikn}\bar C_{\bar \jmath}^{nl}$ and partial integration or a Feynman graph expansion 
of a master functional, which fulfills a heat equation type of equation. The most vantage 
point of view is to find directly  a solution for this master integral, which we achieve at 
least iteratively in the base modulus.  

\subsubsection{The propagators and rigid special geometry}  
\label{ridigspecialgeometry}

It was shown in~\cite{Yamaguchi} as a consequence of special geometry 
that the covariant derivatives $D_i$ in (\ref{generalholomorphicanomaly}) close 
on anholomorphic generators, that are closely related to the propagators, up to 
rational functions of ${\underline z}$. We call the corresponding ring generated 
by these  anholomorphic generators over the rational functions 
in ${\underline z}$ the BCOV ring.    

This principal structure of the BCOV ring can be understood by the analogy 
to the ring of almost holomorphic forms $\bigoplus_{k}{\widehat {\cal M}}_{k}(\Gamma_0)=
\mathbb{C}[\hat E_2, E_4,E_6]$ of $\Gamma_0={\rm SL}(2,\mathbb{Z})$ 
(or congruent subgroups thereof) that is graded by the weight $k$.
Let us denote the modular transformation $\tau\mapsto \tau_\gamma=\frac{a \tau+ b}{c\tau +d}$ 
with $\left(\begin{array}{cc} a & b\\ c& d\end{array}\right)\in {\rm SL}(2,\mathbb{Z})$. 
For $k\in 2\, \mathbb{N}_+$  the normalized Eisenstein series $E_k$  of modular weight are 
defined as~\cite{Zagier}
\be 
E_k(\tau)=\frac{1}{2\zeta(k)}\sum_{m,n\in \mathbb{Z}\atop (m,n)\neq (0,0)}\frac{1}{(m \tau + n)^k}=
1+\frac{(2 \pi i )^k}{(k-1)! \zeta(k)}\sum_{n=1}^\infty \sigma_{k-1}(n) q^n\ ,     
\label{Eisenstein}
\ee
where the last equal sign holds straightforwardly only for $k>2$. 
Here $\sigma_k(n)$ is the sum of the $k$-th power of the positive 
divisors of $n$ and $\zeta(k)=\sum_{r\geq 0}\frac{1}{r^k}$ with    
\be 
\zeta(2k)=-\frac{(2 \pi i)^{2k}B_{2k}}{4k(2k-1)!}\ .
\label{evenzeta}
\ee
The Bernoulli numbers $B_k$  $k\neq 1$ are defined by the 
generating function 
\be 
\sum_{m=0}^\infty \frac{B_m x^m}{m!}:=\frac{x}{e^x-1}\ . 
\label{Bernoulli}
\ee
If the sum in (\ref{Eisenstein}) converges $k>2$ it 
is obvious that $E_{k}(\tau_\gamma)=(c \tau + d)^{k} E_k(\tau)$
transforms as a weight $k$ form under modular transformations. 
For $k=2$ the second equal sign in (\ref{Eisenstein}) can be viewed 
as a regularization of the sum on the right~\cite{Zagier}. In this case the 
modular transformation is broken to 
\be 
E_2(\tau_\gamma)=( c \tau+ d)^{2}E_2(\tau)-\frac{6 i}{\pi} c(c\tau+d)\ .  
\label{quasimodularshift} 
\ee
Since $\frac{1}{{\rm Im}(\tau_\gamma)}=\frac{(c\tau+ d)^2}{{\rm Im}(\tau)}- 2 i c (c\tau+d) =\frac{|c\tau+ d|^2}{{\rm Im}(\tau)}$ 
one can define an almost holomorphic object
\be 
\hat E_2(\tau)=E_2(\tau)-\frac{3}{\pi {\rm Im}(\tau)},
\label{almostholomorphic}  
\ee
which transforms as a modular form of weight $2$.     

In the solution of the holomorphic anomaly equation $\hat E_2(\tau)$  
plays the r\^ole of the propagator  and the Maass derivative 
$D_\tau:\hat {\cal M}_{k}\rightarrow \hat  {\cal M}_{k+2}$
\be 
D_{\tau} f_k =\left(d_\tau-\frac{k}{4 \pi {\rm Im}(\tau)}\right)f_k     
\ee
on weight $k$ forms, with $d_\tau=\frac{1}{2 \pi i}\frac{\rm d}{{\rm d}\tau}$, 
plays the role of the covariant derivative $D_i$. The an-holomorphic 
version of the Ramanujan identities 
\be
\label{ramanujan} 
D_{\tau} \hat E_2=\frac{1}{12}(\hat E^2-E_4),\quad 
D_{\tau} \hat E_4=\frac{1}{3}(\hat E_2 E_4-E_6),\quad
D_{\tau} \hat E_6=\frac{1}{2}(\hat E_2 E_6-E_4^2),
\ee 
show that it closes on $\bigoplus_{k}{\widehat {\cal M}}_{k}(\Gamma_0)$. 
Note that by the isomorphy  (\cite{KZ} Prop. 1)  
of the ring $\bigoplus_{k}{\widehat {\cal M}}_{k}(\Gamma_0)$  and the one quasi 
modular forms $\bigoplus_{k}{\tilde {\cal M}}_{k}(\Gamma_0)= \mathbb{C}[E_2, E_4,E_6]$ 
one can take the holomorphic limit in (\ref{ramanujan}) by replacing 
$\hat E_2\rightarrow E_2$ and $ D_{\tau}\rightarrow  d_\tau$ without 
losing information about the ring structure. 

The relation between the BCOV ring of the  $X_{18}(1,1,1,6,9)$ Calabi-Yau manifold 
and the ring of almost holomorphic forms is more then an analogy.  In the limit of large 
fibre~\ref{locallimit} the special K\"ahler structure of supergravity 
reduces to local rigid special geometry and the BCOV ring of $X_{18}(1,1,1,6,9)$ 
reduces to $\bigoplus_{k}{\widehat {\cal M}}_{k}(\Gamma_0(3))=
\mathbb{C}[\hat E_2, G_2,G_4,G_6]$~\cite{Aganagic:2006wq}. 

In the limit of large base the relevant ring becomes $\bigoplus_{k}{\widehat {\cal M}}_{k}(\Gamma_0)$
and the structure of the holomorphic anomaly equation can be be even 
more neatly combined into the ring of weak Jacobi forms, which include all 
powers of the topological string coupling as we will discuss in 
section~\ref{fibremodularity}. 

\subsubsection{Projective special K\"ahler manifolds}
\label{projective} 

The form (\ref{FK3T2}) suggest that the all  genus 
amplitude in $K3\times T2$ is related to a modular form 
associated to a degenerate Riemann surface of genus $3$. 
In this section we recall some structure of projective special 
K\"ahler manifolds, see~\cite{FVP} for a modern review, that relates 
for N=2 compactifications of type II string the complex structure 
moduli of Calabi-Yau 3 folds to a family of Riemann surfaces of genus $h_{21}+1$. 

Of course we could first address the simpler question, how 
the objects  of rigid special K\"ahler geometry described in 
the last section extend to the moduli space of Riemann surfaces of higher genus. 
However this is done in~\cite{KPSW}, with the result that a 
holomorphic analog of $E_2$ does not exist, but a meromorphic one  
does.       

In the formalism of projective special K\"ahler manifolds the 
formal description of the formulae for the propagators and the closing of the 
covariant derivatives on the propagators, i.e the analogs of the 
relations~(\ref{ramanujan},\ref{propa2.6}), as well as the recursive integration w.r.t. to the anholomorphic generators 
all simplify relative to the formalism~\cite{Yamaguchi} and resemble closely the 
formalism in the rigid case~\cite{GKMW}. 

However the price one has to pay is very complicated anholomorphic 
embedding of the complex moduli space of the Calabi-Yau manifold including 
the string coupling $\lambda$, that is related to the projective scaling, into the 
Siegel upper half plane $\mathbb{H}_{h_{21}+1}$ of Riemann surfaces of 
genus $h_{21}+1$, which is given by~\cite{de Wit:1984pk}  
\be
\label{sugramap} 
{\cal N}_{IJ}=\bar F_{IJ} + 2 i \frac{ {\rm Im}(F_{IL}) X^L  {\rm Im}(F_{JK}) X^K}{{\rm Im}(F_{KL}) X^K X^L}\ , 
\ee
see \cite{Louis:1996ya} for a concise review~\footnote{We would like to thank Jan Louis for a guidance to the literature.}. 
Here upper case letter refer to the projective coordinates of the complex 
moduli space given by a set of  $A$-cycle periods $X^I({\underline{z}})$, $I=0,\ldots, h_{21}(W)$ 
in the integral symplectic basis discussed in section~\ref{symplecticbasis}.  
$F_{IJ}({\underline{z}})=\frac{\partial^2 F}{\partial X^I \partial X^J}$ with $F({\underline{z}})=(X^0)^2 {\cal F}^{(0)}$ 
the prepotential in homogeneous coordinates. The positivity, or rather negativity 
in the supergravity conventions ${\rm Im} ({\cal N}_{IJ})<0$, is discussed 
in~\cite{VanProeyen:1995sw}. It is physically enforced by the condition that 
kinetic terms in vector multiplets and the gauge kinetic terms 
${\cal L}=-\frac{1}{4}g^{-2}_{IJ}F^I_{nm}F^{J\ {nm}}-\frac{\theta_{IJ}}{32 \pi^2}F^I\tilde F^J+\ldots$ 
of the graviphoton $F^0_{nm}$ and the vector bosons $F^i_{nm}$, $i=1,\ldots, h_{21}$ in vector 
multiplets with  $g^{-2}_{IJ}=\frac{i}{4}( {\cal N}_{IJ} -\bar {\cal N}_{IJ})$ 
and $\theta_{IJ}=2 \pi^2 ( {\cal N}_{IJ} +\bar {\cal N}_{IJ})$  must be 
positive definite.  

Note that ${\cal N}_{IJ}$ transforms under real or integer 
symplectic transformations, w.r.t. to  (\ref{eta})
\be 
\left(\begin{array}{c} X'^A\\ F'_A\end{array}\right)= \left(\begin{array}{cc} D&B\\ B&A\end{array}\right) \left(\begin{array}{c} X^A\\ F_A\end{array}\right), 
\ee
with $A^TC=C^TA$, $B^TD=D^TB$ and $A^TD-C^TB=1$, such as the monodromies discussed 
in section \ref{monodromiesversusinvolution}, in the canonical way
\be
{\cal N}'=(A{\cal N}+B) (C{\cal N}+D)^{-1} \ .
\ee
This suggest that the special monodromy family of Riemann surface of genus 
$3$ and a better understanding of the map (\ref{sugramap}) might be 
essential to understand whether the formulae (\ref{expansion},\ref{ansatz})  
come from an underlying modular object as in the $N=4$ case.

\subsubsection{Involution symmetry on the propagators}  
\label{involutiononthepropagators}

For the concrete calculation we follow~\cite{Yamaguchi} and work in the in-homogenous 
coordinates and use the BCOV propagators with a shift by derivatives of the 
K\"ahler potential as the generators~\cite{Alim:2007}. We denote the propagators 
after this shift (\ref{shift}) as $S^{ij},S^i,S$. The definition of the 
propagators involves a choice of rational functions in the moduli $\underline{z}$~\cite{BCOV,Klemm:2004km}, 
which is not unique. It has to be accompanied by a consistent choice of rational functions 
of $\underline{z}$ in the closing relations~(\ref{propa2.6})~\cite{Alim:2007}  that generalize~(\ref{ramanujan}). 
The derivation these rational functions  for non-hypergeometric, but more interesting Apery 
like Picard-Fuchs equations that arise in non-abelian gauged linear $\sigma$ models 
requires a minimal choice that uses the properties of differential equations~\cite{JKM}. For the 
$X_{18}(1,1,1,6,9)$ model, we summarize two canonical choices  and the derivation 
of the action under the involution symmetry in the appendix~\ref{propagators}.

Let us now state the action of the involution symmetry on the propagators. Suppose 
$f(z_1,z_2)$ is a rational function of $z_1, z_2$, we denote with the tilde symbol $\tilde{f} = 
f(x_1,x_2)$, the transformed function by replacing the $z_i$'s in the arguments with $x_i$'s in the 
involution (\ref{invo2.4}). The $\tilde{f}$ can be considered as a function of 
either $x_i$'s or $z_i$'s using the transformation (\ref{invo2.4}). For example under  
the involution, the three point functions (\ref{3point2.3}) transform as a tensor except for a 
minus  sign due to (\ref{involutionFg})
\begin{eqnarray} \label{3ptrans2.5}
\tilde{C}_{i j k} = - \frac{\partial z_l} {\partial x_i} \frac{\partial z_m}  {\partial x_j} \frac{\partial z_n}  {\partial x_k}C_{lmn}.      
\end{eqnarray} 

We find in~\ref{propagators} that the action of the involution on the propagators has the 
following form  
\begin{eqnarray} \label{transprop2.9}
\tilde{S}^{ij} = - \frac{\partial x_i} {\partial z_k}  \frac{\partial x_j} {\partial z_m} S^{km}, ~~~ 
\tilde{S}^{i} = - \frac{\partial x_i} {\partial z_k}  S^{k} +f^i, ~~~ 
\tilde{S} = -  S +f^0,
\end{eqnarray} 
The precise form of the rational functions $f_i,f_0$ depends on the choice of the 
rational functions in the choice of the propagator as well as in~(\ref{propa2.6}).  
Fixing this choice as in (\ref{holoambi2.7})~\cite{Alim:2007} we find the solution
\begin{eqnarray} 
\label{shift2.11}
f^1 &=&  \frac{5 z_1 z_2(1 - 1296 z_1  + 559872 z_1^2 )}{12 (1 - 432 z_1)^2}, \nonumber \\
f^2 &=& \frac{3732480 z_1^3 z_2 [1 - 1296 z_1 + 559872 z_1^2 -  80621568 z_1^3 (1 + 27 z_2) ]}{(1 - 432 z_1)^6}\nonumber \\ 
f^0 &=& 5 [45 z_2 -67 + 1296 z_1 (67 - 135 z_2) -  559872 z_1^2 (67 - 315 z_2) \nonumber \\
&& + 80621568 z_1^3 (67 - 630 z_2)]/(23328 (1 - 432 z_1)^3)\ . 
\end{eqnarray}

\subsubsection{Involution symmetry on the higher genus amplitudes}

The topological string amplitudes transform under the involution symmetry (\ref{invo2.4}) 
by replacing the propagators with the transformed propagators (\ref{transprop2.9} and the 
complex structure coordinates $z_i$'s with $x_i$'s.  We would like to understand how the 
transformed amplitude is related to the original amplitude and how many constraints 
the symmetry implies.  

At higher genus $g\geq 2$, the BCOV holomorphic anomaly equation \cite{BCOV} 
can be written as partial derivatives with respect to the an-holomorphic propagators. 
Assuming the algebraic independence of $K_i$'s, we can write the partial derivatives in 
terms of lower genus amplitudes. Here the $K_i$ appear as in the covariant derivative 
of higher genus amplitude $D_i \mathcal{F}^{(g)} = \partial_ i \mathcal{F}^{(g)}  -(2g-2) K_i \mathcal{F}^{(g)}$. 
We need to be a little careful for the genus one case, as the $K_i$ already appear at the ordinary derivative
\begin{eqnarray} \label{F1i2.13}
\partial_i \mathcal{F}^{(1)} = \frac{1}{2} C_{ijk} S^{jk} - (\frac{\chi}{24}-1) K_i - \frac{1}{2} s^j_{ij} - 
\partial_i ( \frac{1}{24} \sum_{k=1}^2 s_k \log z_k +\frac{1}{12} \sum_{a=1}^2 \log \Delta_a), 
\end{eqnarray} 
where we have computed the derivative of Kahler metric in terms of Christoffel connections and 
use the first equation in (\ref{propa2.6}). The equations for partial derivative  are 
\begin{eqnarray} \label{partial2.14}
\frac{\partial \mathcal{F}^{(g)} }{\partial S^{ij}} & =& \frac{1}{2} \partial_i(\partial _j^{\prime} \mathcal{F} ^{(g-1)}  ) 
+\frac{1}{2} (C_{ijl }S^{lk} - s^k_{ij} )\partial_k^{\prime}  \mathcal{F} ^{(g-1)}   
+ \frac{1}{2}  (C_{ijk}S^k - h_{ij} ) c_{g-1}    \nonumber \\
&& + \frac{1}{2} \sum_{h=1}^{g-1} \partial_i^{\prime}  \mathcal{F}^{(h)}   \partial_j^{\prime}  \mathcal{F}^{(g- h)},  \nonumber \\
\frac{\partial \mathcal{F}^{(g)} }{\partial S^{i}} & =& (2g-3) \partial_i^{\prime}  \mathcal{F}^{(g-1)} 
+ \sum_{h=1}^{g-1} c_h   \partial_i ^{\prime} \mathcal{F}^{(g- h)}, \nonumber \\
\frac{\partial \mathcal{F}^{(g)} }{\partial S } & =& (2g-3) c_{g-1} + \sum_{h=1}^{g-1} c_h c_{g-h}, 
\end{eqnarray}
where the $c_{g}$ is defined as
\begin{eqnarray}  \label{cg3.36}
c_g &=& \left\{
\begin{array}{cl}
 \frac{\chi}{24} -1,    &   g=1 ;   \\
 (2g-2) \mathcal{F}^{(g)},             &   g>1 .   
\end{array}    
\right.
\end{eqnarray}
We have also used the notation $\partial^{\prime} $ to denote 
\begin{eqnarray}  \label{F1extra3.37}
\partial_i^{\prime}  \mathcal{F}^{(g)}  &=& \left\{
\begin{array}{cl}
 \partial_i \mathcal{F}^{(g)}  +  (\frac{\chi}{24}-1) K_i  ,    &   g=1 ;   \\
 \partial_i \mathcal{F}^{(g)} ,             &   g>1 ,
\end{array}    
\right.
\end{eqnarray}
i.e. on the right hand in (\ref{partial2.14}), we use the formula (\ref{F1i2.13}) for 
$\partial_i \mathcal{F}^{(1)}$ omitting the $ - (\frac{\chi}{24}-1) K_i $ term.

Since the propagators are symmetric $S^{ij} = S^{ji} $, we can choose to use only the  $S^{ij}$ with $i\leq j$. 
In the case of $i\neq j$, the right hand side of first equation in (\ref{partial2.14}) need to be 
multiplied by an extra factor of 2 to take account of the double contribution.

Let us make some remarks about (\ref{involutionFg}). When 
we combine the amplitude with its involution transformation  
as in (\ref{involutionFg}), we find the dependence on the 
propagators $S^{ij}, S^i, S$ cancels with our proposal for the 
transformed propagators (\ref{transprop2.9}) and the shifts 
(\ref{shift2.11}), which is a 
consistency check of the formulae (\ref{transprop2.9}), (\ref{shift2.11}).

To make an ansatz for the holomorphic ambiguity $f^{(g)}$ at genus $g$, one considers a 
rational function of $z_1, z_2$ with a pole of $(\Delta_1 \Delta_2)^{2g-2}$. 
When we combine  the correct holomorphic ambiguity with its transformation $\tilde{f}^{(g)}+(-1)^gf^{(g)}$ 
as in (\ref{involutionFg}), we find  the degree of the pole  at discriminant $\Delta_1 \Delta_2$ is 
reduced for $g>2$, and the pole completely cancels for the genus $g=2$ case. Meanwhile, a 
new pole at  $\Delta_3 = 1-432z_1$ appears due to the involution transformation 
(\ref{invo2.4})~\footnote{The genus two and three holomorphic ambiguity formulae in 
\cite{Alim:2012} miss a constant $\frac{511}{144}$ and $-\frac{2105053}{1959552}$.}.

The use of the involution symmetry greatly reduce the number of unknown constants in the 
holomorphic ambiguity. In the naive ansatz in~\cite{Alim:2012}, the holomorphic ambiguity 
has the form $ \frac{p^{(g)}(z_1,z_2)}{(\Delta_1\Delta_2)^{2g-2}}$, where $p^{(g)}(z_1,z_2)$ 
is a polynomial of degree $7(g-1)$ in $z_1$ and degree $5(g-1)$  in $z_2$. So naively there 
are $(7g-6)(5g-4)$ unknown constants. For example, at genus 2 and 3, the naive ansatz has 
$48$ and $165$ unknown constants. We find the equation (\ref{involutionFg}) imposes $36$ and $125$  
conditions on the naive ansatz, so we are left with only $12$ and $40$  unknown constants to be 
fixed by boundary conditions. We present an detailed analysis on the reduction of unknowns
in appendix \ref{effectiveboundaries}. Together  with the conifold gap discussed in 
Appendix~\ref{conifoldgap}, this allows to fix the ambiguity to genus 9. 

\section{Fiber modularity}
\label{fibremodularity}

In this section we reformulate the holomorphic anomaly equation 
expressing the fibre modularity~\cite{Klemm:2012} in a convenient 
way that allows us to summarize  the structure in terms of even weak 
Jacobi forms. We exemplify the geometry setting with the elliptic 
fibration over $\mathbb{F}_1$, which exhibits geometric limits 
that yields the BPS numbers the $K3$ fibre, the BPS numbers on an 
embedded $\frac{1}{2} K3$ and maybe most remarkably by blowing 
down the $\frac{1}{2} K3$, the compact elliptic fibration over 
$\mathbb{P}^2$. The  even weak Jacobi forms are hence a unified 
mathematical language to address heterotic type duality~\cite{Marino:1998pg}, 
the six dimensional superconformal theories~\cite{Ganor:1996gu}\cite{Klemm:1996hh}  
and compact elliptic Calabi-Yau spaces.                      

\subsection{The modular anomaly equation}  

In the simplest case the base $B_{n-1}$ of the elliptic fibration is Fano 
(or almost Fano)~\footnote{To construct already many concrete examples we may 
assume that $B_{n-1}$ is a toric variety $\mathbb{P}_{\Delta^{*B}_{n-1}}$, i.e. 
specified by an $n-1$ dimensional reflexive polyhedron $\Delta^{*B}_{n-1}$.}, 
the degenerate elliptic fibers are only of Kodaira type $I_1$ and the 
elliptic fibre has only one holomorphic section and no rational sections. 
In the language of $F$-theory  the complex structure is so generic 
that the 12-2n dimensional theory is completely higgsed. 

The cohomology and the C.T.C Wall data of $M_3$ is then fixed by 
the cohomology of the base and the fibre type~\cite{Klemm:2012}. 
One key feature of topological string amplitudes that comes from the 
described  geometrical  setting is the $\Gamma_F\in {\rm SL}(2,\mathbb{Z})$  
modularity of the amplitudes. Here $\Gamma_F$  is the modular group 
acting projectively on the fibre modulus. In particular for one 
holomorphic section $\Gamma_F$ is the full\footnote{It gets restricted  to
congruence subgroup  $\Gamma_F \in {\rm SL}(2,\mathbb{Z})$ if the fibration 
has more rational sections or  multi sections.} ${\rm SL}(2,\mathbb{Z})$. 
In this case the $b_{2}(M_3)$ K\"ahler classes of $M_3$ split in 
$b_{2}(B_{n-1})$ and one class, which corresponds to the 
elliptic fibre. Let $\tilde t_i$ be the complexified 
volumes of the curve classes $[C]_i$, $i=1,\ldots, b_{2}(B_2)$ 
in the base and $\tau$ the complexified volume of the elliptic fibre. 
One defines with $q=\exp(2 \pi i \tau)$ a parameter, which becomes 
exponentially small, when the volume ${\rm Im}(\tau)$ of 
the elliptic fibre becomes large and  similarly $\tilde Q^\beta \equiv 
\prod_{\alpha=1}^{h_2(B_n)} \exp( 2 \pi i \tilde t^\alpha \beta_\alpha)$, 
with $\beta\in H_{2}(B_n,\mathbb{Z})$. Let $a^i$ are the intersections 
of the curves $[C]_i$ with the canonical class $K_B$ of the base.  
E.g.\ for $B_2$ a $\mathbb{P}^2$ one has $a^1=3$, while for $B_2$ the 
Hirzebruch surfaces $\mathbb{F}_n$ one has $(a^1,a^2)=(2,2-n)$ etc, 
see~\cite{Klemm:2012} for the other toric bases. In order to 
make the ${\rm SL}(2,\mathbb{Z})$ duality  in the fibre direction  
more manifest it is convenient to redefine  
\be
t^i=\tilde t^i+ \frac{\tau a^i}{2}\ , 
\label{redefinitionti}
\ee
so that the base moduli are summarized in 
$Q^\beta \equiv \prod_{\alpha=1}^{h_2(B_n)} 
\exp( 2 \pi i t^\alpha \beta_\alpha)$. 

We can write an expansion of the disconnected 
topological string amplitudes also known as the free 
energy in the large volume limit 
\be 
F(\tau,\lambda,Q)=class({\underline t})+\sum_{g=0}^\infty \lambda^{2 g-2} F^{(g)}(\tau,Q)\ . 
\label{freeenenergy}
\ee
The classical terms 
\be 
class({\underline t})= \lambda^{-2} c_3({\underline t})+ c_1({\underline t})
\label{class}
\ee
are at genus zero a cubic polynomial $c_3({\underline t})$ in the K\"ahler parameter $t$ that can be 
read from (\ref{prepot}) without the constant term. Similarly $c_1({\underline t})$ can be read 
from (\ref{F1}). Constant terms are proportional to the Euler number and are included via the 
$0$ in the sum over $H_{2}(B_n,\mathbb{Z})$ in (\ref{freeenenergybase}). 
This a formal expansion in the string coupling $\lambda$, but each 
genus $g$ amplitude $F^{(g)}(\tau,Q)$  has a finite radius of convergence 
in the $(q,Q)$ parameters, determined by the discriminant of the Picard-Fuchs equations.   
We decompose  the genus $g$ amplitudes in fibre $q$ and base moduli $Q$ as   
\be
F^{(g)}(\tau,Q)=\sum_{\beta\in H_{2}(B_n,\mathbb{Z})} F^{(g)}_\beta(\tau) Q^\beta\ .
\label{freeenenergybase}
\ee  
 
There are two closely related properties partially characterizing the $F^{(g)}_\beta(q)$.
\begin{itemize}
\item The  $F^{(g)}_\beta(\tau)$ are, up to a potential sign 
transformation, meromorphic modular forms in $q$ of weight $2g-2$. 
More precisely the $F^{(g)}_\beta(\tau)$ can be decomposed into
\be
 F^{(g)}_\beta(\tau)=Z^{{\rm osc}}_\beta(\tau) P^{(g)}_\beta(\tau)\ .
 \label{defPgb}
\ee 
The first factor $Z^{{\rm osc}}_\beta(\tau)$ the partition function 
of 
\be 
n(\beta)=12 \sum_i\beta_i a^i  
\ee 
bosons, i.e. given by $Z_{\beta}^{osc}=\frac{1}{\eta^{n(\beta)}}$. Here $i=1,\ldots, h_2(B_{n-1})$. 
Since $n(\beta)/12$ can be an odd integer (\ref{defPgb}) can take an additional 
sign under modular transformations.   

The second factor $P^{(g)}_\beta(\tau)$ is a polynomial in the generators  
of almost holomorphic modular forms with respect to the subgroup $\Gamma_F$ of 
${\rm SL}(2,\mathbb{Z})$. In particular for our main examples with the $X_6(1,2,3)$ elliptic fibre type 
$\Gamma_F={\rm SL}(2,\mathbb{Z})$ and the generators are given by almost holomorphic 
weight $2$ Eisenstein series $\hat E_2$  and the holomorphic Eisenstein series $E_4$ and $E_6$. 
The $w$ weight and hence the degree of the polynomial $P^{(g)}_\beta(q)$ grow linearly 
with $g$ and $\beta$ as 
\be
w=2g+6 \sum_i \beta_i a^i -2\ .
\ee 
      
\item They satisfy a holomorphic anomaly equation, which for elliptic fibrations 
one section and only $I_1$ fibers is completely determined by topological data of 
the base~\cite{Klemm:2012} 
\begin{equation} 
\frac{\partial F^{(g)}_{\beta}(\tau)}{\partial \hat E_2} =  
\frac{(-1)^{n+1}}{24}\left( \sum_{h=0}^g 
\sum_{\beta' + \beta''=\beta}  \left(\beta'\cdot\beta''\right) F^{(h)}_{\beta'} \, F^{(g-h)}_{\beta''}+ 
\beta\cdot(\beta-K_B) \,  F^{(g-1)}_{\beta}\right)\ . 
\label{anomaly}  
\end{equation}  
In particular $n$ is the dimension of the base\footnote{Since only one dimension was considered the  
alternating factor was absorbed in the definitions in~\cite{Klemm:2012}. Here we exhibit it to specialize 
to from the threefold to the surface case.}, 
$K_B$ is the canonical class of the base and the dot $\cdot$ is the intersection form on the 
base. $Z^{osc}_\beta$ can be factored out and (\ref{anomaly}) holds equivalently for $P^{(g)}_\beta$.      
Clearly (\ref{anomaly}) can determine the $P^{(g)}_\beta$ only up to a polynomial $P^{(g),hol}_\beta$, 
which depends only on the holomorphic  generators of $\Gamma_F$, i.e.\ in particular not on $\hat E_2$, 
and is called the holomorphic or modular ambiguity.  
\item We define a similar split expansion in fibre and base degree for the 
partition function $Z=\exp(F(\tau,\lambda,Q))$
\be 
\label{generating4.204}
Z=1+\sum_{\beta \in H_2(B_2,\mathbb{Z})} Z_\beta (\tau,\lambda) Q^\beta\ . 
\ee 
We identify now $\lambda=z$ with the elliptic parameter of the weak 
Jacobi forms\footnote{We also use the standard arguments of these forms $(\tau,z)$.}   
and assign modular weight $-1$ to it so that $Z$ and 
$Z_\beta (\tau,\lambda)$ have modular weight zero. From 
(\ref{anomaly}) follows 
\be 
\label{anomalyindex}
\left(\partial_{\hat E_2}+
\frac{\beta\cdot(\beta- K_B)}{24} z^2\right) Z_\beta(\tau,z)=0       
\ee
Basic properties of the weak Jacobi forms are summarized in 
the next section, see equation (\ref{anomalyweak}), yield then 
that the index of $Z_\beta(\tau,z)$ is $\frac{\beta\cdot(\beta- K_B)}{2}$.

\end{itemize}  

As a corollary to the holomorphic anomaly equation  (\ref{anomaly}) we 
have the following. If the base has itself a fibration $B_{n-2}\cdot F=1$ and $F^2=0$, 
we can decompactify the fibre and consider the topological string on the local 
geometry of a rational fibration over the $B_{n-2}$. 

\subsection{The Hirzebruch surface $\mathbb{F}_1$ as base} 

In this section we present the base  $\mathbb{F}_{\alpha}=\mathbb{F}_1$ as an explicit example and discuss the 
relation between the holomorphic anomaly equation in various dimensions. Our base has 
two divisor classes a minus $-\alpha$ curve which we call $S$  and the fibre class $F$. The canonical 
class of the base is $K_B=2S+(2+\alpha)F$. One has the intersections
$S^2=-\alpha$, $FS=1$ and $F^2=0$. We denote $\beta=m F+ n S$ and a short form 
$F^{(g)}_{mF + nS}=F^{(g)}_{m,n}$ so that our anomaly formula (\ref{anomaly}) reads 
\begin{equation} 
\begin{array}{rl} 
\displaystyle{\frac{\partial F^{(g)}_{m,n}(q)}{\partial \hat E_2}} =  
-&\displaystyle{\frac{1}{24}\biggl( \sum_{h=0}^g \sum_{k=0}^m \sum_{l=0}^n [k(n-l)+l(m-k)- \alpha l(n-l) ] F^{(h)}_{k,l} F^{(g-h)}_{m-k,n-l}+}\\ &  
\displaystyle{[2m(n-1)-n(\alpha (n-1)+2) ] F^{(g-1)}_{m,n}\biggr)}\ .
\end{array}
\label{anomalyF1} 
\end{equation} 
This formula has interesting specialization, e.g. for the $\alpha=1$ case we get 
\begin{itemize}
\item In the $n=0$ case corresponds to the corollary mentioned above. The two dimensional 
surface is a specially polarized $K3$. In this case we get a very simple recursion in the genus 
\begin{equation} 
\frac{\partial F^{(g)}_{m}(q)}{\partial \hat E_2} =  -\frac{m}{12} F^{(g-1)}_{m}\ ,
\end{equation}
but no recursion in the base degree. An all genus result can be obtained for all $m$ using the heterotic 
one loop result.   
\item The $m=0$ case corresponds likewise to the corollary. The rational elliptic surface is the 
$\frac{1}{2}K3$. The solution has been discussed first in~\cite{Klemm:1996hh} and in the 
context of the (refined) holomorphic anomaly equation in~\cite{Huang:2013yta,Haghighat:2014,Cai:2014}.  
A quiver description has been found in~\cite{Kim:2014}. 
\item $m=n$ corresponds to the blow down of the Hirzebruch surface $\mathbb{F}_1$ to 
$\mathbb{P}^2$. For this geometry we get hence 
\begin{equation} 
\label{anomalyP2} 
\frac{\partial F^{(g)}_{m}(q)}{\partial \hat E_2} =  -\frac{1}{24}\biggl( \sum_{h=0}^g \sum_{k=0}^m  [k(m-k)] F^{(h)}_{k} F^{(g-h)}_{m-k}
+[m(m-3) ] F^{(g-1)}_{m}\biggr)\ .
 \end{equation} 
The solution will be discussed in section~\ref{exactbase}.    
\end{itemize}

\subsection{The ring of weak Jacobi forms}
\label{ringofweakjacobiforms}
Jacobi forms $\varphi:\mathbb{H}\times \mathbb{C}\rightarrow \mathbb{C}$ depend on a modular 
parameter $\tau\in \mathbb{H}$ and an elliptic parameter $z\in \mathbb{C}$. They transform 
under the modular group~\cite{EZ}
\be
\tau \mapsto \tau_\gamma=\frac{a \tau+ b}{c \tau + d}, \quad z \mapsto z_\gamma=\frac{z}{c \tau + d}\quad  {\rm with} \quad  
\left(\begin{array}{cc}a&b\\ a& c \end{array}\right) \in {\rm SL}(2; \mathbb{Z})  
\ee 
as  
\be
\varphi\left(\tau_\gamma, z_\gamma\right)= (c \tau + d)^k e^{\frac{2 \pi i m c z^2}{c \tau + d}} \varphi(\tau,z)        
\label{mod}
\ee
and under quasi periodicity in the elliptic parameter as
\be
\varphi(\tau,z +\lambda \tau+ \mu)=e^{- 2 \pi i m (\lambda^2 \tau+ 2 \lambda z)}\varphi(\tau,z), \quad \forall \quad \lambda, \mu \in \mathbb{Z} \ .
\label{shiftJacobi} 
\ee
Here $k\in \mathbb{Z}$ is called the {\sl weight} and $m\in \mathbb{Z}_{>0}$ is called the {\sl index} of the Jacobi form. 

The Jacobi forms have a Fourier expansion
\be 
\phi(\tau,z)=\sum_{n,r} c(n,r) q^n y^r, \qquad {\rm where} \ \ q=e^{2 \pi i \tau},\ \  y=e^{2 \pi i z}
\ee
Because of the translation symmetry one has $c(n,r)=:C(4 n m-r^2,r)$, which depends 
on $r$  only modulo  $2m$. For a holomorphic Jacobi form $c(n,r)=0$ unless $4 mn \ge r^2$, 
for cusp forms $c(n,r)=0$ unless $4 mn > r^2$, while for weak Jacobi forms one has only the 
condition $c(n,r)=0$ unless $n\ge 0$. 

According to~\cite{EZ}~\footnote{A review of the theory can be found in~\cite{Dabholkar:2012}. 
We try to follow the notation used there.}, a weak Jacobi form of given index $m$ and even modular weight $k$ is freely 
generated over the ring of modular forms of level one, i.e. polynomials in $R=E_4(\tau)$,  $P=E_6(\tau)$, 
$A=\phi_{0,1} (\tau,z)$, $B=\phi_{-2,1} (\tau,z)$ as
\be 
J^{weak}_{k,m}=\bigoplus_{j=0}^m M_{k+2j} ({\rm SL}(2,\mathbb{Z})) \varphi^j_{-2,1} \varphi_{0,1}^{m-j}\ . 
\ee
We summarize the weights and index of some important forms in the table 1. 
\begin{table} [!ht]
\label{weightsandindices}
\begin{center}{
\begin{tabular} {c|cccccc|}  
            & $Q=E_4$     & $R=E_6$     &   $A=\phi_{-2,1}$     & $B=\phi_{0,1}$  &  $\varphi_{d_B}$ &  $Z_{d_B}(\tau, z)$ \\ \hline 
weight\ \ k:& 4           & 6           & -2                     & 0 &  $ 16 d_B$                        &     0  \\  
index \ \ m:& 0           & 0           & 1                     &  1 &  $\frac{1}{3}d_B(d_B-1)(d_B+4)$   &   $\frac{d_B(d_B-3)}{2}$    \\ 
\end{tabular}}
\caption{$E_4$,  $E_6$,  $\phi_{0,1}$, $\phi_{-2,1}$ are generators of the ring of weak Jacobi forms with even weights, 
$\varphi_{d_B}(\tau,z)$ captures the all genus amplitudes for the $X_{18}(1,1,1,6,9)$ CY 3-fold and  
$Z_{d_B}(\tau, z)$ is the ratio of weak Jacobi forms.}
\end{center} 
\end{table}
The first generators are defined in (\ref{Eisenstein}). Our conventions for the elliptic theta function and weak Jacobi forms $A$ and $B$ 
are the followings\footnote{Our conventions for the  $\theta$ functions associated to the spin structure on the torus are  
\be 
\Theta\left[a \atop b\right](\tau,z)=\sum_{n\in \mathbb{Z}} e^{\pi i (n + a)^2 \tau + 2 \pi i z (n+a) + 2 \pi i b}\  
\ee
and the  Jacobi theta functions $\theta_1=i\Theta\left[\frac{1}{2} \atop \frac{1}{2}\right]$, 
$\theta_2=\Theta\left[\frac{1}{2} \atop 0\right]$, $\theta_3=\Theta\left[0 \atop 0\right]$ and 
$\theta_4=\Theta\left[0 \atop \frac{1}{2}  \right]$.}
\begin{eqnarray}  \label{conventions4.205}
\theta_1(\tau,z) &=&  z \cdot \eta(\tau)^3   \exp[ \sum_{k=1}^{\infty} \frac{ B_{2k}}{2k(2k)!} (iz)^{2k} E_{2k} (\tau) ],  \nonumber \\ 
A=\phi_{-2,1} (\tau,z) &=& -\frac{ \theta_1(\tau,z)^2} {\eta^6(\tau)},  \nonumber \\ 
B=\phi_{0,1} (\tau,z) &=& 4[\frac{ \theta_2(\tau,z)^2} {  \theta_2(0, \tau)^2 } +\frac{ \theta_3(\tau,z)^2} {\theta_3(0, \tau)^2 } +\frac{ \theta_4(\tau,z)^2} {  \theta_4(0, \tau)^2 } ]. 
\end{eqnarray}

The weak Jacobi form of index $\phi_{-2,1}$ has simple product form using the 
Jacobi triple product for $\theta_1$ and 
$x=\left(2 \sin\left(\frac{\lambda}{2}\right)\right)^2=-(y^\frac{1}{2}-y^\frac{1}{2})^2$  
\be 
A=- x \prod_{n=1}^\infty \frac{(1-y q^n)^2 (1-y^{-1} q^n)^2}{(1-q^n)^4})\ . 
\ee
and the weight zero index one form is one half of the elliptic genus of the $K3$
\be
\label{ellipticgenusK3}
\chi(K3;q,y)=2 \phi_{0,1}(\tau,z)=
\left(2 y+20+\frac{2}{y}\right)+\left(\frac{20}{y^2}-\frac{128}{y} +216 -128 y+20 y^2\right)q+  {\cal O}(q^2)
\ee

\subsection{Weak Jacobi Forms and holomorphic anomaly equation}
\label{weakjacobiHA} 
There is a simple connection between the ring of quasi modular forms, which 
have ring isomorphism to the ring of almost holomorphic forms\cite{KZ}, 
which are crucial in the solution of the holomorphic anomaly equations in rank 
one Seiberg Witten theories~\cite{Huang:2006si} and local 
Calabi-Yau spaces~\cite{Aganagic:2006wq}. 

In this section we will show by a very simple argument that the master holomorphic anomaly 
equation (\ref{anomalyindex}) for fibre modularity, which reads 
for the main example $X_{18}(1,1,1,6,9)$
\be
\label{anomalyindexP2}
\left(\partial_{\hat E_2}+
\frac{d_B(d_B-3)}{24} z^2\right) Z_{d_B}(\tau,z)=0,        
\ee
is solved by a weak Jacobi Form of index  $m=\frac{d_B(d_B-3)}{2}$. 

Because of (\ref{mod}) and (\ref{quasimodularshift}) given a weak Jacobi form 
$\varphi_{k,m}(\tau,z)$ one can always define modular form of weight $k$ as follows  
\be
\tilde \varphi_{k}(\tau,z)=e^{\frac{\pi^2}{3} m z^2 E_2(\tau)} \varphi_{k,m}(\tau,z) \ .
\label{makemodular}
\ee
It follows that the weak Jacobi forms $\varphi_{k,m}(\tau,z)$ have a Taylor expansion in  
$z$ with coefficients that are quasi-modular forms as~\cite{EZ,Dabholkar:2012}  
\begin{equation}
\varphi_{k,m}=\xi_0(\tau)+\left(\frac{\xi_0(\tau)}{2}+\frac{m \xi'_0(\tau)}{k}\right) (2 \pi i z)^2 + \left( \frac{\xi_2(\tau)}{24} + \frac{m \xi'_1(\tau)}{2 (k+2)}+ \frac{m^2 \xi''_0(\tau)}{2  k(k+1)}\right)
(2 \pi i z)^4+{\cal O}{(z^6)}\ . 
\label{WeakJacobiTaylor}
\end{equation}
Here the $'=d_{\tau}=\frac{\rm d}{2 \pi i {\rm d}\tau}$ and $\xi_\nu\in \tilde {\cal M}_{k+2 \nu}(\Gamma_0)$, i.e. the coefficients of can be expressed as polynomials 
of Eisenstein series $E_2(\tau), E_4(\tau)$ and $E_6(\tau)$. Moreover from (\ref{makemodular}) one has  
\be
\left(\partial_{E_2} + \frac{m z^2}{12}\right)\varphi_{k,m}(\tau,z)=0\ .
\label{anomalyweak} 
\ee
{\sl Prop. 1} of~\cite{KZ}\footnote{See also \tt{http://people.mpim-bonn.mpg.de/zagier/}.} implies the claim (\ref{anomalyindexP2}). Moreover by (\ref{almostholomorphic}) 
we can write this as holomorphic anomaly equation 
\be
\label{hamaster} 
\left(2 \pi i {\rm Im}^2(\tau)\bar \partial_{\bar \tau} - \frac{m z^2}{4}\right)\hat \varphi_{k,m}(\tau,z)=0\ .
\ee

We note as examples of (\ref{WeakJacobiTaylor}) the first coefficients in the expansion of  $\phi_{-2,1} (\tau,z)$ and $\phi_{0,1} (\tau,z)$ 
are  
\begin{eqnarray}
 \phi_{-2,1} (\tau,z) &=& -z^2 + \frac{E_2 z^4}{12} + \frac{-5 E_2^2 + E_4}{1440} z^6 + \frac{35 E_2^3 - 21 E_2 E_4 + 4 E_6}{362880}z^8  + \mathcal{O}(z^{10}), \nonumber \\ 
  \phi_{0,1} (\tau,z) &=& 12 - E_2 z^2 + 
 \frac{E_2^2 + E_4}{24} z^4 + \frac{-5 E_2^3 - 15 E_2 E_4 + 8 E_6}{4320} z^6  + \mathcal{O}(z^8) . 
 \end{eqnarray} 
Therefore $\phi_{-2,1} (\tau,z)$ and $\phi_{0,1} (\tau,z)$ can be thought of as quasi-modular forms and we can see that 
they satisfy the modular anomaly equation  
\begin{eqnarray}  \label{anomaly4.206} 
\partial_{E_2}  \phi_{-2,1} (\tau,z) =-\frac{z^2}{12} \phi_{-2,1} (\tau,z), ~~~~
\partial_{E_2}  \phi_{0,1} (\tau,z) =-\frac{z^2}{12} \phi_{0,1} (\tau,z). 
\end{eqnarray}

Let us finish the section with a comparison of (\ref{anomalyindexP2},\ref{hamaster}) with Witten's wave 
equation for the the topological string partition function that 
reads~\cite{Witten:1993ed}
\be
\label{Witten}  
\left(\frac{\partial}{\partial (t')^{\bar a}} +\frac{i}{2}\lambda^2 C_{\bar a\bar b \bar c} g^{b\bar b} g^{c \bar c} \frac{D}{D t^b} \frac{D}{D t^c}\right) Z(\lambda,\tau,t_B)=0\ .
\ee
If we apply this equation to $Z$ defined in (\ref{generating4.204}) with $(t')^{\bar a}=\bar \tau$ and $Q^\beta=e^{2 \pi i d_B t_B}$, 
we get in the large base because of the special from of the intersection matrix of elliptically fibered Calabi-Yau 3 folds only derivatives in the base 
direction $t_B$  for $t^b$ and $t^c$. Identifying $\lambda$ with $z$ we see already that the index will 
grow quadratically in $d_B$. A more detailed analysis as in \cite{Klemm:2012} shows also the shift by $K_B$ so that 
the large base limit of (\ref{Witten}) becomes equivalent to all equations (\ref{anomalyindex}).

\subsection{Exact formulae for base degree zero} 

For base degree ${\underline d_B}={\underline 0}$, since the BPS numbers $n^{(g)}_{d_E,{\underline 0}}=0$ 
for genus $g\geq 2$, we get from (\ref{schwinger}) the formula 
\begin{equation} 
\label{multicover} 
P^{(g)}_0 =-\chi \frac{B_{2g}}{2 g (2g-2)!} \left( \frac{B_{2g-2}}{2 (2 g-2)} -\sum_{d=1}^\infty {\rm Li}_{3-2g}(q^d)\right),
 ~~~ g\geq 2, 
\end{equation}
where $\chi$ is the Euler number. This formula holds not just
for the cases of elliptic fibrations with only $I_1$ Kodaira 
fibers~\cite{Klemm:2012}, whose toric description is reviewed  
in section \ref{toricellipticcy}, but also for the case discussed 
in~\cite{Haghighat:2014vxa} with higher Kodaira fibre and 
non-abelian gauge symmetries.          

We recognize in the multi-cover formula (\ref{multicover}) 
the well-known formula for Eisenstein series 
\begin{eqnarray}  \label{zero3.177}
P^{(g)}_0 = -  \chi \frac{B_{2g}B_{2g-2} }{8 g (g-1)(2g-2)!}  E_{2g-2}(q), ~~~ g\geq 2.  
\end{eqnarray} 
So we see for $\chi=-540$ that $P^{(2)}_0=-\frac{3}{32} E_2$, and in general 
$P^{(g)}_0$ is a $SL(2,\mathbb{Z})$ modular form of weight $2g-2$ for $g>2$. 

For the $d_B=0$  genus one amplitude for our case, we can compute the 
base degree zero amplitude $P^{(1)}_0$ from the formula (\ref{genusoneamplitude}) and modularity formulas proven 
in~\cite{HKKZ}. Taking $z_2\rightarrow 0$ of theses, we compute the determinant of 
K\"ahler metric as $\det(G) \sim (\partial_{z_1} t_E) (\partial_{z_2} t_B) \sim \frac{E_4(q_E)^{\frac{5}{2}}} {\eta(q_E)^{24}} \frac{1}{z_2},$
with $\partial_{z_1} t_E = \frac{E_4(q_E)^{\frac{5}{2}}} {\eta(q_E)^{24}} + \mathcal{O} (z_2)$. We can hence  
write the genus one amplitude in the $z_2\rightarrow 0$ limit as 
\be
\mathcal{F}^{(1)} =  -\frac{25}{4} \log(E_4) -\frac{1}{2} \log(\frac{E_4(q_E)^{\frac{5}{2}}} {\eta(q_E)^{24}} ) -
\frac{19}{4} \log(z_1) -\frac{3}{2} \log(z_2) -\frac{1}{12} \log(\Delta_1\Delta_2)  + \mathcal{O} (z_2).  
\ee
A careful calculation show that 
\begin{eqnarray} \label{F01basezero}
P^{(1)}_0 = -48 \log(\eta(q_E)) -\frac{3}{2} \log(q_Bq_E^{\frac{3}{2}}). 
\end{eqnarray}

For the $d_B=0$  genus the zero amplitude  $P^{(0)}_0$ has formally weight $-2$. We therefore   
consider the three point coupling which is the well defined observable after fixing the the 
conformal Killing symmetries on the world-sheet. The classical contribution is a cubic polynomial 
of the flat coordinates (\ref{prepot}), which encode the classical triple intersection numbers. 
Similar to genus one formula (\ref{F01basezero}), we should shift the base flat coordinate according 
to (\ref{redefinitionti}) and find 
\begin{eqnarray}
\partial_{t_E}^3 P^{(0)}_0 =\frac{9}{4}E_4(q_E),
\end{eqnarray}
which confirms the effective modular weight of $-2$ since each derivative 
increases modular weight by 2.

\subsection{Exact formulae for higher base degrees} 
\label{exactbase} 
Our main claim is that the all genus partition function for the topological string 
elliptic fibration over  $\mathbb{P}^2$ realized e.g. as degree 18 hypersurface 
in the  weighted projective space is given by 
\be
Z_{d_B}(\tau,z)=\frac{\varphi_{d_B}(\tau,z)}{\eta(\tau)^{36 d_B}\prod_{s=1}^{d_B} \varphi_{-2,1}(\tau,s z)},
\label{formZdB}
\ee
where $\varphi_{d_B}$ is an weak Jacobi form of weight $k=16 d_B$ and 
index $m=\frac{d_B}{3}(d_B-1)(d_B+4)$. Scaling $z$ in $\varphi_{k,m}(\tau,z)\in J_{k,m}$ 
by $s$ corresponds to a Hecke-like operator $U_s:J^*_{k,m}\rightarrow J^*_{k, s^2m}$. The
weights 
and indices of the relevant  weak Jacobi forms are summarized in Table~\ref{weightsandindices}.  
Since 
\be
{\rm dim} M_k({\rm SL}(2,\mathbb{Z}))=\left\{
\begin{array}{rl} \left[k/12\right]+1 & {\rm if} \ k\neq 2 \ {\rm mod} \  12 \\ 
                  \left[ k/12\right] & {\rm if} \ k=2 \ {\rm mod} \  12\\  
\end{array}\right\} 
\ee
the number of coefficients in $\varphi^{ts}_{d_B}$ growth like $\frac{d_B^6}{108}$ 
for large $d_B$. The first few number for $d_b=1,2,\ldots$ are  $2, 17, 84, 278, 737, 
1692, 3501,\ldots$. 

For the case of $d_B=1$ the two coefficients are fixed by two genus $0$  BPS numbers   
\begin{eqnarray}  \label{formulanb1}
\varphi_1 = - \frac{Q (31 Q^3 + 113 P^2)}{48} . 
\end{eqnarray} 
This determines for base degree 1 all genus BPS invariants by 
(\ref{formZdB},\ref{generating4.204},\ref{defPgb}) and the multi-covering 
formula (\ref{schwinger}). Up to $g=6$ we and $d_E=6$ we  list them in the table (\ref{basedegree1table}). 
We notice that they  match all the numbers that have been obtained 
to genus $8$ using the direct integration of the holomorphic anomaly 
condition, the involution symmetry and the conifold gap condition. A 
detailed  discussion on the perfect matching with the accessible 
enumerative invariants can be found in section \ref{sec:geomcurves}     
\begin{table} [h]
\begin{center} {\footnotesize
 \begin{tabular} {|c|c|c|c|c|c|c|c|} \hline $g\backslash d_E$  
   & $d_E=0$ & 1 & 2 & 3 & 4 & 5 & 6 \\  \hline 
$g=0$ &  3& -1080 & 143370& 204071184& 21772947555& 1076518252152& 33381348217290 \\  \hline 
 1 & 0 &     -6&   2142&   -280284& -408993990 &  -44771454090& -2285308753398 \\  \hline 
 2 & 0& 0& 9& -3192& 412965& 614459160& 68590330119\\  \hline 
 3 & 0& 0& 0& -12& 4230& -541440& -820457286\\  \hline 
 4 & 0& 0& 0& 0& 15& -5256& 665745 \\  \hline 
 5 & 0& 0& 0& 0& 0& -18& 6270 \\ \hline 
 6 & 0& 0& 0& 0& 0& 0& 21 \\ \hline 
 \end{tabular}
 }
\caption{Some BPS invariants $n_{(d_E,1)}^{g}$ for base degree $d_B=1$ and $g,d_E\le 6$ 
as determined by (\ref{formulanb1}) for all $g,d_E$.}  
\label{basedegree1table}
 \end{center}
\end{table} 

Note that(\ref{formZdB}) implies Hilbert-Scheme like infinite product formulae for the generating functions. E.g. for 
$d_B=1$ one gets    
\be 
Z_1(\tau,z) =  \left[\frac{1}{2i \sin(\frac{z}{2})\eta(\tau)^{18}}\prod_{n=1}^{\infty} \frac{(1-q^n)^2}{(1-e^{iz}q^n)(1-e^{-iz}q^n)} \right]^2 \varphi_1 \ . 
\label{basedegree1} 
\ee
where $B_{2k}$ are the Bernoulli numbers. This infinite product formula alone  
eliminates the subspace $V_{\pm}^{(1, [\frac{g-1}{3}], 1)}$ in the holomorphic 
ambiguity (\ref{holoX3.94}) and  allows already to topological string 
amplitudes up to  $g\leq 18$ for this model.

For the $d_B=2$ three of $17$ coefficients can be already fixed by 
demanding that there is no pole $z^{-4}=\lambda^{-4}$ in $P_2(\tau,z)$. 
Note that this pole has to be canceled by the $(Z_1)^2$ contribution 
in $P_2(\tau,z)$. This explains the first term in (\ref{formulanb2}).  
The vanishing bound which we will demonstrate in Section~\ref{sec:geomcurves} is 
a Castelnuovo-like criterium, namely  
that 
\be 
n_{d_E,2}^{g}=0,\qquad {\rm for}\ \  d_E \ge 5, \ g \ge 2 d_E-3         
\ee
fixes eleven other constants and implies all vanishing on the 
right from the edge.  The actual nonzero values 
of the BPS numbers on the generic part of the  edge for 
$d_E\ge 5$  
\be 
n_{d_E,2}^{2 d_E-4}=12-6 d_E 
\label{generic}   
\ee
are then calculated geometrically and yield no further constraints. In order 
to fix the remaining three constants one needs  the 
information of any three nonzero numbers in the same row 
or the same column way from the numbers (\ref{generic}). Notice 
that  three numbers $n_{d_E,2}^{2 d_E-5}$, $d_E=5+i$, $i=0,1,2$ 
do yield only one constraint. Random patterns of nonzero 
numbers away (\ref{generic}) will  lead in general to 
independent equations. The result e.g. fixed from three genus 
zero numbers gives 
\begin{eqnarray} \label{formulanb2}
\varphi_2 &=& \frac{B^4 Q^2 \left(31 Q^3+113 R^2\right)^2}{23887872}+ 
\frac{1}{1146617856}[2507892 B^3 A Q^7 R+9070872 B^3 A Q^4 R^3\nonumber \\  && 
+2355828 B^3 A Q R^5+36469 B^2 A^2 Q^9+764613 B^2 A^2 Q^6 R^2-823017 B^2 A^2 Q^3 R^4\nonumber \\  && 
+21935 B^2 A^2 R^6-9004644 B A^3 Q^8 R-30250296 B A^3 Q^5 R^3-6530148 B A^3 Q^2 R^5\nonumber \\  
&& +31 A^4 Q^{10}+5986623 A^4 Q^7 R^2 +19960101 A^4 Q^4 R^4+4908413 A^4 Q R^6]\ , 
\end{eqnarray}  
which predicts the  BPS numbers in all genus and fibre classes for $d_B=2$. 
\begin{table} [h]
\begin{center} {\scriptsize
 \begin{tabular} {|c|c|c|c|c|c|c|c|} \hline $g\backslash d_E$  
   & $d_E=0$ & 1 & 2 & 3 & 4 & 5 & 6 \\  \hline 
$g=0$& 6& 2700& -574560& 74810520& -49933059660& 7772494870800& 31128163315047072 \\  \hline 
 1 & 0& 15& -8574& 2126358& 521856996& 1122213103092& 879831736511916 \\  \hline 
 2 & 0& 0& -36& 20826& -5904756& -47646003780& -80065270602672\\  \hline 
 3 & 0& 0& 0& 66& -45729& 627574428& 3776946955338\\  \hline 
 4 & 0& 0& 0& 0& -132& -453960& -95306132778 \\  \hline 
 5 & 0& 0& 0& 0& 0& -5031& 1028427996\\ \hline 
 6 & 0& 0& 0& 0& 0& -18& -771642 \\ \hline
 7 & 0& 0& 0& 0& 0& 0& -7224 \\ \hline
 8 & 0& 0& 0& 0& 0& 0& -24 \\ \hline
 \end{tabular}}
\caption{Some BPS invariants for  $n_{(d_E,2)}^{g}$}  
\label{basedegree2}
 \end{center}
\end{table}

For $d_B=3$ the vanishing,  given by
\be 
n^{g}_{d_E,3}=0,\qquad {\rm for}\ \  d_E \ge 8, \ g \ge 3 d_E-10         
\ee
fixes $74$ other constants. These conditions are not independent from 
the conditions that eliminate the $z^{-6}$ and $z^{-4}$ poles and implies 
all vanishing on the right from the edge. If we impose the conditions 
successively we get $\{25, 47, 58, 63, 67, 70, 71, 72, 73, \overline {74}\}$ 
conditions for $d_E=8,9,\ldots$. The genus zero and one invariants yield $5$ and $5$ 
conditions fixing all coefficients of $\varphi_3$.

As we will argue in section \ref{sec:geomcurves}  the vanishing conditions for general $d_B$ is 
\be 
n_{d_E,d_B}^{g}=0,\qquad {\rm for}\ \  d_E \ge 3 d_B-1, \ g \ge d_B d_E-\frac{1}{2} (3 d_B^2 - d_B - 4 )\ .       
\ee
Together with the data up to genus eight tabulated in appendix A, this  allows to fix 
$\varphi_4$ with $20$  and $\varphi_5$ with $2$ nontrivial checks respectively. 
The tables of base degree $d_B=4$ and $d_B=5$ to genus $27$ and $41$ respectively  
can be found likewise in the appendix. We appended the expression for 
$\varphi_d$, $d=1,..,5$  in a for algebraic programs readable form to this TeX file. 

One can make further checks on the results  bases on considerations in section (\ref{sec:geomcurves}). 
For examples we see that for each $d_B$ there is a critical value 
\be 
(g^c(d_B)=\frac{1}{2}(3 d_B^2-d_B+2),d^c_E(d_B)=3 d_B-1)
\ee
in the $(g,d_E)$ plane so that the BPS states on 
the line in the $(g,d_E)$ takes the value 
\be 
n^{d_B d_E-g_c(d_B)+2}_{d_E,d_B}= (-1)^\frac{2 d_E d_B+d_B (d_B-1)}{2}  3 ( d_B d_E - (3 d_B^2+ d_B -6)/2)\ {\rm for} \ d_E\geq 3 d_B-1,  
\label{genericvalue}
\ee 
which can be confirmed as well as many other predictions by direct 
curve counting in section~\ref{sec:geomcurves}.

\begin{table} [h]
\begin{center} {\scriptsize
 \begin{tabular} {|c|c|c|c|c|c|c|c|c|c|c|} \hline $g\backslash d_E$  
   & $0$ & 1 & 2 & 3 & 4 & 5 & 6 & 7& 8& 9  \\  \hline 
$0$& 27&-17280&5051970&-91(5)00&22(8)00&-42(10)68&40(12)60&-16(16)20&55(17)80&12(22)00\\  \hline 
1 & -10&4764&-1079298&15(5)86&-16(7)80&-33(9)68& 12(12)88&-55(14)44&10(18)24&37(21)76\\ \hline 
2 & 0&27&-16884&4768830&-81(5)36&28(8)63&67(10)51& 45(13)02&36(16)60&54(19)78 \\  \hline 
3 & 0&0&-72&48036&-14(4)90&29(5)64&-79(9)12& -63(12)64&-94(15)98&-11(19)86\\  \hline 
4 & 0&0&0&154&-110574&38(4)41&21(8)04&50(11)76& 12(15)04&20(18)021\\  \hline 
5 & 0&0&0&0&-306&247014&-25(6)52&-21(10)36& -11(14)06&-30(17)40 \\  \hline 
6 & 0&0&0&0&0&612&1401468&50(8)62&74(12)04& 34(16)43 \\ \hline 
7 & 0&0&0&0&0&0&17386&-49(6)70&-35(11)70& -32(15)78\\ \hline
8 & 0&0&0&0&0&0&63&3396663&11(10)79& 23(14)91 \\ \hline
9 & 0&0&0&0&0&0&0&32418&-21(8)14&-13(13)22  \\ \hline
10 & 0&0&0&0&0&0&0&108&18(6)07&57(11)85             \\ \hline    
11& 0&0&0&0&0&0&0&0&-1151442&-16(10)36          \\ \hline    
12& 0&0&0&0&0&0&0&0&-10917&29(8)80              \\ \hline    
13& 0&0&0&0&0&0&0&0&-36&-24(6)86              \\ \hline    
14& 0&0&0&0&0&0&0&0&0&1458792           \\ \hline    
15& 0&0&0&0&0&0&0&0&0&13770             \\ \hline    
16& 0&0&0&0&0&0&0&0&0&45             \\ \hline    
  \end{tabular}}
\caption{Some BPS invariants for  $n_{(d_E,3)}^{g}$. To save space we only give 
the first and the last two significant digits and the number of omitted digits in brackets.}  
\label{basedegree3}
 \end{center}
\end{table}

\begin{figure}
\center
\includegraphics[width=14cm]{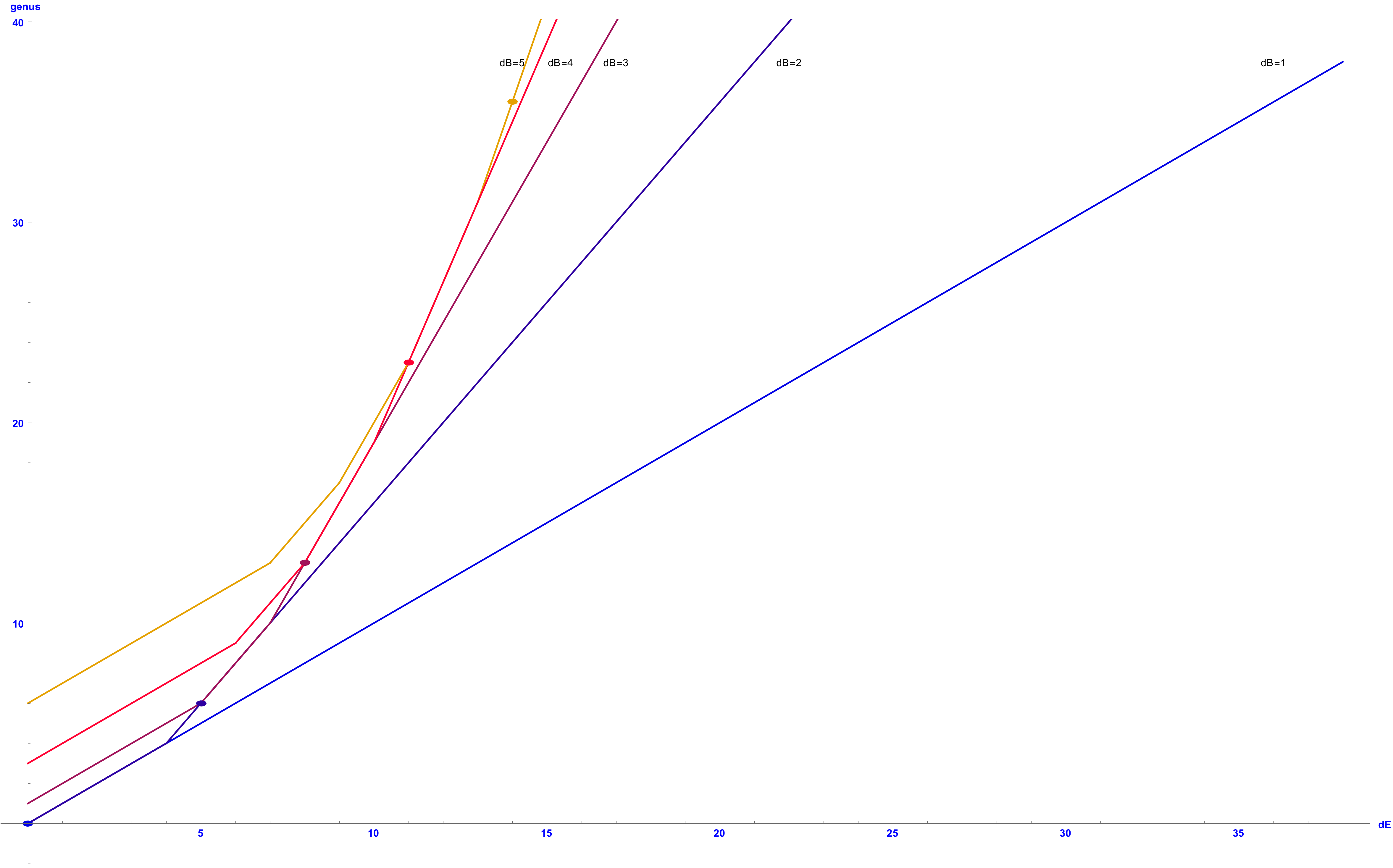}
\caption{The figure shows the boundary of non-vanishing curves  for the  values of 
$d_B=1,2,3,4,5$. The dot on each $d_B$ line at $(g^c(d_B)=\frac{1}{2}(3 d_B^2-d_B+2),d^c_E(d_B)=3 d_B-1)$ 
indicates the value  at which the boundary slope becomes generic and the numbers of 
BPS states on the line are given by (\ref{genericvalue}).}
\label{mon}
\end{figure}

These formulas (\ref{formulanb1}, \ref{formulanb2}) are valid for all genera, so we can consider their implications for 
Gopakumar-Vafa (GV) invariants at large genus. We use the Mathematica program to compute the GV invariants to genus 8, 
and the results are listed in the tables in the Appendix. A well-known feature of the GV invariants from algebraic 
geometric arguments \cite{KKV} is that, for a given 2nd-homology class, 
in our case a pair of fiber and base degrees $(d_E, d_B)$, the  GV invariant $n_{(d_E, d_B)}^g$ vanishes for sufficiently 
large genus $g$. The largest genus with non-vanishing GV invariant is called the top genus for the given degrees $(d_E, d_B)$. 
For example, from the GV tables in the Appendix \ref{appendixA} we can see that the top genus for the degrees $(d_E,1)$ is always $d_E$, 
while the top genus for  $(d_E,2)$ is $d_E$ for the cases of $d_E\leq 4$, and is $6, 8$ for $d_E=5, 6$ respectively. The situation is 
summarized in Figure 1 and the bounds are verified from algebraic geometry in section~\ref{GVfromgeometry}.

We consider whether such vanishing conditions help to fix the unknown coefficients in the ansatz of the weak Jacobi form. 
It turns out for the first case of $d_B=1$, a general ansatz (\ref{formZdB}) ensures the vanishings of $n_{(d_E, 1)}^g=0$ for $g>d_E$, 
therefore this does not help to fix the two coefficients in the numerator in the formula  (\ref{formulanb1}). 
To understand this structure, let us analyze the formula for $\phi_{-2,1} (\tau,z)$ in (\ref{conventions4.205}), 
which can be also written in the infinite product form. For comparing with GV invariants,
 we define the variable $x=(2\sin(\frac{z}{2}))^2$, then we find 
\begin{eqnarray} \label{expan4.212}
\phi_{-2,1} (\tau,z) = -x \prod_{n=1}^{\infty} \frac{(1+ xq^n -2q^n+ q^{2n})^2 }{(1-q^n)^4} .
\end{eqnarray}
So we see that excluding the total factor of $x$, in the power series expansion of the infinite product, 
the power of $q$ is always no less than that of $x$. Using the series expansion $(1+x)^{-2} = 
\sum_{n=0}^\infty (-1)^n(n+1) x^n $, we find the inverse can be written as
\begin{eqnarray} \label{expansion4.213}
\frac{1}{\phi_{-2,1} (\tau,z)} = - \frac{1}{x} \sum_{n=0}^\infty (-1)^n (n+1) q^n [x^n +f_{n-1} (x)] ,
\end{eqnarray}
where $f_{n-1} (x)$ is a polynomial of $x$ of degree $n-1$. Topological string free 
energy in the $d_B=1$ case has no multi-cover contributions from lower degrees, and 
according to the general ansatz (\ref{formZdB}) is simply $\phi_{-2,1} (\tau,z)^{-1}$ multiplying a modular 
form of weight 16 and the factor of $(\frac{q}{\eta(q)^{24})})^{\frac{3}{2}}$, which only 
increase  the $q$ power in the series expansion. Here the factor $q^{\frac{3}{2}}$ comes from the shift of base Kahler parameter such that the partition function has a nice modular form. From the well-known GV expansion of 
the topological free energy, the coefficient of $q^{d_E} x^{g-1}$ is exactly the GV invariant $n_{(d_E,1)}^g$. 
Therefore we see from the expansion (\ref{expansion4.213}) that indeed  the GV invariant $n_{(d_E,1)}^g$ 
vanishes for $g>d_E$. Furthermore, using the formula (\ref{formulanb1}) we can easily calculate the 
top genus numbers $n_{(g,1)}^g = 3(-1)^g (g+1)$. This agrees with the numbers in the tables of 
GV invariants and can be also derived from algebraic geometric arguments in section \ref{sec:geomcurves}.

Now we consider the general cases of higher base degrees. Similar to the case of $\phi_{-2,1} (\tau,z)$ in  (\ref{expan4.212}), we can also expand $ \phi_{0,1} (\tau,z)$ as power series of $x,q$ 
\begin{eqnarray}
\phi_{0,1} (\tau,z) = -x(1-10 x q +x^2q^2) +\sum_{n=0}^{\infty} q^n f_n(x), 
\end{eqnarray}
where $f_n(x)$ symbolizes a degree $n$ polynomial of $x$. Overall the difference of 
the powers of $x$ and $q$ in both the series $\phi_{-2,1} (\tau,z)$ and $\phi_{0,1} (\tau,z)$ is always no larger than 1. 
So we can consider the sub-family of ansatz with only $\phi_{-2,1} (\tau,z) $ in the denominator, namely 
\begin{eqnarray} \label{special3.41} 
 \sum_{k = -1}^ {\frac{d_B(d_B-3) }{2}}  \frac{f_{18d_B +2k }(E_4, E_6)}{\eta(\tau)^{36d_B}}   \phi_{-2,1} (\tau,z)^{k} \phi_{0,1} (\tau,z)^{ \frac{d_B(d_B-3) }{2} -k }, 
\end{eqnarray}
where  $f_{18d_B +2k}(E_4, E_6)  $ symbolizes a modular form of weight $18d_B +2k$, which is a polynomials of $E_4$ and $E_6$ with many unknown coefficients, so that the total modular weight and index are the same as the more general ansatz (\ref{formZdB}) with the 
product $\prod_{k=1}^{d_B} \phi_{-2,1}(\tau, k z)$ in the denominator.  From the power series of 
$\phi_{-2,1} (\tau,z)$ and $\phi_{0,1} (\tau,z)$ in $x,q$ and including the factor of $q^{\frac{3d_B}{2}}$ from the shift of base Kahler parameter, we deduce that the sub-family of ansatz (\ref{special3.41})  
contribute only to the GV invariants $n_{(d_E, d_B)}^g$ with $g\leq d_E+ \frac{d_B(d_B-3) }{2} +1$. 
On the other hand, from the algebra-geometric arguments in section \ref{sec:geomcurves}, the top genus with non-vanishing GV invariants  
for the case of $d_E=0$ is always  $ \frac{d_B(d_B-3) }{2}+1$. For a given $d_B$, the top genus as a function of 
$d_E$ is strictly monotonically increasing, i.e. the top genus for $(d_E+1,d_B)$ is always strictly bigger 
than that of $(d_E,d_B)$. Therefore we deduce that the top genus with non-vanishing GV invariants for 
degrees $(d_E,d_B)$ is always no less than $d_E+ \frac{d_B(d_B-3) }{2} +1 $. Furthermore, 
when  $d_E$ is sufficiently large, from the algebraic geometric arguments, the top genus is 
given by a formula $d_B d_E + \frac{1}{2} (-3 d_B^2 + d_B + 2) $, which is certainly much larger 
than $d_E+ \frac{d_B(d_B-3) }{2} +1$. We can choose a particular solution from the general ansatz (\ref{formZdB}), so that the GV invariants vanish at sufficiently large genus, then the addition of extra ansatz of the form (\ref{special3.41}) does not affect the conditions. So we conclude that  the sub-family of ansatz 
(\ref{special3.41}) can not be fixed by vanishing GV conditions. 

The reverse is also true, namely, any remaining ambiguities which can not be fixed by vanishing 
GV conditions are necessarily of the form (\ref{special3.41}). To see this, we note that 
$(2\sin(\frac{kz}{2}))^2 $ can be written as a polynomial of $x\equiv (2\sin(\frac{z}{2}))^2 $ of 
degree $k$, more explicitly \begin{eqnarray} \label{Cheby3.42}
(2\sin(\frac{kz}{2}))^2  = x f_{k-1}(x) , 
\end{eqnarray} 
where $f_{k-1}(x)$ is a degree $k-1$ polynomial, related to the Chebyshev polynomials. For example, 
for $k=2$, we have $(2\sin z )^2 = x(4-x)$. The constant term  $f_{k-1}(0)$ in the polynomial 
is always $k^2$ since we have  $\frac{(2\sin(\frac{kz}{2}))^2}{x} \sim k^2 $ in the 
limit $x\sim z\sim 0$. When we expand $\frac{1}{\phi_{-2,1} (\tau,k z)}$ for $k> 1$ as 
power series of $x, q$, we see that the power of $x$ has no upper bound for a given power of 
$q$, unlike the case of $k=1$.  So if the factor  $f_{k-1}(x)$ is not cancelled in the denominator, 
the ansatz would contribute to the GV invariants $n_{(d_E, d_B)}^g$ of arbitrarily large 
genus $g$ for the given degrees $(d_E, d_B)$, and we can always fix them with vanishing GV 
invariants, even without the precise knowledge of the top genus. Furthermore, we require 
that the power of $x$ in the generating function of the GV invariants is always no 
less than $-1$, since the lowest genus is zero. Those ansatz that contribute to 
$x^{-n}$ with $n>1$ can be always fixed by this requirement. Overall we have at most one 
factor $x$ in the denominator for those ansatz that are not affected by these considerations. 
Thus all ansatz except for those of the form  (\ref{special3.41}) can be fixed in this way.  

This is checked by actual calculations. One can easily count that the sub-family of ansatz 
(\ref{special3.41}) has $2, 3, 10, 26$ unknown coefficients for base degrees $d_B=1,2,3,4$ respectively. 
After we impose the conditions that the generating function of GV invariants start from genus 
zero and vanish at sufficiently large genus for a given fiber degree, these are the indeed the 
remaining numbers of unfixed coefficients in the more general ansatz (\ref{formZdB}) with the denominator $\prod_{k=1}^{d_B} \phi_{-2,1}( \tau,k z)$.

We can also see the scaling of the top genus for a given base degree $d_B$ and for large fiber 
degree $d_E$. We can expand more explicitly 
\begin{eqnarray} \label{expand3.43} 
\frac{1}{\phi_{-2,1} (\tau,k z)} = -\frac{1}{x f_{k-1}(x)} + \sum_{n=1}^{\infty}  q^n g_{n-1} (x f_{k-1}(x) ) ,
\end{eqnarray} 
where $f_{k-1}(x)$ is the polynomial in (\ref{Cheby3.42}) and $g_{n-1}$ is a degree $n-1$ polynomial.  
There are also multi-cover contributions from lower degrees. For example, we can consider $k$-cover 
contributions from the base degree one formula. The denominator can be expanded as 
\begin{eqnarray}
\frac{1}{\phi_{-2,1} (k\tau, k z)} = -\frac{1}{x f_{k-1}(x)} + \sum_{n=1}^{\infty}  q^{kn} g_{n-1} (x f_{k-1}(x) ), 
\end{eqnarray}
where the polynomials $f_{k-1}, g_{n-1}$ are the same as in (\ref{expand3.43}). We see that the first terms, which 
would have unbounded powers of $x$ in series expansion, are the same. At the end, the polynomial $f_{k-1}(x)$ must be 
cancelled in the denominator, so that the resulting GV invariants have a top genus. In this way we see that the 
minimal requirements for the denominator of the ansatz at base degree $d_B$ is indeed $\prod_{k=1}^{d_B} \phi_{-2,1}(\tau, k z)$, 
otherwise there will be non-vanishing GV invariants at arbitrarily large genus for fixed degrees from the multi-cover 
contributions. Since the factor $f_{k-1}(x)$ in the denominator of the first term in (\ref{expand3.43}) is eventually 
canceled, this term contributes a finite polynomial of $x$ in the generating function of GV invariants and does not 
affect the leading scaling behavior of the top genus for large fiber degree. At the base degree $d_B$,  the leading 
contribution to the top genus comes from the second term in (\ref{expand3.43}) for $k=d_B$ with large $n$, and 
scales like $g_{top} \sim d_B d_E$. However it is more tricky to determine the exact formula in this way 
for the cases of $d_B>1$. We find that the exact top genus numbers are not universal for the general ansatz (\ref{formZdB}),  
and the top genus formula $d_B d_E + \frac{1}{2} (-3 d_B^2 + d_B + 2) $ for large $d_E$ derived from algebraic arguments 
only appears for the the particular solution in topological string theory.   

This approach can be combined with the B-model holomorphic anomaly approach to compute higher genus topological string amplitudes. Using the involution symmetry and the boundary conditions at the conifold point, 
we find that the exact formula at base degree $d_B$ can provide sufficient boundary data to 
fix the B-model formula  at genus $9(d_B+1)$, valid for all base and fiber degrees, see the Appendixes for more details. 
On the other hand, in order to fix the exact formula at base degree $d_B$, we need to fix 
the ansatz (\ref{special3.41}) with some non-vanishing GV invariants. In particular, consider 
the term in (\ref{special3.41}) with $k=\frac{d_B(d_B-3) }{2} $ in the sum. Since 
$ \phi_{-2,1} (\tau,z)\sim x$ in the power series expansion of $x,q$, this term contributes 
only to and can be only fixed by topological string free energy or some non-vanishing 
GV invariants of genus no less than $\frac{d_B(d_B-3) }{2}+1$.  The contributions 
of the other terms in (\ref{special3.41}) start from lower genus. 
On the other hand, the modular form coefficients in the ansatz (\ref{special3.41}) have different modular weights for 
different $k$, and their contributions can not cancel when we extract the contributions at a given genus by expanding 
for small $z$. Therefore if we only know the B-model formula at any genus no less than $\frac{d_B(d_B-3) }{2}+1$, 
it is also already sufficient to fix the ansatz (\ref{special3.41}) at base degree $d_B$.  This is checked by actual 
low base degrees calculations. Thus, as long as $9(d_B+1) \geq \frac{(d_B+1)(d_B-2) }{2}+1$, we can repeat this procedure 
to fix the exact formula with increasing base degrees. If no other obstacle arises, we foresee 
the eventual termination of this recursive process only at $d_B=20$. In this way we can in 
principle determine the exact formula up to base degree $d_B=20$ (for all genera and fiber degrees), 
and the topological string free energy up to genus 189 (for all base and fiber degrees). In practice 
we provide here the solution of the topological strings up to $d_B=5$ for all genera and genus 
$g=8$ for all degrees.

\section{BPS invariants}

After a technical recapitulating of the the physical definition of the 
BPS invariants in Section~\ref{physicsdefBPS}, the main part of the section 
is devoted to the direct geometrical calculation of the unrefined BPS 
invariants or Gopakumar-Vafa invariants for the $X_{18}(1,1,1,6,9)$ 
Calabi-Yau 3-fold in section \ref{GVfromgeometry}.

\subsection{Physical definition of the BPS invariants}
\label{physicsdefBPS} 
In this section we recall the definition of the integer BPS 
invariants from ${\cal F}(\lambda, {\underline t})$ or 
${\cal Z}(\lambda,{\underline t})$. In view of the similarity 
of the all genus expansions discussed in section~\ref{exactbase} 
and the product expansion of (\ref{borcherdsproduct}) based 
on the Borcherds lifting of automorphic forms of $O(2,n,\mathbb{Z})$, 
see \cite{Borcherds} and \cite{kontsevich1} the formulae 
(\ref{productunrefined}) and (\ref{productrefined}) should 
give further hints how to characterize ${\cal Z}$. We comment on 
possible refinements in \ref{refinedBPSinvariants} in analogy to the $E$-string.

\subsubsection{Unrefined BPS invariants}  
\label{unrefinedBPSinvariants}

The unrefined BPS states multiplicities $n_\beta^g\in \mathbb{Z}$ 
are defined by their contribution to a Schwinger-Loop integral in $5d$  $N=2$ 
supergravity theory~\cite{Gopakumar:1998jq}\footnote{Caligraphic letters should refer to quantities in the holomorphic limit.}     
\begin{equation} 
\label{schwinger}
{\cal F}(\lambda, {\underline t})=class({\underline t})+\sum_{\beta\in H_2(M_3,\mathbb{Z})}\sum_{g=0}^\infty\sum_{m=1}^\infty \frac{n_\beta^g}{m} 
\left( 2 \sin\left(\frac{m \lambda}{2}\right)\right)^{2g -2} q^{\beta m} \ .  
\end{equation}  
Note that $q^{\beta}=\exp(2\pi i \sum_{\alpha=1}^{b_2(M_3)} t^\alpha \beta_\alpha)$ and  
$y^\frac{m}{2}-y^{-\frac{m}{2}}= 2i\sin\left(\frac{m \lambda}{2}\right)$.  
One can interpret $\beta$ as BPS charge and $\sum_\alpha t^\alpha \beta_\alpha$ as the  
central charge mass term. 
In the 5-field theory the fugacity $\lambda$ traces the left spin of the BPS states.                

The $n_\beta^g\in \mathbb{Z}$  may be calculated  by identifying the above expression to the 
perturbative topological string expansion on $M$ 
\begin{equation} 
{\cal F}(\lambda, {\underline t})=class({\underline t})+ \sum_{g=0}^\infty \lambda^{2g-2} {\cal F}^{g}({\underline t})= \sum_{g=0}^\infty \sum_{\beta\in H_2(M,\mathbb{Z})} 
\lambda^{2g-2} r^g_\beta q^\beta\ ,
\end{equation} 
where  $\lambda$ is now the topological string coupling constant. The Gromov-Witten invariants  
$r^g_\beta\in \mathbb{Q}$ can be mathematically defined and in suitable cases calculated 
using localization. In this article we use mirror symmetry, modularity, the 
holomorphic anomaly equation and direct curve counting to obtain the $F^{(g)}({\underline t})$. 
With 
\be 
\sum_{m=1}^\infty\frac{1}{m} \frac{ q^m}{ \left( 2 \sin\left(\frac{m \lambda}{2}\right)\right)^2}=\sum_{g=0}^\infty \lambda^{2 g -2} (-1)^{g+1} 
\frac{B_{2g}}{2 g (2 g-2)!} {\rm Li}_{3 -2g} (q)\ 
\ee
where ${\rm Li}_n(x)=\sum_{k=1}^\infty \frac{x^k}{k ^m}$ and the Bernoulli 
numbers are defined in (\ref{Bernoulli}). Using further $\zeta(-n)=-\frac{B_{n+1}}{n+1}$ for $n\in \mathbb{N}_+$ 
we get 
\be 
F_g({\underline t})=(-1)^{g+1}n_0^0\frac{|B_{2g} B_{2g-2}|}{2g (2g-2) (2g-2)!}+{\cal O}({\underline q}), \quad  {\rm for }\quad n>1\ ,
\ee
which matches the constant map contribution~\cite{faberpandharipande} for  $n_0^0=-\frac{\chi}{2}$. 
Exponentiating  (\ref{schwinger}) one gets~\cite{Klemm:2004km,Katz} 
\begin{equation}
{\cal Z}(\lambda,{\underline t})=\prod_{\beta}\left[
\left(\prod_{m=k}^\infty (1-y^k q^\beta)^{k n_\beta^0}\right)
\prod_{g=1}^\infty \prod_{l=0}^{2g-2}(1-y^{ g-l-1} q^\beta)^{(-1)^{g+l} 
\left(2 g-2\atop l\right) n_\beta^g}\right]\  . 
\label{productunrefined}
\end{equation}

\subsubsection{Refined BPS invariants}  
\label{refinedBPSinvariants} 
Let us define as in~\cite{Huang:2010} in analogy with (\ref{freeenenergy},\ref{freeenenergybase})
\be 
F=\sum_{g=0}^\infty (\epsilon_1+\epsilon_2)^{2 h} (\epsilon_1\epsilon_2)^{g-1} F^{(g,h)}(q,Q)\ .
\label{refinedfreeenenergy}
\ee 	
We decompose the genus $g$ amplitude as 
\be
F^{(g,h)}=\sum_{\beta\in H_{2}(B_n,\mathbb{Z})} F^{(g,h)}_\beta(q) Q^\beta\ .
\label{refinedfreeenenergybase}
\ee  
then holomorphic anomaly equation for E-strings~\cite{Huang:2013yta} becomes  
\be
\label{refinedmodularF}
\frac{\partial F^{(g,h)}_{n}}{\partial \hat E_2} = \frac{1}{24} \sum_{h_1=0}^h\sum_{g_1=0}^g \sum_{s=1}^{n-1} s(n-s) 
F^{(g_1,h_1)}_s F^{(g-g_1,h-h_1)}_{n-s}  +\frac{n(n+1)}{24} F^{(g-1,h)}_{n} - \frac{n}{24} F^{(g,h-1)}_{n}\  .
\ee
Making an analogous fibre and base separation as in (\ref{generating4.204}) 
of the $E$-string  geometry $\frac{1}{2} K3$ surface with the rational 
elliptic fibration ${\cal E}\rightarrow \mathbb{P}^1$, one finds \cite{Haghighat:2014}  for 
the massless $E$-string the analogous equation to (\ref{anomalyindex}), compare \cite{Haghighat:2014pva}
\be 
\label{refinedmodularZb}
\left(\partial_{E_2}+\frac{1}{24}[\epsilon_1\epsilon_2 (n^2 + n) -(\epsilon_1+\epsilon_2)^2 n]\right) Z^{E-string}_n(\tau,\epsilon_1,\epsilon_2)=0 \ .    
\ee
We note that in the Nekrasov-Shatashvili limit the solution will be a weak Jacobi form. 
Indeed  the index growth linearly as in the $N=4$ case discussed in the introduction.
In the general case we expect as solution a meromorphic weak Jacobi theta function. The 
refined BPS invariants 
$$
N^{j_L,j_R}_\beta\in \mathbb{N}_0\ . 
$$
are labelled by $\beta\in H_2(M,\mathbb{Z})$ and $j_L,j_R$,  the $5d$ left and right spin
$$j_L,j_R\in \frac{1}{2}\mathbb{Z}_{\ge 0}\ . $$
If we define 
$$[j]_x=x^{2j}+ x^{-2j+2}+\ldots + x^{2j-2}+x^{2j}$$
then (\ref{schwinger}) generalizes with $u=e^{i\epsilon_1}$ and $v=e^{i\epsilon_2}$  to 
\begin{equation} 
{\cal F}(\epsilon_1,\epsilon_2, t)=\sum_{\beta\in \mathbb{Z}^r}\sum_{j_J,j_R}^\infty\sum_{m=1}^\infty \frac{(-1)^{2(j_L+ j_R)} N_\beta^{j_L,j_R} [m j_L]_{(uv)}[m j_R]_{\left(\frac{u}{v}\right)}}
{m (u^\frac{m}{2} -u^{-\frac{m}{2}})(v^\frac{m}{2} -v^{-\frac{m}{2}})} q^{m \beta}\ ,
\end{equation}  
while (\ref{productunrefined}) generalizes to~\cite{Iqbal:2007ii,Choi:2012jz}\footnote{See \cite{Choi:2012jz} for a mathematical definition of the $N^\beta_{j_Lj_R}$ 
in the local case using the virtual Bialynicki-Birula decomposition of the 
moduli space Pandharipande-Thomas invariants.}                  
\begin{equation}
{\cal Z}=\prod_\beta \prod_{j_{L/R}=0}^\infty \prod_{m_{L/R}=-j_{L/R}}^{j_{L/R}}\prod_{m_1,m_2=1}^\infty \left(1-\left(\frac{u}{v}\right)^{m_L} (uv)^{m_R} u^{(m_1-\frac{1}{2})}
v^{(m_2-\frac{1}{2})} q^{\beta}\right)^{ (-1)^{2(j_L+j_R)} N^\beta_{j_Lj_R}}\ ,
\label{productrefined}
\end{equation}
For the E-string (\ref{basedegree1}) generalizes to 
\be 
{\cal Z}_1=\sum_{j_L,j_R} \sum_{n_e=0}^\infty N^{j_L,j_R}_{n_e,1} [j_L]_u [j_R]_v q^{n_e}
= \frac{ E_4(q)}{ \eta(q)^8 \prod_{n=1}^{\infty} (1-u v q^n) (1-\frac{u q^n}{v}) 
 (1-\frac{q^n}{uv})(1-\frac{vq^n}{u})  }\ .
\label{refinedwinding1BPS}
\ee
This can easily seen to be compatible with (\ref{refinedmodularF}). Eqs 
(\ref{refinedmodularF}),(\ref{refinedmodularZb}) and (\ref{refinedwinding1BPS}) 
suggest that a refinement of  (\ref{anomalyP2}), (\ref{anomalyindex}) and 
(\ref{basedegree1}). We will discuss this further in \cite{HKKZ}.

\subsection{Gopakumar-Vafa invariants from geometry} 
\label{GVfromgeometry} 
\subsubsection{Geometry of curves}
\label{sec:geomcurves}
In this section, we describe the geometry of a number of families of curves 
in the $X_{18}(1,1,1,6,9)$ model,
from which a good number of GV invariants can be extracted.  For simplicity, 
we will simply refer to this model as $X$ in this section.
 
We begin by reviewing some results and notation 
from~\cite{Candelas:1994hw}, which will
be used to describe the geometry.

The Calabi-Yau $X$ is described as a blowup of a degree~18 weighted
hypersurface in $\mathbb{P}(1,1,1,6,9)$ at its (unique) singular point
$x_1=x_2=x_3=0$.  As noted earlier, $X$ has Hodge numbers $h^{1,1}=2$ and $h^{2,1}=272$,
and Euler characteristic $-540$.

The exceptional divisor $E$ is isomorphic to $\mathbb{P}^2$.  
Projection to the first 
three coordinates of the weighted projective space presents $X$ as 
a Weierstrass elliptic fibration 
$\pi:X\to\mathbb{P}^2$.  Let $L=\pi^{-1}(\ell)$ be the pullback of a line
$\ell\subset\mathbb{P}^2$ to $X$, and let $H=3L+E$.  Then 
$\pi^*(x_1)$, $\pi^*(x_2)$, and $\pi^*(x_3)$ are sections of $L$,
$\pi^*(x_4)$ is a section of $2H$, and $\pi^*(x_5)$ is a section of $3H$.

For the benefit of those readers more familiar with the toric description, we
refer to (\ref{F1case}) where the two descriptions of the divisors
are related: the toric divisor $D_{x_0}$ is identified with $E$, the divisor
classes of the toric divisors $D_{x_i}$ are all equal to $L$ for $i=1,2,3$, the toric
divisor $D_x$
is in the class $2H$, and the toric divisor $D_y$ is in the class $3H$.

The K\"ahler cone of $X$ is generated by $H$ and $L$.  The triple intersections
are
\begin{equation}
  H^3=9,\qquad H^2L=3,\qquad HL^2=1,\qquad L^3=0.
  \label{eq:yukawa}
\end{equation}

Dually, the Mori cone is generated by the class of an elliptic fiber $f$
and the class of a line $\tilde{\ell}\subset E\simeq\mathbb{P}^2$.

Now, let $C\subset X$ be a connected curve (not necessarily reduced or
irreducible, but having no embedded points).  The invariants of $C$ are
\begin{equation}
  d_E=C\cdot H, \qquad
  d_B=C\cdot L, \qquad   g=p_a(C),
  \label{eq:debg}
\end{equation}
where the last is the arithmetic genus of $C$.  
The curves of fixed $(g;d_E,d_B)$ are parametrized
by a Hilbert scheme, and we will be able to describe many of these as well as
the corresponding GV invariants.

As a comment on the mathematical rigor of our calculations, we are
actually 
studying PT moduli spaces and computing 
PT invariants, which are mathematically equivalent
to the GV invariants 
invariants by \cite{PT}.  However, we will frequently simplify 
the exposition by using the language of Hilbert schemes as in \cite{KKV}
and then describing corrections required by the use of stable pairs.  The reader
is referred to \cite{PT} for more details, and to \cite{Choi:2012jz} for an
amplification of the comparison between the methods of \cite{Choi:2012jz}
and \cite{KKV}

We say that $(d_E,d_B)$ is the
{\em degree\/} of $C$, and sometimes use the notation $C_{d_E,d_B}$ for a curve
of degree $(d_E,d_B)$.  A fiber $f$ has degree $(1,0)$ and a line $\tilde{\ell}
\subset
E$ has degree $(0,1)$.

We are able to completely describe the moduli space of connected curves 
of degree 
$(d_E,d_B)$ and genus $g$ for various $d_E,\ d_B$, and $g$.  In this section
we describe the geometric principles that these descriptions are based on.  In the next
section, we will apply these principles to compute various moduli spaces with $d_B\le 5$
and use this description to compute the associated Gopakumar-Vafa invariants, and make
some comments about general $d_B$.

\begin{lem}
Let $C$ have degree $(d_E,d_B)$.  Then $\pi(C)$
is a plane curve of degree $d_B$, including multiplicity.  
  \label{lem:dproj}
\end{lem}

\begin{proof}
We compute the degree of $\pi(C)$ as
\begin{equation}
  \pi(C)\cdot \ell=C\cdot\pi^{-1}(\ell)=C\cdot L=d_B,
  \label{eq:dprojc}
\end{equation}
where $\ell\subset\mathbb{P}^2$ is a line. 
\end{proof}

\begin{cor}
  A curve of degree $(d_E,0)$ is a union of $d_E$ fibers (including multiplicity).
\label{cor:fibers}
\end{cor}

\begin{proof}
By Lemma~\ref{lem:dproj},  $\pi(C)$ is a finite point set, hence
$C$ is a union of fibers.  Since each fiber has degree $(1,0)$, the result follows.  
\end{proof}

\begin{lem}
If $d_E<3d_B$, then for any curve $C$ of degree $(d_E,d_B)$ at least one component of
$C$ (with its reduced structure) is  contained in $E$.  
In particular, if in addition $C$ is irreducible, then
$d_E=0$ and $C_{\mathrm{red}}\subset E$, where $C_{\mathrm{red}}$ is the reduced 
(multiplicity~1) structure on $C$.
  \label{lem:inplane}
\end{lem}

\begin{proof}
To see this, we first note that if a curve $C'$ is not contained in $E$,
then $C'\cdot E\ge0$, as the intersection number is just a count of the intersection
points of $C'$ and $E$ with multiplicity.  Then we compute 
\begin{equation}
  C\cdot E=C\cdot\left(H-3L\right)=d_E-3d_B.
  \label{eq:inte}
\end{equation}
If this is negative, it follows immediately that some component $C'$ of $C$ must be contained
in $E$. If in addition $C$ is irreducible, it follows that $C_{\mathrm{red}}\subset E$.  
Since  $H\cdot E=0$, it follows 
that $d_E=C\cdot H=0$ as claimed.
\end{proof}

\smallskip\noindent
{\bf Remark.} If  $C_{\mathrm{rd}}\subset E$, it is possible for $C$ itself to not be
contained in $E$ due to thickenings in  a direction transverse to $E$.

\begin{prop}
If $C$ has degree $(0,d_B)$, then its genus satisfies
\[
g\le \frac12\left(d_B-1\right)\left(d_B-2\right),
\]
with equality holding if and only if $C$ is a plane curve of degree $d_B$ after identifying
$E$ with $\mathbb{P}^2$.
  \label{prop:plane}
\end{prop}

The second half of the proposition says that if we try to thicken components of $C$  outside of $E$, the genus would be less than if we thicken inside $E$.  

\begin{proof}
We have $C\cdot H=d_E=0$.  Since $H$ is ample on the weighted hypersurface in
$\mathbb{P}(1,1,1,6,9)$ before the blowup, the only way for the intersection $C\cdot H$
to be 0 is for
$C_{\mathrm{red}}$ to be contained in $E$.  Since any 
$\mathbb{P}^2$ in a Calabi-Yau threefold has a neighborhood which is
isomorphic to a neighborhood of $\mathbb{P}^2$ inside local $\mathbb{P}^2$,
the result follows from the corresponding result for local $\mathbb{P}^2$.

\end{proof}

\noindent
{\bf Remark.} We will see in the next subsection
that for each $d_B\ge2$, non-reduced curves occur generically in moduli spaces
of curves with particular $d_E$ and $g$.

\smallskip
We now set out to find the largest possible genus of a connected
curve of degree $(d_E,d_B)$.
Part of our strategy in studying curves is to first study irreducible curves
and then study how the components can glue together.

By Lemma~\ref{lem:dproj}, a curve $C$ of degree $(d_E,d_B)$ is  contained in a surface
$\pi^{-1}(D)$ where $D\subset E\simeq\mathbb{P}^2$ is a plane curve of degree less than
or equal to $d_B$.  The degree will be strictly less than $d_B$
if and only if some component of $\pi(C)$ has
multiplicity greater than~1.

By analogy with the case of curves in $\mathbb{P}^3$, where the genus of a curve of
degree $d$ is maximized when the curve is contained in a plane, one might expect that in
our situation the
genus is maximized if the degree of $D$ is one, i.e.\ when $C\subset \pi^{-1}(\ell)$
for some line $\ell\subset E$.  This turns out to be the case for irreducible curves
with $d_E>0$.

\begin{prop}
Suppose $C$ is an irreducible curve of degree $(d_E,d_B)$ with $d_E>0$.  Then
\[
g\le d_Ed_B-\frac12\left(3d_B^2-d_B-2
\right),
\]
with equality holding if and only if $C\subset\pi^{-1}(\ell)$ for some $\ell$.
In that case, $C\subset S$ is the zero locus of a section of 
$\mathcal{O_S}(d_EL+d_BE)$.  The moduli space of all curves $C$ (not necessarily
irreducible) given by all $\ell$ and all sections of  $\mathcal{O_S}(d_EL+d_BE)$ is a
$\mathbb{P}^{d_Ed_B-(1/2)(3d_B^2+d_B-4)}$-bundle over $\mathbb{P}^2$.
  \label{prop:ins}
\end{prop}

Part of the proof involves an analysis of certain curves in surfaces of the form
$S=\pi^{-1}(\ell)$ which will be given below. Another
part of the proof is analogous to the proof of the Castelnuovo bound
for the genus of a curve of fixed degree in $\mathbb{P}^3$.  The proof
itself is omitted.

\medskip
We now describe curves in a smooth $S=\pi^{-1}(\ell)$.

Suppose that $S\in|L|$, i.e.\ $S=\pi^{-1}(\ell)$ for some line $\ell\subset\mathbb{P}^2$.  Then for $C\subset S$, $\pi(C)$ is either a point $p$ 
(in which case $C$ is an elliptic 
fiber of $\pi$, hence completely understood), or $\pi(C)=\ell$.  Conversely,
if $C$ is any curve in $X$ for which $\pi(C)=\ell$, then $C\subset S=\pi^{-1}(\ell)$.

We now describe all divisors on $S$ associated to line bundles on $S$
which are restrictions of line bundles on $X$.  We 
denote restrictions to $S$ by a subscript, so that the basic divisor
classes on $S$ are $E_S$ and $L_S$.  For divisors $D$ and $D'$ on $X$ we have
for the intersection on $S$ of their restrictions
\begin{equation}
  \label{eq:intrest}
  D_S\cdot D'_S=D\cdot D'\cdot L,
\end{equation}
where the intersection on the right hand side is taken in $X$.

{}From (\ref{eq:intrest}), (\ref{eq:yukawa}), and $E=H-3L$
we get
\begin{equation}
  E_S^2=-3,\ E_SH_S=0,\ H_S^2=3,\ H_SL_S=1,\ 
E_SL_S=1,
\ L_S^2=0.
  \label{eq:int}
\end{equation}
Let $C\subset S$ be in the class $[C]=d_EL_S+d_BE_S$.  Then by
(\ref{eq:int}) we see that $C\cdot L=d_B$ and $C\cdot H=d_E$, consistent
with our earlier conventions.  The moduli of such $C$ in fixed $S$ is 
a projective space of dimension depending on $(d_E,d_B)$ of a genus
also depending on $(d_E,d_B)$.  We now turn to computing the genus and 
dimension.

For simplicity of exposition, we assume that $S$ is smooth.
The adjunction formula says
\begin{equation}
  \label{eq:adjunction}
  K_S={\cal O}_S(L_S).
\end{equation}
By the projection formula we get
$H^0(S,K_S)=H^0(S,{\cal O}_S(L_S))
=H^0(\ell,{\cal O}_\ell(1))$ and so $p_g=\dim H^0(S,K_S)=2$.  Here we are using
the language of classical algebraic geometry, where $p_g=
\dim H^0(S,K_S)$ is the geometric genus of a complex surface.

It is not hard to see that $q:=h^1(S,\mathcal{O}_S)=0$.
So 
\begin{equation}
  \chi({\cal O}_S)=1-q+p_g=3.
\end{equation}

The genus is given by
\begin{equation}
  g=\frac12 C\left(C+K_S\right)+1=\frac12\left(d_EL_S+d_BE_S\right)
\left(\left(d_E+1\right)L_S+d_BE_S\right)+1
\end{equation}
which simplifies to
\begin{equation}
  \label{eq:glin}
g=    d_Ed_B+\frac12\left(-3d_B^2+d_B+2\right).
\end{equation}

The curves $C\subset S$ are a projective space of dimension $h^0(S,{\cal O}(C))
-1$.  If $d_E>0$ and $d_B>1$ then $C-K_S=C-L_S$ is ample, so by Kodaira
vanishing we have $h^i(S,{\cal O}(C))=0$ for $i>0$.  So we only need to
compute $\chi({\cal O}(C))$:
\begin{equation}
  \chi({\cal O}(C))=\frac12C\left(C-K_S\right)+\chi({\cal O}_S)
=   d_Ed_B-\frac12\left(3d_B^2+d_B-6\right).
  \label{eq:chioc}
\end{equation}
Since the moduli of these $S$ is just $\mathbb{P}^2$ corresponding to the
family of lines in $\mathbb{P}^2$, the moduli space of such curves
$C$ is a $\mathbb{P}^{\chi({\cal O}(C))-1}$-bundle over $\mathbb{P}^2$, with
$\chi({\cal O}(C))$ given by (\ref{eq:chioc}).

\subsubsection{Examples and Computations}
\label{examplesandcomputations}
In this section, we apply the results of the previous section to describe many moduli spaces of curves and compute
the associated Gopakumar-Vafa invariants.  We denote the moduli space of curves of
degree $(d_E,d_B)$ and genus $g$ by $\mathcal{M}^g_{d_E,d_B}$.  If 
$\mathcal{M}^g_{d_E,d_B}$ is smooth and there are no curves of degree $(d_E,d_B)$ and
genus strictly greater than $g$, we have
\begin{equation}
  \label{eq:smoothgv}
  n^g_{d_E,d_B}=\left(-1\right)^{\dim \mathcal{M}^g_{d_E,d_B}}
\chi\left(\mathcal{M}^g_{d_E,d_B}\right),
\end{equation}
a formula that may need correction as discussed in Section~\ref{sec:geomcurves}.

Now let $\mathcal{C}^g_{d_E,d_B}$ be the universal curve, also supposed smooth, and suppose
that there are no other curves of degree $(d_E,d_B)$ and
genus greater than or equal to $g-1$.  Then we have, by \cite{KKV} 
\begin{equation}
  \label{eq:smoothg-1}
  n^{g-1}_{d_E,d_B}=\left(-1\right)^{\dim \mathcal{M}^g_{d_E,d_B}+1}\left(
\chi\left(\mathcal{M}^g_{d_E,d_B}\right)+\left(2g-2\right)
\chi\left(\mathcal{M}^g_{d_E,d_B}\right)\right).
\end{equation}
This formula may also need correction.

\noindent
{\bf Examples:}  We begin with $d_E=0$ and any $d_B>0$.  Then we assume
$d_E>0$ and treat 
$0\le d_B\le 5$ in turn.

\smallskip\noindent
$\mathbf{d_E=0}$.

If $g=(d_B-1)(d_B-2)/2$, then Proposition~\ref{prop:plane}
says that $\mathcal{M}^g_{0,d_B}$ is identical
to the moduli space of plane curves of degree $d_E$ in local $\mathbb{P}^2$.  
In particular, the 
GV invariants $n^g_{0,d_B}$ for $g\le p_a(d_B)$ are identical with the well-known
GV invariants of  local $\mathbb{P}^2$.
This
observation provides the geometric verification of the GV invariants in
the entire first columns of Tables~\ref{genus0GVtable}--\ref{genus8GVtable}.

We will see later that a similar argument can be used to compute or verify 
many of the GV invariants in some of the subsequent columns.

Now let $C$ have degree $(d_E,d_B)$ with $d_E>0$ and we set out to bound the genus. 

\smallskip\noindent
$\mathbf{d_B=0}$.

If $d_B=0$, then by Corollary~\ref{cor:fibers} any curve with $d_B=0$
is a union of fibers. 

It is well known that for an elliptically fibered Calabi-Yau threefold $X$ over
a base $B$, the only nonzero Gopakumar-Vafa invariants in multiples of the
fiber class $f$ are
\begin{equation}
  \label{eq:gvfiber}
  \begin{array}{ccl}
    n^1_{df}&=&\chi(B),\qquad d>0\\
n^0_f&=&-\chi(X).
  \end{array}
\end{equation}
We therefore get
\begin{equation}
  n^0_{d_E,0}=540\ (d_E>0),\qquad n^1_{1,0}=3,
\end{equation}
in agreement with the $d_B=0$ rows of
Tables~\ref{genus0GVtable}--\ref{genus8GVtable}.

\smallskip\noindent $\mathbf{d_B=1}$.  

We compute some Gopakumar-Vafa invariants
for $d_B=1$ by geometry and always find complete agreement with Table~\ref{basedegree1table}
and additional calculations we have done beyond the table.

We put $D=\pi(C)$ and note that $D=\ell$ is a line, by~Lemma~\ref{lem:dproj}.  
Then $C$ has a unique irreducible component $C'$ mapping by $\pi$ onto $D$.  
Since any other components of $C$ map to points, they must be unions of fibers, which have 
$d_B=0$.  Thus $C$ has $d_B=1$.  By Proposition~\ref{prop:plane} we have 
$g(C)\le 0$.  Hence $g(C)=0$ and $C$ is a line in $E$, again by Proposition~\ref{prop:plane}.   
Since $C$ has degree $(0,1)$, we see that the remaining components have degree $(d_E,0)$, 
hence consists of $d_E$ fibers by the $d_B=0$ case.  Thus $C$ is the union of a line in $E$ and $d_E$ 
fibers, and necessarily has genus $d_E$.
>From this description, we see that $\mathcal{M}^{d_E}_{d_E,1}$ is the relative Hilbert 
scheme of $d_E$ points on lines in the plane, a $\mathbb{P}^{d_E}$-bundle over
$\mathbb{P}^2$.  Hence by (\ref{eq:smoothgv}) we have
\begin{equation}
\label{eq:db1}
\begin{array}{ccl}
n^{d_E}_{d_E,1}&=&(-1)^{d_E}3(d_E+1)\\
n^{d_E}_{d,1}&=&0\ (d<d_E).
\end{array}
\end{equation}
The second equation follows since the same description shows that all curves of degree $(d,1)$
have genus $d<d_E$.

These results agree with those presented in the $d_B=1$ row of Tables~\ref{genus0GVtable}--\ref{genus8GVtable}.  
We have also checked these results for much larger values of $d_E$, and always find agreement.

\smallskip\noindent $\mathbf{d_B=2}$.  

We compute some Gopakumar-Vafa invariants
for $d_B=2$ by geometry and always find complete agreement with Table~\ref{basedegree2}
and additional calculations we have done beyond the table.

We $D=\pi(C)$ and note that $D$ has degree 1 or 2.

If $D$ is a line $\ell$, then $C\subset \pi^{-1}(\ell)$.  By Proposition~\ref{prop:ins}  we get 
$g = 2d_E-4$.

If $D=\ell_1\cup\ell_2$ is union of lines, then just as in the $d_B=1$ case we have that $C$ 
contains irreducible components $C_1$ and $C_2$ mapping isomorphically by $\pi$ to $\ell_1$ and 
$\ell_2$ respectively.   Since each $C_i$ has degree $(0,1)$, the remaining components of 
$C$ have degree $(d_E,0)$, hence consist of $d_E$ fibers.   Thus $C$ consists of the 
reducible conic $C_1\cup C_2$ in $E$ with $d_E$ fibers attached.  

If $D$ is an irreducible conic, then $C$ has a unique component $C$ mapping to $C$ with 
degree~1.  Hence by smoothness of $C$ we see that $C$ is isomorphic to $C$.   By 
Proposition~\ref{prop:plane}, $C$ is a conic in $E$.  Since $C$ has degree $(0,2)$, we see just 
as in the last case that $C$ is the union of $C$ and $d_E$ fibers.
The last two cases combine into one moduli space of curves, parametrizing conics in $E$ 
(both smooth and singular) with $d_E$ fibers attached.  This description shows that moduli 
space is isomorphic to the relative Hilbert scheme of $d_E$ points on plane conics.  This 
space is smooth for $d_E\le 3$ \cite{KKV}.

Thus the moduli space of all curves of degree $(d_E,2)$ has two components: the one just 
described as the relative Hilbert scheme of $d_E$ points on 
plane conics, and the other component 
parametrizing curves in a surface $S=\pi^{-1}(\ell)$ as in Proposition~\ref{prop:ins}.

The curves of the first component have genus $d_E$, and the curves of the second component 
have genus $2d_E-4$.  Thus we have $g\le \mathrm{max}(d_E,2d_E-4)$.  
The maximum is realized on the first component for $d_E\le 4$ and on the 
second component for $d_E\ge4$.

So for $d_E\le 3$, all curves of genus $g=d_E$ are unions of plane conics in $E$ with
$d_E$ fibers attached, so this moduli space is isomorphic to
the relative Hilbert scheme of $d_E$ points on plane conics, and this Hilbert scheme is
smooth.  Computing Euler characteristics exactly as in \cite{KKV} or \cite{Choi:2012jz}, we have
by (\ref{eq:smoothgv}) 
\begin{equation}
  \label{eq:conicandfibers}
  n^0_{0,2}=-6,\
  n^1_{1,2}= 15,\
  n^2_{2,2}= -36,\
  n^3_{3,2}= 66
\end{equation}
in agreement with the numbers in Tables~\ref{genus0GVtable}--\ref{genus3GVtable}.  The same argument also
shows that all numbers to the left of these numbers are zero.

For $d_E\ge 5$, the above discussion shows that curves of genus $g=2d_E-4$ are all contained in surfaces $S=\pi^{-1}(\ell)$ and
are described by Proposition~\ref{prop:ins}.  So $\mathcal{M}^{2d_E-4}_{d_E,2}$ is a $\mathbb{P}^{2d_E-5}$ bundle
over $\mathbb{P}^2$, hence by (\ref{eq:smoothgv}) we get  
\begin{equation}
  n^{2d_E-4}_{d_E,2}=-3\left(2d_E-4\right)
  \label{eq:db2}
\end{equation}
and 
\begin{equation}
  n^{2d_E-4}_{d,2}=0,\qquad d<d_E,
  \label{eq:db20}
\end{equation}
in agreement with the $d_B=2$ lines of  Tables~\ref{genus6GVtable}--\ref{genus8GVtable}
for $d_E\ge5$.  We have also checked these results for much larger values of $d_E$
and always find agreement.

Our last set of checks come from the application of (\ref{eq:smoothg-1}) to curves
of type $((1,2))$.  From the genus formula $\mathrm{max}(d_E,2d_E-4)$ we see that the
hypotheses of (\ref{eq:smoothg-1}) hold if $d_E<2d_E-5$, so that $d_E\ge6$.  We need
to describe the universal curve $\mathcal{C}^{2d_E-4}_{d_E,2}$.  Given a point of
the universal curve, which we denote by $(C,p)$, we have a natural map
$\mathcal{C}^{2d_E-4}_{d_E,2}\to X$ obtained by forgetting $C$.  We claim that the fiber
is a $\mathbb{P}^{2d_E-6}$-bundle over $\mathbb{P}^1$, hence the universal curve is
smooth.  

To see this, note that the line $\ell$ for which $C\subset S=\pi^{-1}(\ell)$ must contain
the point $q=\pi(p)$.  Therefore the lines $\ell$ used to describe curves in 
the fiber over $p$ form a $\mathbb{P}^1$
rather than a $\mathbb{P}^2$.  The fiber is itself fibered over this $\mathbb{P}^1$, and
it only remains to find the fiber of this last fibration.  
Now fixing $\ell$, hence fixing $S$, the possible curves
$C$ are those corresponding to sections of $H^0(S,\mathcal{O}(d_EL_S+d_BE_S))$ which vanish at $p$,
which is a hyperplane in  $H^0(S,\mathcal{O}(d_EL_S+d_BE_S))$.  So the fiber is a 
$\mathbb{P}^{2d_E-6}$ rather than a $\mathbb{P}^{2d_E-5}$, as claimed.

We can now apply (\ref{eq:smoothg-1}) to compute
\[
n^{2d_E-5}_{d_E,2}=(-540)(2)(2d_E-5)+(4d_E-10)3(2d_E-4),\qquad d_E\ge6,
\]
which simplifies to
\begin{equation}
  \label{eq:db2g-1}
n^{2d_E-5}_{d_E,2}=12(2d_B-5)(d_B-92),\qquad d_E\ge6,
\end{equation}
which agrees with the results in Table~\ref{basedegree2} as well as larger values of
$d_E$.  Furthermore, all numbers to the left of these numbers are zero, which also checks.

\bigskip
The general procedure should now be clear.  We let $D=\pi(C)$ and we study the cases where 
$D$ has irreducible components of each possible degree and multiplicity.  In other words
the degree $d_B$ curve $D$ can split up into components $D_j$ of degree $d_j$ and multiplicity $m_j$.  So the
possible components of moduli can be indexed by partitions of $d_B$ into 
unordered lists of pairs $((d_j,m_j))$ with repetition allowed
satisfying 
\begin{equation}
  \label{eq:partd}
  d_B=\sum_j m_jd_j
\end{equation}
Correspondingly, $C$ splits up into a union of components $C_j$, mapping by $\pi$
to $D_j$ with multiplicity $m_j$, possibly together with fibers.  We investigate the
maximum genus of each $C_j$.  The genus of $C$ can then be maximized by configuring
the curves $C_j$ and any fibers so that they intersect in as many points (including
multiplicity) as possible.

For $d_B=1$, only the partition $((1,1))$ is possible.  For $d_B=2$ we could have
either $((2,1)),\ ((1,2))$, or $((1,1),(1,1))$, and all cases occurred in our discussion
above.

The analysis gets more intricate for each degree. Rather than 
carry out this program for $d_E\le 5$, we content ourselves with presenting some
Gopakumar-Vafa invariants that  can be computed by these methods, giving the
corresponding $((d_j,m_j))$, and matching to the tables. 

For any $d_B\ge2$ some strata combine naturally, as we already  saw for $d_B=2$.
Let $d_B=\sum_{j=1}^k$ be a partition of $d_B$ and consider the list
$((d_j,1)_{j=1}^k)$.  These strata combine to describe the Hilbert scheme of $d_E$
points on plane curves of degree $d_B$.  The stratification is just an artifact of
our method, corresponding  to splitting up the moduli space
according to the degrees of the components of the plane curve. 

\smallskip\noindent $\mathbf{d_B=3}$.

We compute some Gopakumar-Vafa invariants
for $d_B=3$ by geometry and always find complete agreement with Table~\ref{basedegree3}
and additional calculations we have done beyond the table.

In this case, the possible lists are $((1,1),(1,1),(1,1)),\ 
((3,1)),\ ((1,3)), ((2,1),(1,1)),\
((1,2)$, or $(1,1))$, 

\smallskip\noindent
$((3,1)),\  ((2,1),(1,1))$, and $((1,1),(1,1),(1,1))$ combine to give 
plane cubics union $d_E$ fibers, with moduli space the relative Hilbert scheme of
$d_E$ points on plane cubics, which are curves of genus $d_E+1$.

\smallskip\noindent
$((1,3))$ are curves in a surface of the form $S=\pi^{-1}(\ell)$, which have genus
$3d_E-11$ by Proposition~\ref{prop:ins}.

\smallskip\noindent
$((1,2),(1,1))$ is the union of a curve in a surface of the form $S=\pi^{-1}(\ell)$,
a line in $E$, and possibly some fibers.  It is not hard to see by extending the analysis
below that the maximum genus is attained when there are no additional fibers.  So
$C$ is the union of a curve $C_{d_E,2}$ in $S$ and a line $C_{0,1}\subset E$ (which has 
degree $(0,1)$).  Here and in the sequel, a subscript denotes the degree of a curve. 
We have already seen that $C_{d_E,2}$ has genus $2d_E-4$.

The genus of $C$ is therefore $(2d_E-4)+m-1$, where $m$ is the number of intersection
points $Z$ of $C_{d_E,2}$ and $\ell$.  We claim that either $m=1$ or we are actually
in the $((1,3))$ case. Since $C_{0,1}\subset E$ we see that  $Z\subset 
C_{d_E,2}\cap E$.  Since $S\cap E$ is a line,  it follows that if $m\ge2$, then $C_{0,1}$
must be equal to that line.  In that case $C\subset S$, and we are really in the $((1,3))$
case.  So $m=1$ and $C$ has genus $2d_E-4$.  

Combining the three cases, we have $g\le \mathrm{max}(d_E+1,3d_E-11,2d_E-4)$.

The maximum genus is attained in the first case  for $d_E\le 5$, in the second case for
$d_E\ge 7$, and in the third case for $5\le d_E\le 7$. 

If $d_E\le 4$, then the maximum genus $d_E+1$ can only be attained in the first case,
and $\mathcal{M}^{d_E+1}_{d_E,3}$ 
 is the relative Hilbert scheme of $d_E$ points on
plane cubics, which is smooth by \cite{KKV}.  So we compute from the relative 
Hilbert scheme
\begin{equation}
  \label{eq:cubicandfibers}
  n^1_{0,3}=-10,\ n^2_{1,3}=27,\ n^3_{2,3}=-72,\ n^4_{3,3}=154,\ n^5_{4,3}=-306.
\end{equation}
This agrees with Tables~\ref{genus1GVtable}--\ref{genus5GVtable}.   In addition, our
now-familiar argument shows that all numbers to the left of these numbers in the 
tables are zero, in agreement with the tables.

If $d_E\ge 8$, then the maximum genus $3d_E-11$ can only be attained in the second case
$((1,3))$,
and $\mathcal{M}^{3d_E-11}_{d_E,3}$ 
is a $\mathbb{P}^{3d_E-13}$-bundle over
$\mathbb{P}^2$ by Proposition~\ref{prop:ins}.  Hence, for $d_E\ge 8$ we have
\begin{equation}
\label{eq:db3}
  n^{3d_E-11}_{d_E,3}=(-1)^{d_E+1}9(d_E-4).
\end{equation}
For these $d_E$ we also  have $n^{3d_E-11}_{d,3}=0$ for $d<d_E$.  We have checked these
results for large values of $d_E$ and always find agreement.

For $d_E=6$, the maximum genus $2d_E-4=8$ can only be realized in the third case
$((1,2),(1,1))$.  By the discussion above, we have $C=C_{6,2}\cup C_{0,1}$, where $C_{6,2}$
lies in a surface $S=\pi^{-1}(\ell)$ and the line $C_{0,1}$ intersects $C_{6,2}$ in
exactly one point.  But $C_{6,2}\cdot E=0$ by (\ref{eq:inte}).  On the other
hand, $C_{6,2}$
cannot be disjoint from $E$, since it intersects $C_{0,1}\subset E$ nontrivially.
Therefore $C_{6,2}$ has a component which lies entirely in $E$, necessary a line or a conic.
If it contains a conic plane curve, 
then since it also contains the line $C_{0,1}$ then it contains a
(reducible) cubic plane curve.  So we are really in the case of a cubic with $d_E$ fibers
already considered.  So $C_{6,2}$ contains a line.  The remaining components are a curve 
$C_{6,1}$, which by the $d_B=1$ case are a line and 6 fibers.  If the two lines  are distinct,
then they form a conic and we are in the previous case.  Otherwise the line has 
multiplicity~2, and the curve $C_{6,2}$ is the union of the line
$E_S$ with multiplicity 2 (in $S$) and  6 fibers.  The line $C_{0,1}$ is arbitrary.
So $\mathcal{M}^8_{6,3}$ is the product of $\mathbb{P}^2$ and
the relative Hilbert scheme of 6 points in
lines in $\mathbb{P}^2$.  Since this relative Hilbert scheme is a $\mathbb{P}^6$-bundle
over $\mathbb{P}^2$,  this gives
\begin{equation}
n^8_{6,3}=63,
\end{equation} 
in agreement with Table~\ref{genus8GVtable}.  Furthermore, all of the numbers to the left
are zero, again in agreement.

We can also apply (\ref{eq:smoothg-1}) for curves of type $((1,3))$, which have genus
$3d_E-11$.  The hypotheses hold if $d_E\ge9$ and then we get
\[
n^{3d_E-12}_{d_E,3}=\left(-1\right)^{d_E}\left((-540)(2)(3d_E-13)+(6d_E-24)(3)(3d_E-12)
\right),
\]
which simplifies to
\begin{equation}
\label{eq:db3g-1}
n^{3d_E-12}_{d_E,3}=(-1)^{d_E}(54)(d^2-68d+276),\qquad d_E\ge9,
\end{equation}
which agrees with the results in Table~\ref{basedegree3} as well as larger values of
$d_E$.  Furthermore, all numbers to the left of these numbers are zero, which also checks.

\smallskip\noindent
$\mathbf{d_B=4}$.  

We compute some Gopakumar-Vafa invariants
for $d_B=4$ by geometry and always find complete agreement with Table~\ref{basedegree4}
and additional calculations we have done beyond the table.
By now we have provided enough examples that we can simply present
our results and the reader will be able to check details.

\begin{sloppypar}
Combining the cases $((4,1)),\ ((3,1),(1,1)),\ ((2,1),(2,1)),\ ((2,1),(1,1),(1,1))$,
and $((1,1),(1,1),(1,1),(1,1))$, we consider plane quartic curves in $E$ with $d_E$
fibers attached.  These curves have genus $d_E+3$.  If $d_E\le6$, any curve of degree
$(d_E,4)$ and genus $d_E+3$ can be seen to be of this type, and furthermore there
are no  curves of that degree and higher genus.  We conclude that 
$\mathcal{M}^{d_E+3}_{d_E,4}$ is the relative Hilbert scheme of $d_E$ points on plane
quartic curves, a $\mathbb{P}^{14-d_E}$ bundle over
$\mathrm{Hilb}^{d_E}(\mathbb{P}^2)$, and is smooth for $d_E\le5$ \cite{KKV}.  
We therefore get 
\end{sloppypar}
\begin{equation}
  \label{eq:d4andf}
  n^{d_E+3}_{d_E,4}=(-1)^{d_E}(15-d_E)\chi(\mathrm{Hilb}^{d_E}(\mathbb{P}^2)),
\end{equation}
hence by computing
Euler characteristics
\begin{equation}
  \label{eq:quarticandfibers}
 n^3_{0,4}=15,\ n^4_{1,4}=-42,\ n^5_{2,4}=117,\ n^6_{3,4}=-264,\ n^7_{4,4}=561,\
n^8_{5,4}=-1080.
\end{equation}
The same argument confirms in the usual way that all numbers to the left of these are zero,
in agreement with the tables.

For curves of type $((1,4))$, i.e.\ curves contained in a surface $S=\pi^{-1}(\ell)$,
we have $g=4d_E-21$, and $\mathcal{M}^{4d_E-21}_{d_E,4}$ is a $\mathbb{P}^{4d_E-24}$-bundle 
over
$\mathbb{P}^2$ by Proposition~\ref{prop:ins}.  
If $d_E\ge11$, then any curve of this degree and genus is of this type, and there are no
curves of this degree and higher genus.  This gives 
\begin{equation}
  \label{eq:gv4ins}
n^{4d_E-21}_{d_E,4}=3(4d_E-23)=12d_E-69,\ d_E\ge11
\end{equation}
in complete agreement with our calculations.  Also $n^{4d_E-21}_{d,4}=0$ for $d<d_E$.

For $d_E=7$ we have curves of type $((1,2),(2,1))$, which consist of a curve $C_{7,2}$
of degree $(7,2)$ union a conic $C_{0,2}\subset E$.  
The curve $C_{7,2}$
lies in a surface $S=\pi^{-1}(\ell)$ and has genus $2d_E-4=10$.  Since $C_{0,2}$
is to meet $C_{7,2}$ in two points but  $C_{7,2}\cdot E=1$, we see that $C_{7,2}$ 
must contain  a line in 
$E$.  Repeating an argument from the $d_B=3$ case, we further see that  $C_{7,2}$ must
contain a line doubled in $S$ together with $d_E$ fibers.  
Our geometric description shows that the moduli space
of $C_{7,2}$ is the relative Hilbert scheme of 7~points on lines, a $\mathbb{P}^7$-bundle
over $\mathbb{P}^2$.    There are no curves of degree $(7,2)$ and higher genus,
so we can now multiply be the $\mathbb{P}^5$ moduli of $C_{0,2}$ and compute the Euler
characteristic, giving
\begin{equation}
  \label{eq:gv74}
  n^{16}_{7,4}=144,
\end{equation}
in agreement with our other calculations.  In addition we have $n^{16}_{d,4}=0$
for $d<7$, also in agreement.

We can also apply (\ref{eq:smoothg-1}) for curves of type $((1,4))$, which have genus
$4d_E-21$.  The hypotheses hold if $d_E\ge12$ and then we get
\[
n^{4d_E-22}_{d_E,4}=-\left((-540)(2)(4d_E-24)+(8d_E-44)(3)(4d_E-23)
\right),
\]
which simplifies to
\begin{equation}
\label{eq:db4g-1}
n^{4d_E-22}_{d_E,4}=-54(d^2-68d+276),\qquad d_E\ge12,
\end{equation}
which agrees with the results in Table~\ref{basedegree4} as well as larger values of
$d_E$.  Furthermore, all numbers to the left of these numbers are zero, which also checks.

\smallskip\noindent
$\mathbf{d_B=5}$.

We compute some Gopakumar-Vafa invariants
for $d_B=5$ by geometry and always find complete agreement with Table~\ref{basedegree5}
and additional calculations we have done beyond the table.

We first consider the curves which are unions of plane quintics and $d_E$ fibers, which
have genus $d_E+6$.
We also compute that for $d_E\le 7$, all curves of degree $(d_E,5)$ and genus $d_E+6$
are of this type.
Therefore $\mathcal{M}^{d_E+6}_{d_E,5}$ is the relative
Hilbert scheme of $d_E$ points on plane quintics, which is smooth for
$d_E\le 6$.  We also compute that there are no curves of that degree and higher genus.
Since $\mathcal{M}^{d_E+6}_{d_E,5}$ is a $\mathbb{P}^{20-d_E}$-bundle over
$\mathrm{Hilb}^{d_E}(\mathbb{P}^2)$, we get
\begin{equation}
  \label{eq:d5andf}
  n^{d_E+6}_{d_E,5}=(-1)^{d_E}(21-d_E)\chi(\mathrm{Hilb}^{d_E}(\mathbb{P}^2))\qquad d_E\le6,
\end{equation}
or
\begin{equation}
\label{eq:quinticandfibers}
n^6_{0,5}=21,\ n^7_{1,5}=-60,\ n^8_{2,5}=171,\ n^9_{3,5}=-396,\ n^{10}_{4,5}=867,\
n^{11}_{5,5}=-1728,\ n^{12}_{6,5}=3315.
\end{equation}
in complete agreement with our calculations.  Also the Gopakumar-Vafa invariants to the
left of these numbers are all zero, also in agreement.

For curves of type $((1,5))$, i.e.\ those contained in a surface $S=\pi^{-1}(\ell)$ we have
$g=5d_E-34$ and moduli space a $\mathbb{P}^{5d_E-38}$-bundle over $\mathbb{P}^2$.
For $d_E\ge14$ we compute that all curves of degree $(d_E,5)$ are of this type, and that there
are no curves of higher genus.  This gives
\begin{equation}
\label{eq:db5ins}
n^{5d_E-34}_{5,d_E}=(-1)^{d_E}3(5d_E-37)\qquad d_E\ge14.
\end{equation}
in complete agreement with our calculations.  Also the Gopakumar-Vafa invariants to the
left of these numbers are all zero, also in agreement.

We can also apply (\ref{eq:smoothg-1}) for curves of type $((1,5))$, which have genus
$5d_E-34$.  The hypotheses hold if $d_E\ge15$ and then we get
\[
n^{5d_E-35}_{d_E,5}=\left(-1\right)^{d_E+1}\left((-540)(2)(5d_E-38)+(10d_E-70)(3)(5d_E-37)
\right),
\]
which simplifies to
\begin{equation}
\label{eq:db5g-1}
n^{5d_E-35}_{d_E,5}=(-1)^{d_E+1}(30)(5d_B^2-252d_B+1627),\qquad d_E\ge15,
\end{equation}
which agrees with the results in Table~\ref{basedegree4} as well as larger values of
$d_E$.  Furthermore, all numbers to the left of these numbers are zero, which also checks.


\vspace{0.5in} {\leftline {\bf Acknowledgments}}

We thanks Babak Haghighat, Guglielmo Lockhart and Cumrun Vafa for discussions about 
related structures in M- and E-string theories and Martin Westerholt-Raum and 
Don Zagier for crucial comments on Siegel modular- and weak Jacobi forms. We like to 
thanks Gaetan Borot, Hans Jockers and Marcos Mari\~no for comments on the draft. 
This work started when MH was affiliated with Kavli Institute for the Physics and Mathematics of the Universe 
(Kavli IPMU). He thanks Kavli IPMU for supports. MH also thanks the BCTP, the HCM  and the 
organizers of ``The 2nd Workshop on Developments in M-theory" in Korea for hospitality. 
MH is supported by the ``Young Thousand People" plan by the Central Organization 
Department in China, and by the Natural Science Foundation of China.
S.K. is supported by NSF DMS-12-01089. A.K. thanks  for support by 
KL 2271/1-1 and NSF DMS-11-01089.

\appendix 
\section{Gopakumar-Vafa invariants}
\label{appendixA}
The involution symmetry of the topological strings on elliptic Calabi-Yau manifolds 
restricts the number of the coefficients in the ambiguity to roughly one fourth. 
This allows to extract the higher genus invariants to sufficient genus $g\le 8$ to 
test the conjectures about the even weak Jacobi-Forms in reasonable detail\footnote{As 
we explained it is in principle possible to evaluate them to genus 189.   
The genus $g=0$ agree with~\cite{Hosono:1993}, the $g\leq 1$ invariants 
with~\cite{Candelas:1994hw}  and the invariants $g\leq 3$ invariants 
with~\cite{Alim:2012}.}

\begin{table}[!htbp]
\begin{center} {\scriptsize
 \begin{tabular} {|c|c|c|c|c|c|c|c|} \hline $d_B \backslash d_E$  & 0 & 1 & 2 & 3 & 4 & 5 & 6 \\  \hline 0 &  & 540 & 540 & 540 & 540 & 540 & 540 \\  \hline 1 & 3 & -1080 & 143370 & 204071184 & 21772947555 & 1076518252152 & 33381348217290 \\  \hline 2 & -6 & 2700 & -574560 & 74810520 & -49933059660 & 7772494870800 & 31128163315047072 \\  \hline 3 & 27 & -17280 & 5051970 & -913383000 & 224108858700 & -42712135606368 & 4047949393968960 \\  \hline 4 & -192 & 154440 & -57879900 & 13593850920 & -2953943334360 & 603778002921828 & -90433961251273800 \\  \hline 5 & 1695 & -1640520 & 751684050 & -218032516800 & 51350781706785 & -11035406089270080 & 2000248139674298880 \\  \hline \end{tabular}
 }
\caption{The GV invariants $n^g_{(d_E,d_B)}$ for genus $g=0$ for the elliptic Calabi-Yau $X(1,1,1,6,9)$}  
\label{genus0GVtable}
 \end{center}

\begin{center} {\scriptsize
 \begin{tabular} {|c|c|c|c|c|c|c|c|} \hline $d_B \backslash d_E$  
   & 0 & 1 & 2 & 3 & 4 & 5 & 6 \\  \hline 
 0 &  & 3 & 3 & 3 & 3 & 3 & 3 \\  \hline 
 1 & 0 & -6 & 2142 & -280284 & -408993990 & -44771454090 & -2285308753398 \\  \hline 
 2 & 0 & 15 & -8574 & 2126358 & 521856996 & 1122213103092 & 879831736511916 \\  \hline 
 3 & -10 & 4764 & -1079298 & 152278986 & -16704086880 & -3328467399468 & 1252978673852088 \\  \hline 
 4 & 231 & -154662 & 48907815 & -9759419622 & 1591062421074 & -186415241060547 & 8624795298947118 \\  \hline 
 5 & -4452 & 3762246 & -1510850250 & 385304916960 & -76672173887766 & 12768215950604064 & -1663415916220743876 \\  \hline \end{tabular}
 }
\caption{The GV invariants $n^g_{(d_E,d_B)}$ for genus $g=1$ for the elliptic Calabi-Yau $X(1,1,1,6,9)$}  
\label{genus1GVtable}
 \end{center}

\end{table}
\begin{table}[!htbp]

\begin{center} {\scriptsize
 \begin{tabular} {|c|c|c|c|c|c|c|c|} \hline $d_B \backslash d_E$  
   & 0 & 1 & 2 & 3 & 4 & 5 & 6 \\  \hline 
 0 &  & 0 & 0 & 0 & 0 & 0 & 0 \\  \hline 
 1 & 0 & 0 & 9 & -3192 & 412965 & 614459160 & 68590330119 \\  \hline 2 & 0 & 0 & -36 & 20826 & -5904756 & -47646003780 & -80065270602672 \\  \hline 3 & 0 & 27 & -16884 & 4768830 & -818096436 & 288137120463 & 67873415627151 \\  \hline 4 & -102 & 57456 & -15452514 & 2632083714 & -320511624876 & 18550698291252 & 780000198300540 \\  \hline 5 & 5430 & -4032288 & 1430896428 & -323858122812 & 55058565096630 & -7249216518163620 & 691264676523200805 \\  \hline \end{tabular}
}
\caption{The GV invariants $n^2_{(d_E,d_B)}$ for the elliptic Calabi-Yau $X(1,1,1,6,9)$. 
Note that $n^{(0,d_2)}_{g>1}=0$, $\forall d_2 \in \mathbb{N}$. We therefore omit the $(0,d_2)$ line below.}  
\label{genus2GVtable}
 \end{center}

\begin{center} {\scriptsize
 \begin{tabular} {|c|c|c|c|c|c|c|c|} \hline $d_B \backslash d_E$  & 0 & 1 & 2 & 3 & 4 & 5 & 6 \\  \hline 
  1 & 0 & 0 & 0 & -12 & 4230 & -541440 & -820457286 \\  \hline 2 & 0 & 0 & 0 & 66 & -45729 & 627574428 & 3776946955338 \\  \hline 3 & 0 & 0 & -72 & 48036 & -14756490 & 297044064 & -7900517344212 \\  \hline 4 & 15 & -7236 & 1638918 & -226431351 & 20419274259 & -719284158099 & 236091664016826 \\  \hline 5 & -3672 & 2417742 & -764921214 & 154856849136 & -22866882491772 & 2493418732350750 & -194361733345447458 \\  \hline \end{tabular}
 }
\caption{The GV invariants $n^3_{(d_E,d_B)}$ for the elliptic Calabi-Yau $X(1,1,1,6,9)$}  
\label{genus3GVtable}
 \end{center}

\end{table}
\begin{table}[!htbp]
 
 \begin{center} {\scriptsize
  \begin{tabular} {|c|c|c|c|c|c|c|c|} \hline $d_B \backslash d_E$  
  & 0 & 1 & 2 & 3 & 4 & 5 & 6 \\  \hline  
  1 & 0 & 0 & 0 & 0 & 15 & -5256 & 665745 \\  \hline 2 & 0 & 0 & 0 & 0 & -132 & -453960 & -95306132778 \\  \hline 3 & 0 & 0 & 0 & 154 & -110574 & 38259441 & 218140445904 \\  \hline 4 & 0 & -42 & 26946 & -7824888 & 1386011568 & -172919782116 & -4345528029372 \\  \hline 5 & 1386 & -819123 & 232934157 & -42321589218 & 5500907292240 & -520718843839590 & 38245592568676608 \\  \hline \end{tabular}
   }
\caption{The GV invariants $n^4_{(d_E,d_B)}$ for the elliptic Calabi-Yau $X(1,1,1,6,9)$}  
\label{genus4GVtable}
 \end{center}
 
 \begin{center} {\scriptsize
  \begin{tabular} {|c|c|c|c|c|c|c|c|} \hline $d_B \backslash d_E$  
    & 0 & 1 & 2 & 3 & 4 & 5 & 6 \\  \hline 
  1 & 0 & 0 & 0 & 0 & 0 & -18 & 6270 \\  \hline 
  2 & 0 & 0 & 0 & 0 & 0 & -5031 & 1028427996 \\  \hline 
  3 & 0 & 0 & 0 & 0 & -306 & 247014 & -2562122952 \\  \hline 
  4 & 0 & 0 & 117 & -81225 & 26211942 & -5223900087 & 1263109811373 \\  \hline 
  5 & -270 & 144414 & -36870264 & 5929743618 & -665294451264 & 53375661928620 & -3651471177372864 \\  \hline \end{tabular}
   }
\caption{The GV invariants $n^5_{(d_E,d_B)}$ for the elliptic Calabi-Yau $X(1,1,1,6,9)$}  
\label{genus5GVtable}
 \end{center}

 \begin{center} {\scriptsize
 \begin{tabular} {|c|c|c|c|c|c|c|c|c|} \hline $d_B \backslash d_E$  
   & 0 & 1 & 2 & 3 & 4 & 5 & 6 & 7 \\  \hline 
 1 & 0 & 0 & 0 & 0 & 0 & 0 & 21 & -7272 \\  \hline 
 2 & 0 & 0 & 0 & 0 & 0 & -18 & -771642 & -147864402162 \\  \hline 
 3 & 0 & 0 & 0 & 0 & 0 & 612 & 1401468 & 502063861662 \\  \hline 
 4 & 0 & 0 & 0 & -264 & 200430 & -70438068 & 9510828972 & -29413672557570 \\  \hline 
 5 & 21 & -9972 & 2156373 & -268703481 & 18682746903 & -182455706016 & -81900631565910 & 22514515679407491 \\  \hline \end{tabular}
  }
\caption{The GV invariants $n^6_{(d_E,d_B)}$ for the elliptic Calabi-Yau $X(1,1,1,6,9)$}  
\label{genus6GVtable}
 \end{center}
\end{table}
\begin{table}[!htbp]
 \begin{center} {\scriptsize
 \begin{tabular} {|c|c|c|c|c|c|c|c|c|} \hline \!\!\!\!$d_B\! \backslash\! d_E$ \!\!\!\!\!\!\!\!
   & 0 & 1 & 2 & 3 & 4 & 5 & 6 & 7 \\  \hline 
 1 & 0 & 0 & 0 & 0 & 0 & 0 & 0 & -24 \\  \hline 
 2 & 0 & 0 & 0 & 0 & 0 & 0 & -7224 & 1443561648 \\  \hline 
 3 & 0 & 0 & 0 & 0 & 0 & 0 &  17386& -4962183570  \\  \hline  
 4 & 0 & 0 & 0 & 0 &561&  -447903& 170978160 & 605964021294 \\  \hline 
 5 & 0 &-60&  38340& -10994520 & 1895073858 & -217773585972 & 16072935664050 & -1937578925283840 \\  \hline 
 6 & \!\!\!\!\! 27538\!\!\!\!& \!\!\!\!\! -16386600\!\!\!\! & \!\!\!\!\! 4710791727\!\!\!\! & \!\!\!\! -868872423987\!\!\!\! 
 & \!\!\!\!\! 115076024047737\!\!\!\! & \!\!\!\!\! -11554540079426667\!\!\!\! &\!\!\!\!\! 915758222342784613\!\!\!\! 
 &\!\!\!\!\!  -59159328867116232828\!\!\!\! \\  \hline 
 \end{tabular}
  }

\caption{The GV invariants $n^7_{(d_E,d_B)}$ for $X(1,1,1,6,9)$}  
\label{genus7GVtable}
 \end{center}
 
 \begin{center} {\scriptsize
 \begin{tabular} {|c|c|c|c|c|c|c|c|c|c|} \hline \!\!\!\!$d_B\! \backslash\! d_E$ \!\!\!\!\!\!\!\!
   & 0 & 1 & 2 & 3 & 4 & 5 & 6  & 7 &8\\  \hline 
 1 & 0 & 0 & 0 & 0 & 0 & 0 & 0  & 0 &27\\  \hline 
 2 & 0 & 0 & 0 & 0 & 0 & 0 &-24 & -995490& -203754011670 \\  \hline 
 3 & 0 & 0 & 0 & 0 & 0 & 0 & 63 & 3396663& 11118565777779  \\  \hline  
 4 & 0 & 0 & 0 & 0 & 0 &-1080&951204&-6600292956&-74890630203552  \\  \hline 
 5 & 0 & 0 & \!\!\!\!\! 171& \!\!\!\!\! -119415& \!\!\!\!\! 38611944\!\!\!\!& \!\!\!\!\! -7672460076\!\!\!\!& \!\!\!\!\! 
 1056853387755\!\!\!\!& \!\!\!\!\! -49874149196514\!\!\!\!&\!\!\!\!\!  29707605109699254 \!\!\!\! \\ \hline
 6 &\!\!\!\!-5310\!\!\!\!& \!\!\!\!2949516\!\!\!\!& \!\!\!\!-785916540\!\!\!\!&\!\!\!\! $132969\cdot 10^{7}$&\!\!\!\!  $-159045\cdot 10^{8}$
 &$14105 \cdot 10^{11}$ &  $-952701\cdot 10^{11}$ &  $471505\cdot 10^{13}$ & $-112148\cdot 10^{15} $ \\ \hline
 \end{tabular}
  }
\caption{The GV invariants $n^8_{(d_E,d_B)}$ for $X(1,1,1,6,9)$. For last 6 numbers we 
give only 6 significant digits.}  
\label{genus8GVtable}
 \end{center}
\end{table}

\begin{sidewaystable}[!htbp]
\begin{center} {\tiny
 \begin{tabular} {|c|c|c|c|c|c|c|c|c|c|c|c|c|c|c|c|} \hline $g\backslash d_E$  
 & 0& 1& 2& 3& 4& 5& 6& 7& 8& 9& 10& 11& 12 \\  \hline 
0&-192& 154440& -57879900& 13(7)20& -29(9)60& 60(11)28& -90(13)00& 50(16)00& -48(18)00& -82(20)80& -10(25)00& -13(27)20& 67(30)72\\ 
1&231& -154662& 48907815& -97(6)22&15(9)74& -18(11)47& 86(12)18& 20(15)20& -45(18)44& 29(20)28& -37(24)96& 23(26)82& 68(30)88\\ 
2&-102& 57456& -15452514& 26(6)14& -32(8)76& 18(10)52& 78(11)40& -25(14)44& -15(17)96& 27(20)32& -95(22)96& 44(26)98& 22(30)64\\ 
3&15& -7236& 1638918& -226431351& 20(7)59& -71(8)99& 23(11)26& -65(13)34& 46(16)88& 79(18)02& 19(21)16& 16(25)73& 16(29)30\\ 
4&0& -42& 26946& -7824888& 13(6)68& -17(8)16& -43(9)72& -32(12)50& -32(15)74& -41(18)06& -30(21)82& -33(24)50& -26(27)20\\ 
5&0& 0& 117& -81225& 26211942& -52(6)87& 12(9)73& 45(11)41& 38(14)26& 69(17)56& 80(20)68& 10(24)95& 12(27)02\\ 
6&0& 0& 0& -264& 200430& -70438068& 95(6)72& -29(10)70& -34(13)96& -89(16)70& -15(20)66& -25(23)86& -39(26)80\\ 
7&0& 0& 0& 0& 561& -447903& 170978160& 60(8)94& 20(12)96& 89(15)66& 23(19)94& 55(22)47& 10(26)89\\ 
8&0& 0& 0& 0& 0& -1080& 951204& -66(6)56& -74(10)52& -68(14)40& -29(18)10& -10(22)00& -27(25)42\\ 
9&0& 0& 0& 0& 0& 0& 2136& 2312838& 15(9)21& 39(13)66& 30(17)24& 16(21)56& 61(24)68\\ 
10&0& 0& 0& 0& 0& 0& 0& 38160& -13(7)34& -16(12)88& -25(16)96& -22(20)48&-12(24)34\\ 
11&0& 0& 0& 0& 0& 0& 0& 144& 8502057& 47(10)15& 16(15)89& 26(19)83& 21(23)35\\ 
12&0& 0& 0& 0& 0& 0& 0& 0& 83124& -83(8)68& -85(13)84& -26(18)64& -33(22)98\\ 
13&0& 0& 0& 0& 0& 0& 0& 0& 279& 68(6)48& 33(12)40& 22(17)46& 46(21)01\\ 
14&0& 0& 0& 0& 0& 0& 0& 0& 0& -3921582& -89(10)36& -16(16)08& -56(20)92\\ 
15&0& 0& 0& 0& 0& 0& 0& 0& 0& -37968& 14(9)23& 91(14)86& 59(19)08\\ 
16&0& 0& 0& 0& 0& 0& 0& 0& 0& -126& -11(7)74& -41(13)20& -54(18)60\\ 
17&0& 0& 0& 0& 0& 0& 0& 0& 0& 0& 6330231& 14(12)70& 42(17)12\\ 
18&0& 0& 0& 0& 0& 0& 0& 0& 0& 0& 61632& -34(10)28& -27(16)82\\ 
19&0& 0& 0& 0& 0& 0& 0& 0& 0& 0& 204& 52(8)13& 14(15)85\\ 
20&0& 0& 0& 0& 0& 0& 0& 0& 0& 0& 0&-37(6)20& -62(13)82\\ 
21&0& 0& 0& 0& 0& 0& 0& 0& 0& 0& 0& 1819827& 20(12)15\\ 
22&0& 0& 0& 0& 0& 0& 0& 0& 0& 0& 0& 18648& -46(10)04\\ 
23&0& 0& 0&0& 0& 0& 0& 0& 0& 0& 0& 63& 66(8)95\\ 
24&0& 0& 0& 0& 0& 0& 0& 0& 0& 0& 0& 0& -46(6)82\\ 
25&0& 0& 0& 0& 0& 0& 0& 0& 0& 0& 0&0& 2097267\\ 
26&0& 0& 0& 0& 0& 0& 0& 0& 0& 0&0& 0& 22020\\ 
27&0& 0& 0& 0& 0& 0& 0& 0& 0& 0& 0& 0& 75\\ \hline    
\end{tabular}}
\caption{Some BPS invariants for  $n_{(d_E,4)}^{g}$. To save space we only give 
the first and the last two significant digits and the number of omitted digits in brackets.}  
\label{basedegree4}
 \end{center}
\end{sidewaystable}

\begin{sidewaystable}[!htbp]
\begin{center} {\tiny
 \begin{tabular} {|c|c|c|c|c|c|c|c|c|c|c|c|c|c|c|c|c|} \hline $g\backslash d_E$  
& 0& 1& 2& 3& 4& 5& 6& 7& 8& 9& 10& 11& 12& 13& 14& 15  \\  \hline 
0&1695& -1640520& 751684050& -21(8)00& 51(10)85& -11(13)80& 20(15)80& -54(17)60& 10(20)65& -70(21)20& 25(25)20& 31(27)60& -99(29)70& -81(33)40& -23(36)50& 43(39)00\\ 
1&-4452& 3762246& -15(6)50& 38(8)60& -76(10)66& 12(13)64& -16(15)76& 22(17)36& 10(19)98& -74(21)06& 99(24)52& -18(27)56& -99(29)22& -94(33)24& -25(36)62& 10(40)68\\ 
2&5430& -4032288& 14(6)28& -32(8)12& 55(10)30& -72(12)20& 69(14)05& -39(16)24& -66(18)64& -10(21)92& 70(23)20& -18(27)72& -10(29)93& -35(33)00& -28(35)72& 10(40)80\\ 
3&-3672& 2417742& -764921214& 15(8)36& -22(10)72& 24(12)50& -19(14)58& 11(16)92& -12(18)12& 68(20)02& -70(21)40& -84(25)68& 19(29)60& -36(32)16& 44(35)14& 46(39)96\\ 
4&1386& -819123& 232934157& -42(7)18& 55(9)40& -52(11)90& 38(13)08& -28(15)92& 34(17)58& 23(19)08& -26(22)61& 13(25)28& 11(28)61& -16(31)69& 87(34)33& 93(38)50\\ 
5&-270& 144414& -36870264& 59(6)18& -66(8)64& 53(10)20& -36(12)64& 25(14)44& -17(16)10& -15(19)46& -27(21)42& -22(24)32& -31(27)40& -28(30)90& 11(33)48& 71(37)58\\ 
6&21& -9972& 2156373& -268703481& 18(7)03& -18(8)16& -81(10)10& 22(13)91& -37(15)35& 26(18)12& 12(21)14& 92(23)67& 99(26)51& 10(30)21& 10(33)53& 28(36)36\\ 
7&0& -60& 38340& -10994520& 18(6)58& -21(8)72& 16(10)50& -19(12)40& -21(12)98& -21(17)94& -27(20)58& -24(23)38& -31(26)52& -38(29)24& -46(32)28& -54(35)46\\ 
8&0& 0& 171& -119415& 38611944& -76(6)76& 10(9)55& -49(10)14& 29(13)54& 20(16)13& 42(19)53& 52(22)17& 84(25)73& 13(29)75& 19(32)23& 26(35)66\\ 
9&0& 0& 0& -396& 304434& -108352800& 23(7)36& -48(9)02& -14(12)90& -19(15)44& -52(18)12& -93(21)32& -19(25)94& -39(28)26& -71(31)74& -11(35)70\\ 
10&0& 0& 0& 0& 867& -706221& 268652163& -52(7)55& 81(10)13& 13(14)92& 53(17)38& 13(21)73& 40(24)22& 10(28)60& 24(31)74& 50(34)08\\ 
11&0& 0& 0& 0& 0& -1728& 1494120& -610535214& -12(9)20& -63(12)68& -43(16)74& -17(20)12& -74(23)92& -25(27)14& -76(30)86& -19(34)82\\ 
12&0& 0& 0&0& 0& 0& 3315& -3061881& 13(7)53& 19(11)21& 28(15)59& 18(19)44& 11(23)85& 56(26)96& 21(30)18& 70(33)41\\ 
13&0& 0& 0& 0& 0& 0& 0& -6270& -1610226& -35(9)86& -14(14)78& -16(18)92& -16(22)22& -10(26)54& -57(29)14& -23(33)82\\ 
14&0& 0& 0& 0& 0& 0& 0& 0& -66573& 29(7)54& 52(12)81& 11(17)40& 19(21)11& 19(25)17& 13(29)90& 72(32)56\\ 
15&0& 0& 0& 0& 0& 0& 0& 0& -270& -16295250& -13(11)04& -70(15)32& -20(20)00& -30(24)36& -30(28)78& -21(32)64\\ 
16&0& 0& 0& 0& 0& 0& 0& 0& 0& -165507& 21(9)57& 32(14)83& 18(19)79& 42(23)06& 60(27)28& 56(31)57\\ 
17&0& 0& 0& 0& 0& 0& 0& 0& 0& -564& -16(7)16& -11(13)58& -14(18)92& -52(22)34& -11(27)18& -13(31)80\\ 
18&0& 0& 0& 0& 0& 0& 0& 0& 0& 0& 8362953& 27(11)43& 89(16)62& 57(21)90& 18(26)12&31(30)32\\ 
19&0& 0& 0& 0& 0& 0& 0& 0& 0& 0& 85002& -42(9)80& -46(15)04& -55(20)28& -27(25)04& -67(29)34\\ 
20&0& 0& 0& 0& 0& 0& 0& 0& 0& 0& 288& 30(7)39& 19(14)10& 46(19)61& 37(24)90& 13(29)42\\ 
21&0& 0& 0& 0& 0& 0& 0& 0& 0& 0& 0&-15047856& -61(12)40& -32(18)66& -45(23)04& -23(28)60\\ 
22&0& 0& 0& 0& 0& 0& 0& 0& 0& 0& 0& -154206& 13(11)69& 19(17)96& 49(22)11& 39(27)25\\ 
23&0& 0& 0& 0& 0& 0& 0& 0& 0& 0& 0&-522& -19(9)28& -95(15)98& -48(21)02& -58(26)82\\ 
24&0& 0& 0& 0& 0& 0& 0& 0& 0& 0& 0& 0& 13(7)15& 37(14)01& 40(20)15& 80(25)15\\ 
25&0& 0& 0& 0& 0& 0& 0& 0& 0& 0& 0& 0& -5815272& -11(13)26& -30(19)88& -10(25)94\\ 
26&0& 0& 0& 0& 0& 0& 0& 0& 0& 0& 0& 0& -62616& 23(11)81& 19(18)58& 11(24)73\\ 
27&0& 0& 0& 0& 0& 0& 0& 0& 0& 0& 0& 0& -216& -32(9)98& -10(17)16& -11(23)02\\ 
28&0& 0& 0& 0& 0& 0& 0& 0& 0& 0& 0& 0& 0& 21(7)41& 45(15)55& 10(22)70\\ 
29&0& 0& 0& 0& 0& 0& 0& 0& 0& 0& 0&0& 0& -8650458& -15(14)42& -80(20)54\\ 
30&0& 0& 0& 0& 0& 0& 0& 0& 0& 0& 0& 0& 0& -96237& 43(12)14& 55(19)26\\ 
31&0& 0& 0& 0& 0& 0& 0& 0& 0& 0& 0& 0& 0& -336& -84(10)36& -32(18)78\\ 
32&0& 0& 0& 0& 0& 0& 0& 0& 0& 0& 0& 0& 0& 0& 10(9)61& 16(17)50\\ 
33&0& 0& 0& 0& 0& 0& 0& 0& 0& 0& 0& 0& 0& 0& -63(6)14& -67(15)76\\ 
34&0& 0& 0& 0& 0& 0& 0& 0& 0& 0& 0& 0& 0& 0& 2245968& 22(14)53\\ 
35&0& 0& 0& 0& 0& 0& 0& 0& 0& 0& 0& 0& 0& 0& 27342& -58(12)80\\ 
36&0& 0& 0& 0& 0& 0& 0& 0& 0& 0& 0& 0& 0& 0& 99& 10(11)24\\ 
37&0& 0& 0& 0& 0& 0& 0& 0& 0& 0& 0&0& 0& 0& 0& -12(9)22\\ 
38&0& 0& 0& 0& 0& 0& 0& 0& 0& 0& 0& 0& 0& 0& 0& 74(6)99\\ 
39&0& 0& 0& 0& 0& 0& 0& 0& 0& 0& 0& 0& 0& 0& 0& -2395884\\ 
40&0& 0& 0& 0& 0& 0& 0& 0& 0& 0& 0&0& 0& 0& 0& -30840\\ 
41&0& 0& 0& 0& 0& 0& 0&  0& 0& 0& 0& 0& 0& 0& 0& -114\\ \hline    
\end{tabular}}
\caption{Some BPS invariants for  $n_{(d_E,5)}^{g}$. To save space we only give 
the first and the last two significant digits and the number of omitted digits in brackets.}  
\label{basedegree5}
 \end{center}
\end{sidewaystable}

\newpage

\section{Derivation of the involution symmetry on the propagators}

\label{propagators}
For the compact models we have the  an-holomorphic propagators $S^{ij}, S^{i}, S$, where $i,j$ runs over 
the complex structure moduli. It is convenient to make a change of variables with the derivative of 
K\"ahler potential $K_i=\partial_i K$ by the following
\begin{eqnarray}
S^{ij}\rightarrow S^{ij}, ~~~S^i\rightarrow S^i - S^{ij}K_j, ~~~ S\rightarrow S- S^iK_i +\frac{1}{2} S^{ij}K_iK_j, 
\label{shift}  
\end{eqnarray}
In the following we refer to the propagators $S^{ij}, S^{i}, S$ as the ones after the change of 
variables in~\cite{Alim:2007}. 

The propagators are defined by relating their anti-holomorphic derivatives to three point couplings. 
One can integrate these relations and also special geometry relation, and we get 
\begin{eqnarray} \label{propa2.6}
\Gamma_{ij}^k &=& \delta^{k}_{i}K_j + \delta^k_jK_i -C_{ijl}S^{kl} +s^k_{ij}, \nonumber \\
\partial_iS^{jk} &=& C_{imn} S^{mj} S^{nk} +\delta^j_i S^k +\delta^k_i S^j -s^j_{im} S^{mk} -s^k_{im} S^{mj} +h^{jk}_i,  \nonumber \\
\partial_i S^j &=& C_{imn} S^{mj} S^n +2 \delta^j_i S - s^j_{im} S^m -h_{ik} S^{kj} +h^j_i,  \nonumber \\
\partial_i S &=& \frac{1}{2} C_{imn} S^m S^n -h_{ij} S^j +h_i, \nonumber \\
\partial_i K_j &=& K_iK_j -C_{ijn} S^{mn} K_m +s^m_{ij} K_m -C_{ijk}S^k +h_{ij} .  
\end{eqnarray}
Here the holomorphic ambiguities $s^k_{ij}, h^{jk}_i, h_{ij}, h^j_i, h_i$ are some rational functions from the 
integration constants of the anti-holomorphic derivatives.  We can compute the Kahler potential 
$K$ and Christoffel connections $\Gamma_{ij}^k$ in the holomorphic limit from Picard-Fuchs equations.  
There are some freedom to choose some of the ambiguities such that the set of equations has a consistent solution for the other ambiguities and the propagators $S^{ij}, S^{i}, S$.

We first discuss the gauge choice made in \cite{Alim:2012}, which are the followings 
\begin{eqnarray} \label{holoambi2.7} 
&& s^1_{11} = -\frac{2}{z_1},   ~~~~~~ s^1_{12}=   -\frac{1}{3z_2},  ~~~~~~   s^1_{22}=0, \nonumber \\
&& s^2_{11}=0, ~~~~~~ s^2_{12} = 0, ~~~~~~  s^2_{22} = -\frac{4}{3z_2},  \nonumber \\
&& h^{1 1}_{1} = z_1[\frac{1}{9} - 48z_1 + \frac{5}{6}z_2 - 540z_1z_2], ~~~ 
h^{12}_{1} =  z_2[ -\frac{5}{108}  - \frac{5}{4}z_2 + 20z_1 + 540 z_1z_2],  \nonumber \\
&& h^{22}_{1} = -60z^2_2(1 + 27z_2), ~~~ 
h^{1 1}_{2} = -60z^3_1, \nonumber \\
&& h^{1 2}_{2} = z_1[\frac{1}{9} + \frac{5}{12}z_2 - 48 z_1],  ~~~
h^{2 2}_{2} = z_2[-\frac{23}{54} + 40z_1 - \frac{5}{2} z_2 - 540z_1z_2], \nonumber \\
&& h^1_1 = \frac{155}{27} z_1 - \frac{25}{1296} z_2 + 50 z_1z_2, ~~~
h^2_1 = 0,  \nonumber \\
&& h^1_2 = - \frac{5}{18} z_1 + 120 z^2_1,  ~~~   h^2_2 = \frac{155}{27} z_1 + \frac{1055}{1296} z_2 + 50 z_1z_2, \nonumber \\
&& h_1 = \frac{25}{23328z_1},  ~~~~  h_2 = -\frac{50}{3} z_1 ,    \nonumber \\
&& h_{1 1} = \frac{5}{36z_1z_2}, ~~~  h_{1 2}  = \frac{5}{108z_1z_2}, ~~~  h_{2 2} = 0. 
\end{eqnarray}

We will need to know how the propagators transform under the involution. Our guiding principle is the 
followings. We replace the coordinates $z_i$'s with $x_i$'s and all quantities with their transformations 
in the equations (\ref{propa2.6}). We then require the resulting equations are equivalent to the original ones. 
Here for the holomorphic ambiguities we can simply replace the $z_i$ coordinates with $x_i$'s in the expressions (\ref{holoambi2.7}). 
Since the Picard-Fuchs operators are invariant under the involution, we argue the Kahler potential is invariant $\tilde{K} = K$, and the Christoffel symbols transform as 
\begin{eqnarray} \label{Christoffeltrans}
\tilde{\Gamma}_{ij}^k = \frac{\partial x_k}{\partial z_l} \frac{\partial z_m}{\partial x_i} \frac{\partial z_n}{\partial x_j}
\Gamma_{mn}^l +   \frac{\partial ^2 z_m}{ \partial x_i \partial x_j} \frac{\partial x_k} {\partial z_m}
\end{eqnarray} 

We find that it is possible to solve for the transformed propagators $\tilde{S}^{ij},  \tilde{S}^{i}, \tilde{S}$ 
such that the transformed equations of (\ref{propa2.6}) are equivalent to the original ones. 

It is easy to check that with the choices of the ambiguities $s^i_{jk}$ in (\ref{holoambi2.7}), the transformations are 
\begin{eqnarray} \label{sshift3.25}
\tilde{s}_{ij}^k = \frac{\partial x_k}{\partial z_l} \frac{\partial z_m}{\partial x_i} \frac{\partial z_n}{\partial x_j}
s_{mn}^l +   \frac{\partial ^2 z_m}{ \partial x_i \partial x_j} \frac{\partial x_k} {\partial z_m}. 
\end{eqnarray} 
The shift exactly cancels that of the Christoffel symbols in the first equation in (\ref{propa2.6}), therefore if $S^{ij}$ transform as tensor with a minus sign, the transformed equation is equivalent to the original one by a coordinate transformation.

To solve for the shifts in the transformed propagators $\tilde{S}^{i}$, we combine the involution 
transformation  with the coordinate transformation of the second equation in (\ref{propa2.6}). 
We find the shift in the gauge $s_{ij}^k$ exactly cancels the one from the derivative of $S^{jk}$ 
propagator and all propagators cancel out, so the shifts $f^i$ ($i=1,2$) are determined by the equations
\begin{eqnarray} \label{solf3.26}
\delta^j_i f^k+  \delta^k_i f^j + \tilde{h}^{jk}_i +  \frac{\partial x_j}{\partial z_l}  \frac{\partial x_k}{\partial z_m} \frac{\partial z_n}{\partial x_i} h^{lm}_n =0. 
\end{eqnarray}
We see that if $h^{jk}_i $ transformed as a tensor with a minus sign, 
the shifts would have vanished. The shifts are results of the non-tensorial 
transformation of the gauge choice $h^{jk}_i $.  

Similarly we combine the involution transformation with the coordinate transformation of the third 
equation in (\ref{propa2.6}). Utilizing the solution (\ref{shift2.11}), we again find all 
dependence on propagators cancel, and the shift $f^0$ is determined by 
\begin{eqnarray}
\frac{\partial z_k}{\partial x_i} (\partial_k f^j) = 2\delta^j_i f^0 - \tilde{s}^j_{im} f^m + \tilde{h}^j_i +  \frac{\partial x_j}{\partial z_l} \frac{\partial z_n}{\partial x_i} h^{l}_n, 
\end{eqnarray} 
and we find the the last equation in (\ref{shift2.11})

With the solutions for the shifts (\ref{shift2.11}), one can further check 
that the involution transformation of the fourth and fifth equations in (\ref{propa2.6}) 
are equivalent to the original ones by coordinate transformations. 

It is possible to find a gauge choice different from (\ref{holoambi2.7}), 
such that the shifts in propagators $f^0, f^i$ vanish. We can keep the propagators 
$S^{ij}$ and ambiguity $s^i_{jk}$, and shift the gauge choice for $h^{jk}_i$ in (\ref{holoambi2.7}) by 
\begin{eqnarray}
h^{jk}_i\rightarrow h^{jk}_i + \delta^i_i h^k +  \delta^k_i h^j
\end{eqnarray} 
where as an example we choose 
\begin{eqnarray}
h^1 = \frac{5z_1}{12}(432z_1-1) , ~~~~ h^2 = \frac{5z_2}{216}(1+54z_2) . 
\end{eqnarray}

One can easily check that after the shifts, the ambiguity $h^{jk}_i$ transforms as a tensor with 
a minus sign under the involution transformation. So according to (\ref{solf3.26}) this new 
gauge choice will eliminate the shifts $f^1, f^2$ in the involution transformation of the 
propagators $S^i$. The shift in $h^{jk}_i$ will change the propagators $S^i, S$ and the other 
ambiguities $h^j_i, h_i, h_{ij}$ as well, which can be straightforwardly determined.  

Similarly we can further shift the ambiguity $h^j_i$ to cancel the shift $f^0$ in the 
involution transformation of the propagator $S$. 

The gauge choice which gives no shifts for the involution transformation of the 
propagators is conceptually simple, and give a much better understanding of the 
involution symmetry  as auto-equivalence of the underlying derived categories. 
However, as it turns out such a gauge choice would introduce pole of $1-432z_1$ 
in the ambiguities  $h^j_i$, $h_i$. Without detailed calculations, this point 
can be seen by noting that the shifts  (\ref{shift2.11}) have 
poles in $1-432z_1$ which does not appear in the gauge choice in  (\ref{holoambi2.7}), 
so in order to cancel the shifts, we would likely need to introduce poles of $1-432z_1$ 
in the ambiguities. As a consequence, the holomorphic ambiguities at higher genus 
would have also poles at $1-432z_1$. The total topological string amplitudes at higher 
genus should be regular at  $1-432z_1$, so the singular part of the holomorphic 
ambiguities at the pole $1-432z_1$ can be fixed and the gauge choice does not really 
enlarge the space of holomorphic ambiguities at higher genus but complicates the calculations. 

In practical calculations it is simpler to use the gauge  (\ref{holoambi2.7}), 
for which it turns out the holomorphic ambiguities at higher genus have no pole at  
$1-432z_1$, but only at the conifold divisors $\Delta_1$ and $\Delta_2$. In the followings 
we will still use the gauge choice (\ref{holoambi2.7}).

\section{Reducing the ambiguity with the involution symmetry}
\label{effectiveboundaries} 
We can compute the number of remaining unknown constants for general genus $g$ after imposing the 
condition (\ref{involutionFg}) more precisely. First we can find a particular holomorphic 
ambiguity such that the total amplitude with the propagator dependent part satisfies 
the constrain (\ref{involutionFg}). Then there are freedom to add an additional piece 
$f^{(g)}$ to the holomorphic ambiguity which satisfies the same symmetry condition
\begin{eqnarray}   \label{constrain2.17}
\tilde{f}^{(g)} = (-1)^{g-1} f^{(g)} . 
\end{eqnarray} 
The space of such $f^{(g)}$ form a linear space.  The number of remaining unknown constants is the dimension of this linear space. 

We can define the following linear vector spaces of holomorphic ambiguities 
\begin{eqnarray} \label{Vspaces}
V_0^{(g,m,n)} &:=& \{  f  ~|~  f =  \frac{p(z_1) z_2^n}{(\Delta_1\Delta_2)^{2g-2}} ,  \textrm{ where $p(z_1)$ is a polynomial degree $m$ in $z_1$} \}.  \nonumber \\
V_1^{(g,m,n)}  &:=& \{  f  ~|~  f \in V_0^{(g,m,n)}  \textrm{ and }\tilde{f} \in V_0^{(g,m,n)} \}.   \nonumber\\
V_2^{(g,m,n)}  &:=& \{  f  ~|~  f \in V_1^{(g,m,n)}  \textrm{ and } \tilde{f} = (-1)^{g-1} f  \}.  \nonumber \\
V_3^{(g,m,n)}  &:=& \{  f  ~|~  f \in V_1^{(g,m,n)}  \textrm{ and } \tilde{f}  = (-1)^{g} f  \}. 
\end{eqnarray}
It is easy to check they are indeed linear vector spaces, i.e., any linear combinations of elements of the space is also an element in the space. We would like to construct the explicit linear basis for these spaces.

Clearly $\textrm{dim} (V_0^{(g,m,n)})=m+1$. A convenient linearly independent basis is   
\begin{eqnarray} \label{basisV0}
f(z_1,z_2)= \frac{z_1^k(\frac{1}{432}-z_1)^{m-k} z_2^n}{(\Delta_1\Delta_2)^{2g-2}}, ~~~~ \textrm{where} ~k=0,1,\cdots, m
\end{eqnarray}
We can work out the involution transformation 
\begin{eqnarray} \label{trans2.19}
\frac{z_1^k(\frac{1}{432} -z_1)^{m-k} z_2^n}{(\Delta_1\Delta_2)^{2g-2}} ~ \rightarrow ~ \frac{(-1)^{n}  z_2^n}{(\Delta_1\Delta_2)^{2g-2}} ~z_1^{m-k+3n-6g+6} (\frac{1}{432}-z_1) ^{k-3n+6g-6}. 
\end{eqnarray}
We discuss two cases 
\begin{enumerate}
\item The case $n\geq 2g-2$. The holomorphic ambiguity in (\ref{basisV0})  $f\in V_1^{(g,m,n)}$  for $ k\geq 3n-(6g-6)$. On the other hand,  
any linear combinations of the basis with $k< 3n-(6g-6)$ in  (\ref{basisV0}) is not in the space $V_1^{(g,m,n)}$ since there will be poles 
of $(\frac{1}{432}-z_1)$ in the numerator. Therefore the space  $V_1^{(g,m,n)}$ is generated by the linear combinations of  
$\frac{z_1^k(\frac{1}{432}-z_1)^{m-k} z_2^n}{(\Delta_1\Delta_2)^{2g-2}}$,  where $k=3n-(6g-6),\cdots, m$.  In this case the dimension of $V_1^{(g,m,n)}$ is 
\begin{eqnarray}
\textrm{dim}(V_1^{(g,m,n)}) &=&  
 \left\{
\begin{array}{cl}
 m+6g-6 -3n +1 ,    & ~~  \textrm{if} ~~~ m+6g-6 -3n \geq 0 ;   \\
 0,             & ~~ \textrm{if} ~~~ m< 3n-(6g-6)   .   
\end{array}    
\right.
\end{eqnarray} 

\item The case $n< 2g-2$. A similar discussion shows that the space  $V_1^{(g,m,n)}$ is generated by the linear 
combinations of  $\frac{z_1^k(\frac{1}{432}-z_1)^{m-k} z_2^n}{(\Delta_1\Delta_2)^{2g-2}}$,  
where $k=0,1,\cdots, m +3n-6g+6 $. In this case the dimension of $V_1^{(g,m,n)}$ is 
\begin{eqnarray}
\textrm{dim}(V_1^{(g,m,n)}) &=&  
 \left\{
\begin{array}{cl}
 m+3n -6g+6  +1 ,    & ~~  \textrm{if} ~~~ m+3n -6g+6 \geq 0 ;   \\
 0,             & ~~ \textrm{if} ~~~ m< 3n-(6g-6)   .   
\end{array}    
\right.
\end{eqnarray} 
\end{enumerate}
In both cases we can write the formula 
\begin{eqnarray} \label{V1dim2.22}
\textrm{dim}(V_1^{(g,m,n)}) &=&  
 \left\{
\begin{array}{cl}
 m-|3n -6g+6|  +1 ,    & ~~  \textrm{if} ~~~  m\geq |3n -6g+6|  ;   \\
 0,             & ~~ \textrm{if} ~~~ m< |3n-6g+6|   .   
\end{array}    
\right.
\end{eqnarray}

For any element  $f\in V_1^{(g,m,n)}$ we can write 
\begin{eqnarray}
 f =\frac{1}{2}(f + (-1)^{g-1} \tilde{f} ) + \frac{1}{2}(f + (-1)^{g} \tilde{f} ). 
 \end{eqnarray}
Clearly the first term $\frac{1}{2}(f + (-1)^{g-1} \tilde{f} )  \in V_2^{(g,m,n)}$ 
and the second term $ \frac{1}{2}(f + (-1)^{g} \tilde{f} )  \in V_3^{(g,m,n)}$. 
We also note that $V_2^{(g,m,n)} \cap V_3^{(g,m,n)} =\{0\}$. So we have the decomposition 
\begin{eqnarray}
V_1^{(g,m,n)} = V_2^{(g,m,n)}  \oplus V_3^{(g,m,n)}
\end{eqnarray}

To construct the linear basis for $V_2^{(g,m,n)}$ and $V_3^{(g,m,n)}$, we further define two spaces $V_+^{(g,m,n)}$ and $V_-^{(g,m,n)}$ generated by linearly independent basis  
\begin{eqnarray} \label{basisplus2.25}
 V_+^{(g,m,n)}&:=&  \{ \textrm{linear space generated by } \frac{z_1^{k_1}(\frac{1}{432}-z_1)^{k_2} z_2^n}{(\Delta_1\Delta_2)^{2g-2}} \in V_0^{(g,m,n)} ,  \nonumber \\  && \textrm{with} ~k_1-k_2 = 3n-6g+6. ~\}, \\
 \label{basisplus2.26}
  V_-^{(g,m,n)}&:=&  \{ \textrm{linear space generated by } \frac{(\frac{1}{432}- 2 z_1) z_1^{k_1}(\frac{1}{432}-z_1)^{k_2} z_2^n}{(\Delta_1\Delta_2)^{2g-2}} \in V_0^{(g,m,n)} ,  \nonumber \\  && \textrm{with} ~k_1-k_2 = 3n-6g+6. ~\}
\end{eqnarray}
{}From the transformation (\ref{trans2.19}), we can check that for $f\in V_+^{(g,m,n)}$ we have $\tilde{f} = (-1)^{n} f$, while  for $f\in V_-^{(g,m,n)}$ we have $\tilde{f} = (-1)^{n+1} f $.  
Therefore if $n+g$ is an odd integer, we have  $V_+^{(g,m,n)}\subseteq V_2^{(g,m,n)}$ and $V_-^{(g,m,n)}\subseteq V_3^{(g,m,n)}$, 
while if $n+g$ is an even integer, we have  $V_+^{(g,m,n)}\subseteq V_3^{(g,m,n)}$ and $V_-^{(g,m,n)}\subseteq V_2^{(g,m,n)}$. So we find 
\begin{eqnarray} \label{ineq2.27}
\textrm{dim} (V_+^{(g,m,n)}) + \textrm{dim} (V_-^{(g,m,n)})  \leq \textrm{dim} (V_2^{(g,m,n)}) + \textrm{dim} (V_3^{(g,m,n)}) = \textrm{dim} (V_1^{(g,m,n)}).
\end{eqnarray}  

We can compute the dimensions of $V_+^{(g,m,n)}$ and $V_-^{(g,m,n)}$ easily from the explicit linearly independent 
basis in (\ref{basisplus2.25}) and (\ref{basisplus2.26}). Similar to the case of $V_1^{(g,m,n)}$, one also need 
to consider two cases $n\geq  2g-2$ and $n< 2g-2$, and find a universal formulae 
\begin{eqnarray} \label{V2dim2.28}
\textrm{dim} (V_+^{(g,m,n)})  
 &=&  
 \left\{
\begin{array}{cl}
 [\frac{m - |3n -6g+6|}{2}]+1   ,    & ~~  \textrm{if} ~~~  m\geq |3n -6g+6|    ;   \\
 0,             & ~~ \textrm{if} ~~~ m< |3n-6g+6|   .   
\end{array}    
\right.  \nonumber \\ 
\textrm{dim} (V_-^{(g,m,n)})  
 &=&  
 \left\{
\begin{array}{cl}
 [\frac{m -1- |3n -6g+6|}{2}]+1   ,    & ~~  \textrm{if} ~~~  m-1 \geq |3n -6g+6|    ;   \\
 0,             & ~~ \textrm{if} ~~~ m-1 < |3n-6g+6|   .   
\end{array}    
\right.
\end{eqnarray}
Comparing with (\ref{V1dim2.22}) we find actually the inequality in (\ref{ineq2.27}) is saturated. So the spaces $V_+^{(g,m,n)}$ and $V_-^{(g,m,n)}$ provide the complete linear basis, we have $V_+^{(g,m,n)}=  V_2^{(g,m,n)}, V_-^{(g,m,n)}= V_3^{(g,m,n)}$ for the case of $n+g$ an odd integer, and $V_+^{(g,m,n)}=  V_3^{(g,m,n)}, V_-^{(g,m,n)}= V_2^{(g,m,n)}$ for the case of $n+g$ an even integer.

We can estimate the number of unknown constants for the holomorphic ambiguity at large genus $g$  after imposing the involution symmetry, assuming the degree of $z_1$ is $7(g-1)$ in the numerator  
\begin{eqnarray} \label{ng2.29}
\sum_{n=0}^{\infty} \textrm{dim} ( V_2^{(g,7(g-1),n)})  = \sum_{n=0}^{\frac{13(g-1)}{3}} \frac{7(g-1)- |3n- 6(g-1)  |}{2} \sim \frac{97}{12} g^2, ~~~ \textrm{for large} ~g. 
\end{eqnarray} 
This is about  one quarter of the $(7g-6)(5g-4)$ unknown constants in the naive ansatz at large genus.

\section{Fibre modularity versus involution symmetry}
\label{sec:FI}
In appendix we will prove that the constraints imposed on the amplitudes 
and especially on the holomorphic  ambiguity from the involution symmetry 
is equivalent to the  constraints from the modularity in the fibre 
direction discussed in  section \ref{fibremodularity}. In addition the 
appendix provides many identities that are useful for further discussion
of the model~\cite{HKKZ}. 

{}From the solutions\footnote{We lower the index on the A-periods $X^i$ 
in order to avoid cluttering with exponentials.}  of 
the (\ref{PF2.1}) around the large volume point $(z_1,z_2)\sim (0,0)$,
\begin{eqnarray}
X_0 &=& 1+60z_1+13860 z_1^2 +\cdots, \nonumber \\
X_1 &=&  X_0 \log (z_1) + 312 z_1 + 2z_2 + \cdots, \nonumber \\
X_2  &=&  X_0 \log(z_2) +180 z_1 -6 z_2 +\cdots ,
\end{eqnarray}
we find  
\begin{eqnarray} 
&& q_E :=  q_1 =\exp( \frac{X_1}{X_0})  = z_1+312 z_1^2 +2z_1z_2 +\cdots, \nonumber \\
&& q_B :=  q_2 =\exp( \frac{X_2}{X_0})  = z_2 +180z_1 z_2  -6z_2^2 +\cdots . 
\end{eqnarray}
The above relations can be readily inverted and we can write $z_{1,2}$ as power 
series of $q_{1,2}$. 

If we keep the gauge $X_0^{2g-2}$ explicit and note that $q_B$ is defined without the shift 
(\ref{redefinitionti}) then (\ref{freeenenergybase},\ref{defPgb}) read 
\begin{eqnarray}  \label{base3.48}
X_0^{2g-2} \mathcal{F}^{(g)} = \sum_{k=0}^{\infty} P^{(g)}_k(q_E) ~(\frac{q_E}{\eta(q_E)^{24}})^{\frac{3k}{2}} q_B^k, 
~~~ g\geq 1, 
\end{eqnarray} 
while for the case of $g=0$ we do not need the factor of  $X_0^{-2}$ in the conventional definition of 
the prepotential. 

{}From section \ref{fibremodularity} follows for the $X_{18}(1,1,1,6,9)$ model that the $P^{(g)}_k(q_E)$ 
are quasi-modular forms of ${\rm SL}(2,\mathbb{Z})$, of modular weight 
$18k+2g-2$ and satisfy (\ref{anomalyP2}) as a specialization of (\ref{anomaly}).
We would like to understand whether and to what extent this constraint from the modularity of the fiber 
fixes the holomorphic ambiguity at higher genus. Suppose at genus $g\geq 2$ we have found a particular 
holomorphic ambiguity such that the topological string amplitude $\mathcal{F}^{(g)} $ satisfies the fiber modularity constrain. 
Then we may add an additional holomorphic ambiguity $f^{(g)}$, whose expansion is 
\begin{eqnarray} \label{const2.25}
X_0^{2g-2} f^{(g)} = \sum_{k=0}^{\infty} p^{(g)}_k(q_E) ~(\frac{q_E}{\eta(q_E)^{24}})^{\frac{3k}{2}} q_B^k, 
\end{eqnarray} 
we require  $ p^{(g)}_k(q_E)$ to be also a quasi-modular form of weight $18k+2g-2$. Furthermore, since the sum $P^{(g)}_k(q_E) + p^{(g)}_k(q_E)$ 
also need to satisfy the equation (\ref{anomalyP2}), we find 
\begin{eqnarray}  \label{const2.26}
\frac{\partial p^{(g)}_k}{\partial  {E_2}} = -\frac{1}{12} \sum_{s=1}^{k-1} s(k-s) p^{(g)}_s P^{(0)}_{k-s} .
\end{eqnarray}
In particular for the case $k=0, 1$, this equation is understood as $\frac{\partial p^{(g)}_k}{\partial  {E_2}}=0$, 
i.e. $p^{(g)}_0$ and $p^{(g)}_1 $ are always modular forms.

It turns out that not all rational functions $ f^{(g)}$ satisfy the fiber modularity constraint 
(\ref{const2.25}, \ref{const2.26}), so the fiber modularity does impose some conditions on 
holomorphic ambiguities. However we find that there exist some non-vanishing rational 
functions  $ f^{(g)}$ satisfying these conditions, so the fiber modularity does not completely 
fix the holomorphic ambiguity. We see that the BCOV holomorphic anomaly equation alone 
does not implies the modularity of the fiber at higher genus. 

However it turns out that the modularity of the fiber is  equivalent to the involution 
symmetry. It gives more precise bounds for the degrees in the polynomials in the 
numerators of the holomorphic ambiguity and further reduces the number of unknown constants 
in the holomorphic ambiguity. To understand this equivalence,  we consider a 
holomorphic ambiguity 
$f\in V_0^{(g,+\infty,n)}$. Since $f\sim z_2^n \sim q_B^n$ in the small $z$ limit, it is easy to see that the first non-vanish 
coefficient in the expansion in base degree (\ref{const2.25}) is that of $q_B^n$ 
\begin{eqnarray}  \label{expan2.36}
X_0^{2g-2} f  = \sum_{k=n}^{\infty} p_k(q_E) ~(\frac{q_E}{\eta(q_E)^{24}})^{\frac{3k}{2}} q_B^k, 
\end{eqnarray} 
According to (\ref{const2.26}) the first coefficient $p_n(q_E)$ has no $E_2$ dependence and is purely a  
modular form, i.e. polynomials of $E_4$ and $E_6$. 

We define the linear space 
\begin{eqnarray} \label{W0space}
W_0^{(g,n)} &:=&  \{  f  ~|~  f\in V_0^{(g,+\infty,n)}, \textrm{ and the first non-vanishing coefficient $p_n(q_E)$ in the expansion}  
\nonumber \\ && \textrm{(\ref{expan2.36}) is a $SL(2,\mathbb{Z})$ modular form of weight $18n+2g-2  \}$.} 
\end{eqnarray}

We will prove the following proposition 
\begin{prop}\label{proposition2.37}
\[
W_0^{(g,n)}= V_2^{(g,[\frac{19}{3}(g-1)] ,n)}
\]
\end{prop}

\begin{proof}
The proof will proceed in two steps. We shall first show $V_2^{(g,[\frac{19}{3}(g-1)] ,n)} \subseteq W_0^{(g,n)}$, and then  
$\textrm{dim}(V_2^{(g,[\frac{19}{3}(g-1)],n)}) \geq \textrm{dim} ( W_0^{(g, n)})$. We will see how the number $[\frac{19}{3}(g-1)]$ appears in the process. 

In order to prove the Proposition~\ref{proposition2.37}, we only need the leading order term in the small $z_2$ expansion. 
The first Picard-Fuchs equation in (\ref{PF2.1}) in the $z_2\rightarrow 0$ limit has been studied before 
in \cite{HST, Klemm:2012}. We find the power series solution and A-model parameter $q_E$ are determined by the exact relations 
\begin{eqnarray}  \label{exact2.37} 
E_4(q_E) = X_0^4, ~~~ ~  z_1(1-432z_1) = \frac{\eta(q_E)^{24}}{E_4(q_E)^3}, 
\end{eqnarray}
where we have taken the $z_2\rightarrow 0$ limit for $X_0$ and $q_E$. We will drop the 
argument $q_E=q_1$ in the modular forms and $\eta$ function for simplicity. 

For the other logarithmic solution we need to keep the log term in $z_2\rightarrow 0$ limit
\begin{eqnarray}
X_2 = X_0 \log(z_2) + \xi (z_1) +\mathcal{O}(z_2),
\end{eqnarray}
where $\xi(z_1) = 180z_1 +\mathcal{O}( z_1^2) $ is a power series. Using the first Picard-Fuchs equation in (\ref{PF2.1}) we find 
\begin{eqnarray}
[\theta_1^2 -12 z_1 (6\theta_1+1) (6\theta_1+5)] \xi(z_1) = 3\theta_1 X_0 
\end{eqnarray} 
There is an exact solution for $\xi(z_1)$ using (\ref{exact2.37}) and the well-known Ramanujan derivative relations for quasi-modular forms
\begin{eqnarray} \label{xi3.63}
\xi(z_1) = -\frac{3}{2} X_0 \log [\frac{q_1(1-432z_1)}{z_1} ]. 
\end{eqnarray} 
Of course one can add any linear combination of two solutions of PF equation in the $z_2\rightarrow 0$ limit, which 
are $E_4(q_1)^{\frac{1}{4}}$ and $\log(q_1)E_4(q_1)^{\frac{1}{4}}$. The above solution is the only one with the correct 
asymptotic behavior $\xi(z_1)\sim z_1$ for small $z_1$. So  the A-model parameter $q_B$ is 
\begin{eqnarray}
q_B = z_2 \exp( \frac{\xi(z_1)}{X_0}) = z_2 (\frac{q_1(1-432z_1)}{z_1} )^{- \frac{3}{2}} +\mathcal{O}(z_2^2)
\end{eqnarray} 

We first discuss the case that $n+g$ is an odd integer.  Suppose  
$f\in V_2^{(g,[\frac{19}{3}(g-1)],n)} = V_+^{(g,[\frac{19}{3}(g-1)],n)}$ is a base vector in (\ref{basisplus2.25}) 
\begin{eqnarray} \label{f3.43}
 && f= \frac{ z_1^{k_1}(\frac{1}{432}-z_1)^{k_2} z_2^n}{(\Delta_1\Delta_2)^{2g-2}} ,   
 \nonumber \\ &&  \textrm{with} ~~k_1\geq 0, ~k_2\geq 0, ~ k_1+k_2\leq \frac{19}{3}(g-1), ~ k_1-k_2 = 3n-6g+6,
 \end{eqnarray} 
then the leading term in the expansion (\ref{expan2.36}) is 
\begin{eqnarray} \label{form2.44}
X_0^{2g-2} f &\sim&   E_4^{\frac{g-1}{2}} z_1^{k_1-\frac{3n}{2} } (1-432 z_1) ^{k_2 +\frac{3n}{2}-6g+6} \eta^{36n}
 ~(\frac{q_E}{\eta^{24}})^{\frac{3n}{2}} q_B^n \nonumber \\
 &\sim & E_4^{\frac{g-1}{2} +\frac{9n}{2} -3k_1} \eta^{24k_1}  ~(\frac{q_E}{\eta^{24}})^{\frac{3n}{2}} q_B^n ,
\end{eqnarray} 
where we have used the relations (\ref{exact2.37}) and ignore constant factors. Since  $n+g$ is an odd integer, the exponent of $E_4$ is an integer. Furthermore 
\begin{eqnarray}
\frac{g-1}{2} +\frac{9n}{2} -3k_1= \frac{19}{2}(g-1)-\frac{3}{2} (k_1+k_2) \geq 0,  
\end{eqnarray} 
so the coefficient $E_4^{\frac{g-1}{2} +\frac{9n}{2} -3k_1} \eta^{24k_1} $ is a modular form of weight $18n+2g-2$. 

The case  of $n+g$  an even integer is similar. We need to use the formula $1-864z_1 = E_6/ E_4^{\frac{3}{2}}$. 
In both cases we find $V_2^{(g,[\frac{19}{3}(g-1)] ,n)} \subseteq W_0^{(g,n)}$.

Now we compute $\textrm{dim} ( W_0^{(g, n)})$. We discuss two cases.
\begin{enumerate}
\item The case of $n\leq 2g-2$. We consider the dimension of the space $\mathcal{M}_{18n+2g-2}$ of modular forms of weight $18n+2g-2$. 
It is clear $\textrm{dim} (\mathcal{M}_{18n+2g-2}) = [\frac{9n+g-1}{6}]+1$ if $n+g$ is odd, and $\textrm{dim} (\mathcal{M}_{18n+2g-2}) = [\frac{9n+g-4}{6}]+1$ if $n+g$ is even. 
This is exactly the dimension $\textrm{dim}(V_2^{(g,[\frac{19}{3}(g-1)],n)})$ 
according to  (\ref{V2dim2.28}). The expansion (\ref{expan2.36}) is a linear map from 
$W_0^{(g, n)}$ to $\mathcal{M}_{18n+2g-2}$ and the kernel is the zero element. So in this case we find 
\begin{eqnarray}
\textrm{dim} ( W_0^{(g, n)}) \leq \textrm{dim} (\mathcal{M}_{18n+2g-2}) =  \textrm{dim}(V_2^{(g,[\frac{19}{3}(g-1)],n)})
\end{eqnarray} 

\item The case of $n> 2g-2$. Suppose $f= \frac{h(z_1)z_2^n}{(\Delta_1\Delta_2)^{2g-2}}$ where $h(z_1)$ is a polynomial of $z_1$. We find the leading term in the expansion  (\ref{expan2.36}) 
\begin{eqnarray} 
X_0^{2g-2} f &\sim&   E_4^{\frac{g-1}{2}} \frac{h(z_1)}{ z_1^{\frac{3n}{2} } (1-432 z_1) ^{6g-6 - \frac{3n}{2}} } \eta^{36n}
 ~(\frac{q_E}{\eta^{24}})^{\frac{3n}{2}} q_B^n  ,
\end{eqnarray} 
We consider the zero at $1-432 z_1\sim 0$, which correspond to $E_6+E_4^{\frac{3}{2}}\sim 0$. From see (\ref{exact2.37}) we see 
$\eta^{24}\sim (1-432 z_1)$. So the leading coefficient in (\ref{expan2.36})  $p_n(q_E) \sim (1-432 z_1)^{3n-6g+6}$. 
If $p_n(q_E)$ is a modular form of weight $18n+2g-2$, then $\frac{p_n(q_E)}{\eta^{72(n-2g+2)}}$ should be also a modular form of weight $74(g-1)-18n$. So we find 
\begin{eqnarray}
\textrm{dim} ( W_0^{(g, n)}) \leq \textrm{dim} (\mathcal{M}_{74(g-1)-18n}) =  \textrm{dim}(V_2^{(g,[\frac{19}{3}(g-1)],n)})
\end{eqnarray} 

\end{enumerate} 
\end{proof}

So we have proven $W_0^{(g, n)} = V_2^{(g,[\frac{19}{3}(g-1)],n)}$, and in the proof we see that they are also isomorphic to the space of 
modular forms  $\mathcal{M}_{18n+2g-2}$ for the  $n\leq 2g-2$ case, and $\mathcal{M}_{74(g-1)-18n}$ for the $n> 2g-2$ case. 
We further conjecture this is sufficient for the full fiber modularity constrain of the higher terms. 

\begin{conj}\label{conjecture3.49}
For $f\in   W_0^{(g, n)}$ with the expansion  (\ref{expan2.36}), the higher order terms 
$p_k(q_E)$ with $k>n$ are quasi-modular forms of weight $18k+2g-2$ and satisfy
the equation (\ref{const2.26}). 
\end{conj}

If our conjecture is correct, then the space of holomorphic ambiguities at genus $g$ satisfying the modularity constraint 
of the elliptic fiber is linearly generated by $\sum_{n=0}^{\infty} \oplus W_0^{(g, n)}= \sum_{n=0}^{\infty} \oplus V_2^{(g,[\frac{19}{3}(g-1)],n)}$.  
For example, the explicit basis at genus $g=2$ are the the following 7 elements 
\begin{eqnarray} 
&&\frac{(1-432z_1)^3 z_2}{(\Delta_1\Delta_2)^{2}}, ~~~  \frac{(1-432z_1)^4 z_1z_2}{(\Delta_1\Delta_2)^{2}}, ~~~ 
\frac{(1-864z_1) z_2^2 }{(\Delta_1\Delta_2)^{2}}, ~~~ \frac{(1-864z_1)(1-432z_1)  z_1  z_2^2 }{(\Delta_1\Delta_2)^{2}}, \nonumber \\ &&
~~\frac{(1-864z_1)(1-432z_1)^2  z_1^2   z_2^2 }{(\Delta_1\Delta_2)^{2}}, ~~~ \frac{z_1^3 z_2^3 }{(\Delta_1\Delta_2)^{2}},  ~~~ \frac{ (1-432z_1) z_1^4 z_2^3 }{(\Delta_1\Delta_2)^{2}}. 
\end{eqnarray}
The numbers of unknown constants at genus 2, 3, 4, 5, 6 after using the fiber modularity constrain are 7, 31, 70, 109, 176 respectively.

We can compute the number of unknown constants in the holomorphic ambiguity at genus $g$ with the formula (\ref{V2dim2.28}), and we only 
need to change the bound for degree of $z_1$ from $7(g-1)$ to $\frac{19}{3}(g-1)$ in the calculations (\ref{ng2.29})
\begin{eqnarray} 
\sum_{n=0}^{\infty} \textrm{dim} ( V_2^{(g,\frac{19}{3}(g-1),n)})  = \sum_{n=0}^{\frac{37}{9}(g-1)} \frac{\frac{19}{3}(g-1)- |3n- 6(g-1)  |}{2} \sim \frac{721}{108} g^2, ~~ \textrm{for large} ~g,
\end{eqnarray} 
which is a better estimate than (\ref{ng2.29}).

The degree $[\frac{19}{3}(g-1)]$ of $z_1$  from fiber modularity (\ref{proposition2.37}) 
can be also derived from the regularity condition at $z_1\sim \infty$. We can use the local 
coordinates $(\frac{1}{z_1},z_2)$ around the point $(z_1,z_2)=(\infty,0)$. It is easy to check 
that the power series solutions to the PF equations have the asymptotic behavior 
$w\sim z_1^{-\frac{1}{6}}$ or $w\sim z_1^{-\frac{5}{6}}$. For the scaling factor $X_0$ in the 
topological string amplitude we take the one with the lowest exponent $X_0\sim z_1^{-\frac{1}{6}}$. 
So the scaling exponent including the conifold divisors is 
\begin{eqnarray}
(\frac{X_0}{\Delta_1\Delta_2})^{2g-2}  \sim z_1 ^{-\frac{19}{3} (g-1)}, 
\end{eqnarray} 
and the regularity of $X_0^{2g-2} \mathcal{F}^{(g)}$ at $z_1\sim \infty$ requires the degree of $z_1$ at the numerator in the holomorphic ambiguity be no bigger than $\frac{19}{3}(g-1)$. 
The non-zero scaling exponent in $X_0$ around $z_1\sim \infty$ is characteristic of compact Calabi-Yau models, and we will see that it also is responsible for preventing the 
solution of the model at large genus as in one-parameter models like the quintic.

\subsection{Some (quasi-)modularity formulae} \label{secsomemore}
In order to understand the conjecture (\ref{conjecture3.49}), we provide some more formulae 
relating the complex structure coordinates $z_{1,2}$ to (quasi)modular forms of the fiber 
parameter. We note that for $f\in W_0^{(g,n)}$ in the case of $n+g$ odd, we can factor (\ref{f3.43}) with $k_2=k_1-3n+6g-6$ as the following 
\begin{eqnarray}
f =   [z_1(1-432z_1)]^{k_1} (\frac{z_2}{(1-432z_1)^3})^n (\frac{\Delta_1\Delta_2}{(1-432z_1)^3})^{2-2g}.
\end{eqnarray}
So we need to understand the expansion of $z_1(1-432z_1)$, $\frac{z_2}{(1-432z_1)^3}$, 
and $\frac{\Delta_1\Delta_2}{(1-432z_1)^3}$ in terms of exponential of flat coordinates $q_E, q_B$. 
There is an extra factor of $1-864z_1$ in the case of $n+g$ an even integer. 
Furthermore, we also need to consider the A-model scaling factor $X_0$ which is 
the power series solution to the Picard-Fuchs equation. 

We define the parameter
\begin{eqnarray} \label{defineS3.75}
Z := \frac{z_2}{(1-432z_1)^3E_4(q_E)^{\frac{9}{2}}},
\end{eqnarray}
and as a first step we expand some components as asymptotic series of $Z$, and we find that the coefficients are modular forms of fiber parameter $q_E$ as the followings 
\begin{eqnarray}
z_1(1-432z_1) &=& \frac{\eta^{24}}{E_4^3} \big[1 - \frac{77 E_4^3 E_6 + 211 E_6^3}{144} Z+ \frac{1}{165888} (-11858 E_4^9 + 330395 E_4^6 E_6^2   \nonumber \\ &&  + 1434560 E_4^3 E_6^4 + 
 1066999 E_6^6)Z^2 +\mathcal{O}(Z^3) \big], \label{expan3.54} \\
 \frac{\Delta_1\Delta_2}{(1-432z_1)^3} &=&  1+ \frac{27}{4} (3 E_4^3 E_6 + E_6^3) Z - \frac{9}{128} 
  (E_4^3 - E_6^2) (85 E_4^6 - 612 E_4^3 E_6^2 - 49 E_6^4) 
  Z^2   \nonumber \\ && +\mathcal{O}(Z^3),   \label{expan3.55}  
 \\
 1-864 z_1 &=& \frac{1}{E_4^{\frac{3}{2}}}  \big[ E_6 +  \frac{1}{288} 
 (E_4^3 - E_6^2) (77 E_4^3 + 211 E_6^2) Z - \frac{1}{331776} 
 E_6 (E_4^3 - E_6^2)  \nonumber \\ && \cdot  (383525 E_4^6 + 1458614 E_4^3 E_6^2 + 977957 E_6^4)
 Z^2 +\mathcal{O}(Z^3) \big],  \label{expan3.56}  \\
X_0 &=&  E_4^{\frac{1}{4}} \big[1-   \frac{5}{72} E_6 (E_4^3 - E_6^2) Z  
 + \frac{5}{7962624} 
 (E_4^3 - E_6^2) \nonumber \\ && \cdot (15935 E_4^6 + 488258 E_4^3 E_6^2 + 435839 E_6^4)
  Z^2 +\mathcal{O}(Z^3) \big],   \label{expan3.57} 
\end{eqnarray}
Here we fix the modular forms in low orders by perturbative calculations. In general we conjecture the following

\begin{conj}
\label{conjecture3.75}
With the appropriate normalization factors as in (\ref{expan3.54},  \ref{expan3.55}, \ref{expan3.56},  \ref{expan3.57}),  
the coefficients of $Z^n$ in the expansion are modular forms, and the modular weights are 
$18n$ for the cases of 
$z_1(1-432z_1)$, $ \frac{\Delta_1\Delta_2}{(1-432z_1)^3}$, $X_0$, and $18n+6$ for the case 
of $1-864z_1$.
\end{conj}

In the next subsection \ref{subsecproof3.5}, we shall prove this modularity conjecture.  
We note that we only need to prove the modularity of the coefficients of the expansion of 
$z_1(1-432z_1) $ and $1-864 z_1$, which immediately imply that of $\frac{\Delta_1\Delta_2}{(1-432z_1)^3}$, 
since we can write 
\begin{eqnarray}
\frac{\Delta_1\Delta_2}{(1-432z_1)^3} &=& 1- \frac{27z_2}{(1-432z_1)^3} (1-864z_1) [(1-864z_1)^2+1296z_1(1-432z_1)] \nonumber \\ && -\frac{(27z_2)^2}{(1-432z_1)^6} [432z_1(1-432z_1)]^3.
\end{eqnarray} 

We also notice that the higher order modular form coefficients in the expansion formula (\ref{expan3.56}) 
contain a factor of $\eta^{24} =\frac{E_4^3-E_6^2}{1728}$. This condition is sufficient to ensure 
the modularity of expansion (\ref{expan3.54}), and therefore that of (\ref{expan3.55}) as well. To 
see this, suppose 
\begin{eqnarray} \label{eta3.82}
1-864z_1 = \frac{1}{E_4^{\frac{3}{2}}} \big[ E_6 + (E_4^3-E_6^2)\sum_{n=1}^{\infty} f_{18n-6} Z^n \big],
\end{eqnarray} 
where $f_{18n-6}$ are some modular forms of weight $18n-6$. We can compute (\ref{expan3.54}) easily 
\begin{eqnarray} \label{eta3.83}
z_1(1-432z_1) = \frac{\eta^{24}}{E_4^3} \big[ 1- 2E_6 (\sum_{n=1}^{\infty} f_{18n-6} Z^n) + (E_4^3 -E_6^2) (\sum_{n=1}^{\infty} f_{18n-6} Z^n)^2 \big] 
\end{eqnarray}

Now it is easy to see that if our conjecture (\ref{conjecture3.75}) is correct, then for $f\in W_0^{(g,n)}$, it has the following expansion 
\begin{eqnarray} \label{expan3.58}
f= \sum_{n=0}^{+\infty} f_n Z^n,
\end{eqnarray}  
where $f_n$ are modular forms of $q_E$ of modular weight $18n+2g-2$. 

The next step is to expand $Z$ parameter in terms of the base coordinate $q_B$. For convenience we denote the normalized Kahler coordinate
\begin{eqnarray} \label{defineT}
Q:=(\frac{q_E}{\eta^{24}})^{\frac{3}{2}}q_B, 
\end{eqnarray}
and we find the formulae for low orders 
\begin{eqnarray} \label{pert3.60}
Z &=&  Q  + \frac{1}{576} [-E_2E_4(31  E_4^3 + 113  E_6^2)+ 1759 E_4^3 E_6  + 1841 E_6^3] ~Q^2  \nonumber \\ && + \frac{1}{15925248} [72 E_2^2 E_4^2(31 E_4^3 + 113  E_6^2)^2   - 
 108 E_2E_4 E_6(25249 E_4^6  + 143422 E_4^3 E_6^2 + 163105  E_6^4)\nonumber \\ &&  +102919 E_4^9 + 33568143 E_4^6 E_6^2  + 79676877 E_4^3 E_6^4 + 
 64318109 E_6^6  ]~Q^3 +\mathcal{O}(Q^4)
\end{eqnarray}
We can check that the coefficients $c_n$ of $Q^n$ are quasi-modular forms, and satisfy the following equation as in (\ref{const2.26})
\begin{eqnarray}  \label{modularanomaly3.61}
\frac{\partial c_n}{\partial  {E_2}} = -\frac{1}{12} \sum_{s=1}^{n-1} s(n-s) c_{n-s} P^{(0)}_{s} ,
\end{eqnarray}  
where the $P^{(0)}_n$ are the quasi-modular forms of the prepotential expanded in terms of base degree, for example 
\begin{eqnarray}
P^{(0)}_1 &=& \frac{E_4}{48} (31 E_4^3 + 113 E_6^2), \nonumber \\
P^{(0)}_2 &=&  -\frac{1}{221184} [ 4 E_2 E_4^2 (31 E_4^3 + 113 E_6^2)^2 +208991 E_4^7 E_6 
\nonumber \\ && + 755906 E_4^4 E_6^3 + 196319 E_4E_6^5 ]
\end{eqnarray} 

We define the following space 
\begin{eqnarray} \label{W1space}
W_1 &:=& \{  f  ~|~  f =\sum_{n=1}^{\infty} c_n Q^n,  \textrm{ where $c_n$ are quasi-modular forms } 
\nonumber \\ && \textrm{and satisfy equation (\ref{modularanomaly3.61}). $\}$}
\end{eqnarray}
Here we do not specify the modular weights of $c_n$, so unlike the previous cases, the space $W_1$ is not a 
linear space. Rather, we will see it is a graded module over modular forms. Since we know the modular 
weight of $P^{(0)}_n$ is $18n-2$, the modular weights of all $c_n$'s are actually determined by the 
first non-vanishing coefficient according to equation (\ref{modularanomaly3.61}). It is clear 
if $f\in W_1$ and $g$ is a modular form, then $fg\in W_1$. 

The equation  (\ref{modularanomaly3.61}) can be equivalently written as
\begin{eqnarray} \label{generatingP}
\frac{\partial f} {\partial E_2 } =-\frac{1}{12}Q^2 (\partial_Q f ) \partial_Q P^{(0)},
\end{eqnarray}
where we define the generating function $P^{(0)} := \sum_{s=1}^{\infty} P^{(0)}_s Q^s $.  Suppose $f\in W_1$ and $g\in W_1$, then we have $fg\in W_1$ since 
\begin{eqnarray}
\frac{\partial (f g) } {\partial E_2 } =-\frac{1}{12}g Q^2 (\partial_Q f ) \partial_Q P^{(0)} -\frac{1}{12}f Q^2 (\partial_Q f ) \partial_Q P^{(0)}
=-\frac{1}{12} Q^2 (\partial_Q fg ) \partial_Q P^{(0)}.
\end{eqnarray}
Furthermore, If $f\in W_1$, then any algebraic function of $f$ would also satisfies the equation 
(\ref{generatingP}). However, $f$ may not have the structure of $Q$-expansion defined by $W_1$ 
space in (\ref{W1space}), and the coefficients $c_n$ may be only algebraic functions of quasi-modular forms. 

The equation (\ref{generatingP}) can be further written in a simpler way which will be useful later. 
In the left hand side of (\ref{generatingP}), we regard $f$ as function of $f(Q,E_2,E_4,E_6)$ when 
taking partial derivative with respect to $E_2$. Alternatively, we can also use inverse function and 
regard $Q$ as function of $Q(f,E_2,E_4,E_6)$.  In this perspective we have 
$\frac{\partial f} {\partial E_2 }=0$ and   $\frac{\partial Q} {\partial E_2 }\neq 0$. 
We can take partial derivative of $E_2$ on $f =\sum_{n=1}^{\infty} c_n Q^n$, and find 
\begin{eqnarray}
\sum_{n=1}^{\infty} (\partial_{E_2}c_n) Q^n + (\partial_Q f)( \partial_{E_2}Q) =0. 
\end{eqnarray}
So we can write equation (\ref{generatingP}) as
\begin{eqnarray} \label{simpler3.90}
\frac{\partial \log(Q)}{\partial E_2} =\frac{1}{12}Q \partial_Q P^{(0)}. 
\end{eqnarray}

Based on the perturbative calculations (\ref{pert3.60}), we conjecture 

\begin{conj}
\label{conjecture3.86}
$Z\in W_1$. 
\end{conj}

If our conjecture is correct, then obviously it is also true that $Z^n\in W_1$ 
for any positive integer $n$. Using the expansion (\ref{expan3.58}) and keeping track of  
the modular weights, we can see that the conjectures  (\ref{conjecture3.75}) (\ref{conjecture3.86}) 
and calculations in this subsection would imply the validity of the main conjecture  (\ref{conjecture3.49}).

\subsection{Proof of Vonjecture \ref{conjecture3.75} and Conjecture \ref{conjecture3.86}.}  \label{subsecproof3.5}

In this subsection we shall prove the conjectures (\ref{conjecture3.75}) and (\ref{conjecture3.86}) for our 
compact model. The proof uses some ingredients in the proof of genus zero modular anomaly equation for  the 
non-compact  half K3 model in \cite{HST, Klemm:2012}, but it is much more complicated for compact model. 

To solve the Picard-Fuchs equations (\ref{PF2.1}), it is convenient to define the variable $\tilde{q}$, which 
is related to $z_1$ by 
\begin{eqnarray} \label{tildeq3.92}
J(\tilde{q}) = \frac{1}{z_1(1-432z_1)}, 
\end{eqnarray}
with the asymptotic behavior $\tilde{q} \sim z_1\sim 0$ in small $z_1$ limit. 
This is the Kahler parameter in local large base limit $z_2\rightarrow 0$, but 
in the compact model, it is different from the physical K\"ahler parameter $q_1$ which 
depends on both $z_1$ and $z_2$.  To be clear in this subsection we will keep the 
arguments $\tilde{q}$ and $q_1$ in the modular functions, and we also define the 
logarithmic flat coordinates $t_1 =\log(q_1)$ and $\tilde{t} = \log(\tilde{q})$. 

We can write the ansatz for the solution to the equations (\ref{PF2.1}) in small $z_2$ power series  expansion 
\begin{eqnarray}
w= \sum_{n=0}^{\infty} c_n(z_1) z_2^n
\end{eqnarray}
It is well known in the local $z_2\rightarrow 0$ limit, there are two linearly 
independent solutions, which can be written  as $E_4(\tilde{q})^{\frac{1}{4}}$ and 
$\log(\tilde{q}) E_4(\tilde{q})^{\frac{1}{4}}$, in terms of the modular parameter 
$\tilde{q}$ defined in (\ref{tildeq3.92}). So the initial function 
$c_0(z_1)$ in the expansion is $E_4(\tilde{q})^{\frac{1}{4}}$ or $\log(\tilde{q}) E_4(\tilde{q})^{\frac{1}{4}}$. 

The first and second PF equations in (\ref{PF2.1}) become 
\begin{eqnarray}
&& 12z_1 (6\theta_1+1) (6\theta_1+5) c_n(z_1) =  \theta_1(\theta_1-3n) c_n(z_1) ,
\nonumber \\
&& c_{n+1}(z_1) = \frac{1}{(n+1)^3} \prod_{i=0}^2 (\theta_1-3n - i) c_n (z_1)
\end{eqnarray}
where $\theta_{z_1}=z_1\partial_{z_1}$. We see that the second equation provides 
a recursion relation to compute the higher order coefficients $c_n(z_1)$. Furthermore, 
one can check that the two equations are consistent, i.e. the $c_{n+1}$ from 
recursion in the second equation also satisfies the first equation for $c_n$ 
with by replacing $n$ with $n+1$. The recursion implies 
\begin{eqnarray}
c_n (z_1) = \frac{1}{(n!)^3} \prod_{k=0}^{3n-1} (\theta_{z_1} -k) c_0(z_1) 
\end{eqnarray}

A particular modularity structure for such recursion was derived in \cite{HST, Klemm:2012}. 
Here we repeat the analysis. From the relation  (\ref{tildeq3.92}) it is easy to compute 
\begin{eqnarray}
1-864z_1 = \frac{E_6(\tilde{q}) }{E_4(\tilde{q})^{\frac{3}{2}} }, ~~~~
\theta_{z_1} = \frac{2E_4(\tilde{q}) }{E_4(\tilde{q})^{\frac{3}{2}}+ E_6(\tilde{q}) } \partial_{\tilde{t}}.
\end{eqnarray}
We can write some derivative formulae 
\begin{eqnarray}
\partial_{\tilde{t}}  \frac{E_4(\tilde{q})^{\frac{3}{2}}}{E_4(\tilde{q})^{\frac{3}{2}}+E_6 (\tilde{q})} &=& \frac{E_4(\tilde{q})^{\frac{1}{2}}}{2}
\frac{E_4(\tilde{q})^{\frac{3}{2}}-E_6(\tilde{q})}{E_4(\tilde{q})^{\frac{3}{2}}+E_6 (\tilde{q}) }, \nonumber \\ 
\partial_{\tilde{t}} \frac{E_2(\tilde{q})}{E_4(\tilde{q})^{\frac{1}{2}}} &=& - \frac{ E_2(\tilde{q})^2E_4(\tilde{q})-2E_2(\tilde{q})
E_6(\tilde{q}) +E_4(\tilde{q})^2 }{12E_4(\tilde{q})^{\frac{3}{2}}},  \nonumber \\
\partial_{\tilde{t}} \frac{E_6(\tilde{q})}{E_4(\tilde{q})^{\frac{3}{2}}} &=& - \frac{ E_4(\tilde{q})^3 -E_6(\tilde{q}) ^2 }{2E_4(\tilde{q})^{\frac{5}{2}}},  \nonumber \\
\partial_{\tilde{t}} E_4(\tilde{q})^{\frac{1}{4}}&=&  \frac{ E_2(\tilde{q})E_4(\tilde{q}) -E_6(\tilde{q})  } {12E_4(\tilde{q})^{\frac{3}{4}}}. 
\label{derivative3.97}
\end{eqnarray}

First we consider the case $c_0(z_1)=E_4(\tilde{q})^{\frac{1}{4}}$. From the derivative  formulae (\ref{derivative3.97}) 
one can deduce inductively that the recursion function 
\begin{eqnarray}
\prod_{k=0}^{n-1} (\theta_{z_1} -k) E_4(\tilde{q})^{\frac{1}{4}}  \sim  \frac{ E_4(\tilde{q})^{\frac{1}{4}} }{(E_4(\tilde{q})^{\frac{3}{2}}+E_6 (\tilde{q}))^{n}} P_{6n} (E_2, E_4, E_6), 
\end{eqnarray}
where $P_{6n} (E_2, E_4, E_6)$ is a quasi-modular form of weight $6n$ and is linear in $E_2(\tilde{q})$. We can write 
\begin{eqnarray} \label{dn3.100}
\prod_{k=0}^{n-1} (\theta_{z_1} -k) E_4(\tilde{q})^{\frac{1}{4}}  =  \frac{  E_4(\tilde{q})^{\frac{1}{4}} }{(E_4(\tilde{q})^{\frac{3}{2}}+E_6 (\tilde{q}))^{n}} (a_n E_2(\tilde{q}) + b_n) , 
\end{eqnarray}
where $a_n$ and $b_n$ are modular forms of weight $6n-2$ and $6n$, without $E_2$ dependence. Some low order formulae are the followings 
\begin{eqnarray}
&& a_0=0, ~~~ b_0=1, ~~~a_1 =\frac{E_4(\tilde{q})}{6}, ~~~ b_1 = -\frac{E_6(\tilde{q})}{6}, \nonumber \\
&& a_2 =- \frac{E_4(\tilde{q})  E_6 (\tilde{q})}{3}, ~~~~ b_2= \frac{5E_4(\tilde{q})^3+  7E_6 (\tilde{q})^2}{36}, \nonumber \\
&& a_3 = \frac{E_4(\tilde{q})  [77 E_4(\tilde{q}) ^3+211 E_6 (\tilde{q})^2] }{216}, ~~~
b_3 = -\frac{E_6(\tilde{q})  [197 E_4(\tilde{q}) ^3+91 E_6 (\tilde{q})^2 ]}{216}. \nonumber
\end{eqnarray}
The recursion relation for the modular forms $a_n$ and $b_n$ are 
\begin{eqnarray} \label{recursion3.100}
a_{n+1} &=& -(2n-\frac{1}{6}) E_6(\tilde{q}) a_n +\frac{1}{6} E_4 (\tilde{q}) b_n + 2E_4(\tilde{q})^{\frac{3n+1}{2}}\partial_{\tilde{t}} \big[ \frac{a_n}{E_4(\tilde{q})^{\frac{3n-1}{2}}} \big], \nonumber \\
b_{n+1} &=& -(2n+\frac{1}{6}) E_6(\tilde{q}) b_n -\frac{1}{6} E_4 (\tilde{q})^2 a_n + 2E_4(\tilde{q})^{\frac{3n+2}{2}}\partial_{\tilde{t}} \big[ \frac{b_n}{E_4(\tilde{q})^{\frac{3n}{2}}} \big]. 
\end{eqnarray}

We then consider the case $c_0(z_1)=\tilde{t} E_4(\tilde{q})^{\frac{1}{4}}$. From the recursion relation (\ref{recursion3.100}) we can check inductively that 
\begin{eqnarray}
\prod_{k=0}^{n-1} (\theta_{z_1} -k) [\tilde{t} E_4(\tilde{q})^{\frac{1}{4}}]  =  \frac{  E_4(\tilde{q})^{\frac{1}{4}} }{(E_4(\tilde{q})^{\frac{3}{2}}+E_6 (\tilde{q}))^{n}} \big[  \tilde{t}  (a_n E_2(\tilde{q}) + b_n) 
+  12 a_n \big]
\end{eqnarray}

So we find the power series solution and one logarithmic $\log(z_1)$ solution to the Picard-Fuchs equations (\ref{PF2.1})
\begin{eqnarray} \label{solPF3.103}
X_0 &=& E_4(\tilde{q}) ^{\frac{1}{4}}\big[ 1+ \sum_{n=1}^{\infty} \frac{a_{3n} E_2(\tilde{q}) +b_{3n}}{2^{3n} (n!)^3 E_4(\tilde{q})^{\frac{9n}{2}} } \frac{z_2^n}{(1-432z_1)^{3n}}  \big],  \nonumber \\
X_1 &=& E_4(\tilde{q}) ^{\frac{1}{4}}\big[ \tilde{t}+ \sum_{n=1}^{\infty} \frac{ \tilde{t}( a_{3n} E_2(\tilde{q}) +b_{3n}) +12a_{3n} }{2^{3n} (n!)^3 E_4(\tilde{q})^{\frac{9n}{2}} } \frac{z_2^n}{(1-432z_1)^{3n}}  \big], 
\end{eqnarray} 
We compute the flat coordinate of the compact model 
\begin{eqnarray} \label{deformation3.103}
t_1&=& \frac{X_1}{X_0}   \\ \nonumber 
&=& \tilde{t} + \big[  \sum_{n=1}^{\infty} \frac{ 12a_{3n} }{2^{3n} (n!)^3 E_4(\tilde{q})^{\frac{9n}{2}} } 
\frac{z_2^n}{(1-432z_1)^{3n}}  \big]/ \big[ 1+ \sum_{n=1}^{\infty} \frac{a_{3n} E_2(\tilde{q}) +b_{3n}}{2^{3n} (n!)^3 E_4(\tilde{q})^{\frac{9n}{2}} } \frac{z_2^n}{(1-432z_1)^{3n}}  \big]. 
\end{eqnarray} 
The expression quantifies the deformation of $t_1$ away from $\tilde{t}$ as asymptotic expansion in small $z_2$, where the coefficients are rational functions of quasi-modular forms of $\tilde{t}$
\begin{eqnarray} \label{expan3.104}
t_1-\tilde{t} = \sum_{n=1}^{\infty} \frac{P_{18n-2}(\tilde{q}) }{E_4(\tilde{q})^{\frac{9n}{2}}}  \frac{z_2^n}{(1-432z_1)^{3n}},
\end{eqnarray}
where $P_{18n-2}(\tilde{q})$ are some quasi-modular forms of weight $18n-2$. We will need the following nice formula 
\begin{eqnarray} \label{nice3.104}
\partial_{E_2(\tilde{q})} ( t_1-\tilde{t})^{-1}  = \frac{1}{12} .
\end{eqnarray}
Here we regard $z_2$ as free variable which is independent of $\tilde{q}$, so that $\partial_{E_2(\tilde{q})} z_2=0$.  
 
Before we proceed further, we shall prove a useful general formula on $E_2$ derivative.  
Suppose $P_k$ is a rational function of quasi-modular forms, with modular weight $k$, then we have the formula 
\begin{eqnarray}  \label{formu3.105}
\partial_{E_2} \partial_t P_k = \partial_t \partial_{E_2}  P_k +\frac{k}{12} P_k,
\end{eqnarray} 
where we use the notation $t=\log(q) =2\pi i \tau$. A similar formula was used in 
\cite{Huang:2012} to derive of holomorphic anomaly equation in the Nekrasov-Shatashvili limit, 
and in \cite{Huang:2013eja}  to studying the relation between holomorphic and modular anomaly 
equations in the Seiberg-Witten theory. The derivation of (\ref{formu3.105}) uses the covariant 
derivative.  It is well known that $E_2$ is not a modular form, but with an an-holomorphic  
shift $\hat{E}_2=  E_2 -\frac{6i}{\pi (\tau-\bar{\tau})}$ it becomes a  modular form. To preserve 
the almost holomorphic structure, we shall use the covariant or Maass derivative 
\begin{eqnarray}
D_\tau = \partial_\tau+ \frac{k}{(\tau-\bar{\tau})},
\end{eqnarray}
where $k$ is the modular weight. The derivative with respect to $E_2$ can be related to the 
anti-holomorphic derivative 
\begin{eqnarray}
\bar{\partial}_{\bar{\tau}} = (\bar{\partial}_{\bar{\tau}} \hat{E}_2) \partial_{\hat{E}_2} 
=\frac{6}{\pi i (\tau-\bar{\tau})^2} \partial_{\hat{E}_2}
\end{eqnarray}
We can compute the derivative 
\begin{eqnarray}
\bar{\partial}_{\bar{\tau}} D_\tau P_k =( \partial_\tau + \frac{k}{ (\tau-\bar{\tau})} ) \bar{\partial}_{\bar{\tau}} 
 P_k + (\bar{\partial}_{\bar{\tau}} \frac{k}{(\tau-\bar{\tau})} ) P_k
\end{eqnarray} 
Taking the holomorphic limit and cancel out the infinitesimal factor $(\tau-\bar{\tau})^{-2}$ we arrive at the formula (\ref{formu3.105}). 

The above formula (\ref{formu3.105}) is valid for any rational function $P_k$ of quasi-modular forms, 
with possible fractional powers, which satisfies the homogeneous relation with modular weight  
\begin{eqnarray} \label{homogeneous3.114}
P_k(\lambda^2 E_2, \lambda^4 E_4, \lambda^6 E_6) = \lambda^k  P_k(E_2,  E_4,  E_6). 
\end{eqnarray}
However it is a little tricky for the logarithmic function, which does not satisfy the above relation 
(\ref{homogeneous3.114}). Naively the logarithmic function has zero modular weight and the 
$E_2$ derivative should commute with $\tau$ derivative. But this is not the case, as we can 
easily compute with the formula  (\ref{formu3.105}),
\begin{eqnarray} \label{logarithmic3.115}
\partial_{E_2} \partial_t \log(P_k) = \partial_t \partial_{E_2}  \log(P_k) +\frac{k}{12}. 
\end{eqnarray} 
We will not need this logarithmic formula in this subsection, but it will be useful later.

We can take derivative $n$ times on a rational function of quasi-modular forms, and use formula (\ref{formu3.105}) to compute 
\begin{eqnarray} \label{formulae3.109}
\partial_{E_2} \partial_t^n P_k = \partial_t^n \partial_{E_2}  P_k +\frac{n(k+n-1)}{12}  \partial_t^{n-1}P_k
\end{eqnarray}

Suppose $P_k(q_1)$ is a rational function with possible fractional powers of quasi-modular forms of $q_1$, 
with modular weight $k$.  We can expand it as a power series of $\frac{z_2}{(1-432)^3}$ with coefficients as 
rational functions of (quasi)-modular forms of $\tilde{q}$, using the expansion (\ref{deformation3.103}) 
\begin{eqnarray} \label{expan3.111}
P_k(q_1) = \sum_{n=0}^{\infty}   \frac{(t_1-\tilde{t})^n}{n!} \partial_{\tilde{t}}^n P_k(\tilde{q}) ,
\end{eqnarray} 

We can compute the derivative of $E_2(\tilde{q})$ and use the two formulae (\ref{nice3.104}, \ref{formulae3.109})
\begin{eqnarray} \label{noE23.112}
\partial_{E_2(\tilde{q})} P_k(q_1)   &=& 
 \sum_{n=0}^{\infty}  \{ \frac{(t_1-\tilde{t})^{n-1}}{(n-1)!}  [\partial_{E_2(\tilde{q})}(t_1-\tilde{t}) ]   
 \partial_{\tilde{t}}^n P_k(\tilde{q})+ \frac{(t_1-\tilde{t})^n}{n!} \partial_{E_2(\tilde{q})} \partial_{\tilde{t}}^n P_k(\tilde{q}) \} \nonumber \\
 &=&  \sum_{n=0}^{\infty}  \{ -\frac{1}{12} \frac{(t_1-\tilde{t})^{n+1}}{(n-1)!}   \partial_{\tilde{t}}^n P_k(\tilde{q})  + 
 \frac{(t_1-\tilde{t})^n}{n!} \big[  \partial_{\tilde{t}}^n \partial_{E_2(\tilde{q})} P_k(\tilde{q})\nonumber \\ && + \frac{n(k+n-1)}{12}  \partial_{\tilde{t}}^{n-1} P_k(\tilde{q}) \big]\}, \nonumber \\
 &=& \partial_{E_2(q_1)} P_k(q_1) + \frac{k}{12} (t_1-\tilde{t}) P_k(q_1) . 
\end{eqnarray}
We see the derivative vanishes when $P_k (q_1)$ has zero modular weight $k=0$ and is modular, i.e. no $E_2$ dependence.

Now we are have the necessary ingredients for proving the conjecture (\ref{conjecture3.75}). 
First we consider the case of $1-864z_1=\frac{E_6(\tilde{q}) }{E_4(\tilde{q})^{\frac{3}{2}}}$. 
We shall expand in terms of the modular parameter $q_1$ of the compact model 
\begin{eqnarray} \label{expan3.118}
\frac{E_6(\tilde{q}) }{E_4(\tilde{q})^{\frac{3}{2}}} = \frac{1}{E_4(q_1)^{\frac{3}{2}}} \big[ E_6(q_1) + \sum_{n=1}^{\infty} \frac{f_n(q_1)}{E_4(q_1)^{\frac{9n}{2}} } \frac{z_2^n}{(1-432z_1)^{3n}}  \big] , 
\end{eqnarray}
We can determine the coefficients $f_n$ recursively by expanding the modular forms of $q_1$ in terms of $\tilde{q}$ 
and compare with the left hand side. It is clear that $f_n(q_1)$ are rational functions of 
quasi-modular forms of $q_1$ of modular weight $18n+6$, with possible negative powers of $E_4$, 
but no fractional powers. To show it has no $E_2$ dependence, we take derivative of $E_2(\tilde{q})$ 
on both sides and use the formula (\ref{noE23.112})
\begin{eqnarray}
0 = \sum_{n=1}^{\infty} \frac{\partial_{E_2(q_1)} f_n(q_1)}{E_4(q_1)^{\frac{9n+3}{2}} } \frac{z_2^n}{(1-432z_1)^{3n}},  
\end{eqnarray}
where we have use the fact that the coefficients have zero modular weight. Since $z_2$ is a free variable, we find $\partial_{E_2(q_1)} f_n(q_1)=0$, i.e.  it is modular. 

We also need to check that there is no negative powers of $E_4(q_1)$ in  $f_n(q_1)$. 
{}From the expansion (\ref{expan3.104}) we see that the coefficients of $\frac{z_2^n}{(1-432z_1)^{3n}}$ 
always have the factor of $E_4(q_1)^{-\frac{9n}{2}}$. There could be also negative 
powers of $E_4$ from the derivative $\partial_{\tilde{t}}$ in  (\ref{expan3.111}) 
since each derivative increases the negative power of $E_4$ by one. However, we note that 
since $a_n$ has modular weight is $6n-2$, it can not be purely a power of 
$E_6(\tilde{q})$ and always has a factor $E_4(\tilde{q})$. So the 
$(t_1-\tilde{t})^n$ factor in (\ref{expan3.111}) contributes a factor $E_4(\tilde{q})^n$ 
which cancel the one from taking derivative. Therefore $f_n(q_1)$ is a modular form. 

Furthermore, since the derivative $\partial_{\tilde{t}}^n [\frac{E_6(\tilde{q}) }{E_4(\tilde{q})^{\frac{3}{2}}}]$ 
always contains a factor of $E_4(\tilde{q})^{3}-E_6(\tilde{q})^{2}$, we see that the the coefficient  
$f_n(q_1)$ always have a factor of $E_4({q_1})^{3}-E_6({q_1})^{2}$, as confirmed by the explicit 
expansion in (\ref{expan3.56}). According to the discussion in equations (\ref{eta3.82}, \ref{eta3.83}), 
we immediately infer the modularity of the expansion coefficients of  $z_1(1-432z_1)$.

Alternatively, we can also expand  $z_1(1-432z_1)=\frac{E_4(\tilde{q})^{3}-E_6(\tilde{q})^{2} }{1728E_4(\tilde{q})^{3}}$ 
in the same way as $1-864z_1$ in (\ref{expan3.118}). From the same reasoning as above we can also see the modularity 
of the expansion (\ref{expan3.54}). We  note that the derivative $\partial_{\tilde{t}}^n [E_4(\tilde{q})^{3}-E_6(\tilde{q})^{2}]$ 
also always contains a factor of $E_4(\tilde{q})^{3}-E_6(\tilde{q})^{2}$, and it is factored out in this case in the 
expansion (\ref{expan3.54}), so that the higher order coefficients do not necessarily contain the factor.  

Finally we consider the case of $X_0$, which is a little trickier than the previous two cases since the 
modular weight is not zero. Again we make the ansatz in terms of the modular parameter $q_1$ of the compact model
\begin{eqnarray} \label{ansatz3.115}
X_0=   E_4(q_1)^{\frac{1}{4}} \big[ 1 + \sum_{n=1}^{\infty} \frac{f_n(q_1)}{E_4(q_1)^{\frac{9n}{2}} } \frac{z_2^n}{(1-432z_1)^{3n}}  \big], 
\end{eqnarray}
and determine $f_n(q_1)$ recursively comparing with the solution (\ref{solPF3.103}) and using the 
deformation (\ref{expan3.104}).  From the same  reasonings above we see $f_n(q_1)$ contains no 
negative power of $E_4(q_1)$ and no fractional powers, so it is a quasi-modular form. We only need to show 
it is $E_2$ independent. We compute the $E_2(\tilde{q})$ derivative in the $X_0$ solution in (\ref{solPF3.103}) 
\begin{eqnarray} \label{result3.117}
\partial_{E_2(\tilde{q})} X_0 &=& E_4(\tilde{q}) ^{\frac{1}{4}}\big[ \sum_{n=1}^{\infty} \frac{a_{3n} }{2^{3n} (n!)^3 E_4(\tilde{q})^{\frac{9n}{2}} } \frac{z_2^n}{(1-432z_1)^{3n}}  \big] \nonumber \\
&=& \frac{t_1-\tilde{t}}{12} X_0 , 
\end{eqnarray}
where in the last step we use (\ref{deformation3.103}). On the other hand, we can also compute the 
$E_2(\tilde{q})$ derivative using the ansatz (\ref{ansatz3.115}) and the formula (\ref{noE23.112}) 
with the modular weight $k=1$, and we get 
\begin{eqnarray} \label{result3.118}
\partial_{E_2(\tilde{q})} X_0 = \frac{t_1-\tilde{t}}{12} X_0 + \sum_{n=1}^{\infty} \frac{\partial_{E_2(q_1)} f_n(q_1)}{E_4(q_1)^{\frac{18n-1}{4}} } \frac{z_2^n}{(1-432z_1)^{3n}} . 
\end{eqnarray}
Comparing the two results (\ref{result3.117}, \ref{result3.118}) we find that $f_n(q_1)$ has no 
$E_2$ dependence. Thus we have proven conjecture (\ref{conjecture3.75}). We see the case of $X_0$ is quite 
interesting. It is quasi-modular when we use the $\tilde{q}$ parameter, but the $E_2$ dependence vanishes and it 
becomes purely modular when we use the physical K\"ahler parameter $q_1$ of the compact model. 

We can also show that the modular forms $f_n(q_1)$ in (\ref{ansatz3.115}) contain a factor 
of $\eta^{24}\sim E_4^3-E_6^2$, as can be seen from the explicit formulae in the expansion 
(\ref{expan3.57}). Since  $f_n(q_1)$ is a modular form, it is sufficient to show that $f_n(q_1)=0$ 
at $q_1=0$ in order to be divisible by $\eta^{24}(q_1)$. We note that the derivative of any rational 
function of modular forms and their fractional powers vanish at $q_1$, as long as  it is not 
singular at $q_1=0$, due to its $q$-series expansion. So the contributions to $f_n(0)$ from 
lower order terms in (\ref{ansatz3.115}) vanish. We also need to evaluate the solution 
(\ref{solPF3.103}) for $X_0$ at $\tilde{q}=0$.  We denote the value of $a_n(\tilde{q}), b_n(\tilde{q})$ 
in the solution at $\tilde{q}=0$ simply as $a_n(0), b_n(0)$. The recursion relation (\ref{recursion3.100}) at $\tilde{q}=0$ is 
\begin{eqnarray} 
a_{n+1}(0) &=& -(2n-\frac{1}{6}) a_n(0) +\frac{1}{6}  b_n(0) ,   \nonumber \\
b_{n+1}(0) &=& -(2n+\frac{1}{6})  b_n(0) -\frac{1}{6}  a_n(0)  .  
\end{eqnarray}
Adding together we find 
\begin{eqnarray}
a_{n+1}(0) + b_{n+1}(0) = -2n[a_n(0)+b_n(0)] ,
\end{eqnarray}
so the sum vanish $a_n(0)+b_n(0)=0$ for $n\geq 1$. The solution is simply $X_0=1$ for $\tilde{q}=0$,  
and does not contribute to higher order terms either. Thus $f_n(q_1)$ vanish at $q_1=0$, i.e.  is divisible by $\eta^{24}(q_1)$.

Next we would like to find the solution to Picard-Fuchs equation with leading $\log(z_2)$ behavior.  We make the ansatz 
\begin{eqnarray} \label{ansatz3.119}
X_2 = X_0 [ \log(z_2) - \frac{3}{2} t_1 - \frac{3}{2} \log(\frac{1-432z_1}{z_1})] + \sum_{n=0}^{\infty} \alpha_n(z_1) z_2^n,
\end{eqnarray}
where $X_0$ and $X_0t_1=X_1$ are the two solutions found previously in (\ref{solPF3.103}). The terms 
$\frac{3}{2} t_1 + \frac{3}{2} \log(\frac{1-432z_1}{z_1})$ comes from (\ref{xi3.63}) which is the 
solution to the $z_2$ zero order  equation with the correct leading asymptotic behavior,  so that we 
can choose the initial function $\alpha_0(z_1)=0$. For convenience, we denote 
\begin{eqnarray} 
d_n = \frac{  E_4(\tilde{q})^{\frac{1}{4}} }{(E_4(\tilde{q})^{\frac{3}{2}}+E_6 (\tilde{q}))^{n}} (a_n E_2(\tilde{q}) + b_n), 
\end{eqnarray}
where $a_n$, $b_n$ are modular forms defined in (\ref{dn3.100}), and satisfy the recursion relation (\ref{recursion3.100}), so that the power series solution 
\begin{eqnarray}
X_0 =  \sum_{n=0}^{\infty} \frac{d_{3n}}{(n!)^3} z_2^n. 
\end{eqnarray}
Again after some calculations, we find recursion relation for $\alpha_n(z_1)$ in (\ref{ansatz3.119}) from the second Picard-Fuchs equation
\begin{eqnarray} \label{recursion3.124}
\alpha_{n+1}(z_1) &=& \frac{1}{(n+1)^3} \prod_{i=0}^2 (\theta_1 -3n-i) \alpha_n(z_1)-\frac{3d_{3n+3}  }{(n+1) (n+1)!^3} -\frac{9E_6(\tilde{q}) d_{3n+2}}{(n+1)!^3(E_4(\tilde{q})^{\frac{3}{2}}+E_6(\tilde{q})) } 
\nonumber \\
&& +\frac{9(E_4(\tilde{q})^3+E_6(\tilde{q})^2) d_{3n+1}}{(n+1)!^3 (E_4(\tilde{q})^{\frac{3}{2}}+E_6(\tilde{q}) )^2}  
-\frac{6E_6(\tilde{q}) (3E_4(\tilde{q})^3+E_6(\tilde{q})^2) d_{3n}}{(n+1)!^3 (E_4(\tilde{q})^{\frac{3}{2}}+E_6(\tilde{q}) )^3}  .
\end{eqnarray} 
{}From the recursion relation and initial condition $\alpha_0=0$, it is clear that the coefficients $\alpha_n(z_1)$ have the following structure 
\begin{eqnarray} \label{definecn3.125}
\alpha_n =\frac{  E_4(\tilde{q})^{\frac{1}{4}} }{(E_4(\tilde{q})^{\frac{3}{2}}+E_6 (\tilde{q}))^{3n}} (A_n(\tilde{q}) E_2(\tilde{q}) + B_n(\tilde{q})),
\end{eqnarray} 
where $A_n(\tilde{q}), B_n(\tilde{q})$ are modular forms of $\tilde{q}$ of weight $18n-2, 18n$. Some low order formulae are 
\begin{eqnarray}
&& A_1 = \frac{E_4 (31E_4^3+113E_6^2) }{72}, ~~~ B_1 = -\frac{E_6 (1297 E_4^3+575 E_6^2) }{72},
\nonumber \\ &&
A_2= -\frac{E_4 E_6}{6912} (233753 E_4^6+812750 E_4^3 E_6^2+252953 E_6^4),
\\ \nonumber  &&
B_2 = \frac{1}{82944} (3717955 E_4^9+59441331 E_4^6 E_6^2+66006225 E_4^3 E_6^4+5867321 E_6^6). 
\end{eqnarray}

The second flat coordinate is then
\begin{eqnarray}
t_2 = \frac{X_2}{X_0} =\log(z_2) - \frac{3}{2} t_1 - \frac{3}{2} \log(\frac{1-432z_1}{z_1}) +\frac{1}{X_0} \sum_{n=1}^{\infty} \alpha_n(z_1) z_2^n,
\end{eqnarray}
where $\alpha_n$ are expressed in (\ref{definecn3.125}) in terms of the modular forms $A_n$ and $B_n$. It is clear we can write a series expansion 
\begin{eqnarray}
\log(Q) =\log(Z) +\sum_{n=1}^{\infty} f_n(q_1) Z^n ,
\end{eqnarray} 
where $Q=\frac{\exp(t_2+\frac{3t_2}{2})}{\eta(q_1)^{36}}$, $Z=\frac{z_2}{(1-432z_1)^3E_4(q_1)^{\frac{9}{2}}}$ are the variables 
defined in (\ref{defineT}, \ref{defineS3.75}). Using the relation between $t_1$ and $\tilde{t}$ as in equation 
(\ref{expan3.104}) we see that $f_n(q_1)$ here must be quasi-modular forms of weight $18n$. Because the coefficients 
of $z_2^n$ have zero modular weight, the derivative of $E_2(\tilde{q})$ and $E_2(q_1)$ are the the same for 
$\log(Q)$ according to the formula (\ref{noE23.112}). In order to prove conjecture (\ref{conjecture3.86}), 
we should confirm  the equation (\ref{simpler3.90}). We compute the $E_2$ derivative 
\begin{eqnarray} \label{logQ3.129}
&& ~~\partial_{E_2(q_1)}   \log(Q) =  \partial_{E_2(\tilde{q})}   \log(Q) \nonumber \\
&& =  -\frac{t_1-\tilde{t}}{12X_0}\sum _{n=1}^{\infty} \alpha_n(z_1) z_2^n +\frac{1}{X_0} \sum _{n=1}^{\infty} 
\frac{  E_4(\tilde{q})^{\frac{1}{4}} A_n(\tilde{q})  }{(E_4(\tilde{q})^{\frac{3}{2}}+E_6 (\tilde{q}))^{3n}} z_2^n, 
\end{eqnarray} 
where we have used (\ref{result3.117}). We shall prove the above equation is equal to the 
$\frac{1}{12}\partial_{t_2} \mathcal{F}^{(0)}_{inst}(t_1,t_2)$, where $\mathcal{F}^{(0)}_{inst}$ 
is the instanton part of the the prepotential. 

The classical contribution in the prepotential for our model is 
\begin{eqnarray} \label{Fclassical3.130}
\mathcal{F}^{(0)}_{classical} = \frac{3t_1^3}{2} + \frac{3}{2}t_1^2t_2 +\frac{1}{2} t_1t_2^2
\end{eqnarray}
The prepotential is determined by the fact that $X_0\partial_{t_i}\mathcal{F}^{(0)}$ are the double 
logarithmic solutions of the Picard-Fuchs (PF) equations (\ref{PF2.1}). In order to confirm  
the equation (\ref{simpler3.90}), we shall show the following is a solution to the PF equations 
\begin{eqnarray} \label{w33.151}
w_3&=& X_0 (\frac{3}{2}t_1^2+ t_1t_2)  - (t_1-\tilde{t})\sum _{n=1}^{\infty} \alpha_n(z_1) z_2^n +12\sum _{n=1}^{\infty} 
\frac{  E_4(\tilde{q})^{\frac{1}{4}} A_n(\tilde{q})  }{(E_4(\tilde{q})^{\frac{3}{2}}+E_6 (\tilde{q}))^{3n}} z_2^n
\\ \nonumber 
&=& X_1[\log(z_2)-\frac{3}{2}\log(\frac{1-432z_1}{z_1})]+ \tilde{t}\sum _{n=1}^{\infty} \alpha_n(z_1) z_2^n 
+12\sum _{n=1}^{\infty} \frac{  E_4(\tilde{q})^{\frac{1}{4}} A_n(\tilde{q})  }{(E_4(\tilde{q})^{\frac{3}{2}}+E_6 (\tilde{q}))^{3n}} z_2^n,
\end{eqnarray}
where we have use the solution for $X_2$ in (\ref{ansatz3.119}). 

After some tedious calculations, we check that (\ref{w33.151}) is indeed a solution. We can provide some details in the followings. 
At zero order $z_2$, we need to check the equation from the first PF operator  
\begin{eqnarray}
[\theta_1^2 -12z_1 (6\theta_1+1)(6\theta_1+5)] [-\frac{3}{2} \tilde{t} E_4(\tilde{q})^{\frac{1}{4}} \log(\frac{1-432z_1}{z_1})]
=3\theta_1[ \tilde{t} E_4(\tilde{q})^{\frac{1}{4}} ]. 
\end{eqnarray}
Assuming the two Picard-Fuchs equations are consistent, we only need to further check the second PF equation which provide recursion 
relation for the coefficients of $z_2^n$. We denote the coefficient of $z_2^n$ in the solution (\ref{w33.151})  as
\begin{eqnarray} \label{beta3.133}
\beta_n = \tilde{t} \alpha_n + \frac{ 12 E_4(\tilde{q})^{\frac{1}{4}} A_n(\tilde{q})  }{(E_4(\tilde{q})^{\frac{3}{2}}+E_6 (\tilde{q}))^{3n}} . 
\end{eqnarray}
We see the similar structure for the ansatz of $X_2$ in (\ref{ansatz3.119}). From the same calculations we derive the recursion relation for 
$\beta_n(z_1)$ from the second Picard-Fuchs equation, similar to (\ref{recursion3.124}), 
\begin{eqnarray} \label{recursionbeta3.134}
\beta_{n+1}(z_1) &=& \frac{1}{(n+1)^3} \prod_{i=0}^2 (\theta_1 -3n-i) \beta_n(z_1)-\frac{3e_{3n+3}  }{(n+1) (n+1)!^3} -
\frac{9E_6(\tilde{q}) e_{3n+2}}{(n+1)!^3(E_4(\tilde{q})^{\frac{3}{2}}+E_6(\tilde{q})) } \nonumber \\
&& +\frac{9(E_4(\tilde{q})^3+E_6(\tilde{q})^2) e_{3n+1}}{(n+1)!^3 (E_4(\tilde{q})^{\frac{3}{2}}+E_6(\tilde{q}) )^2}  
-\frac{6E_6(\tilde{q}) (3E_4(\tilde{q})^3+E_6(\tilde{q})^2) e_{3n}}{(n+1)!^3 (E_4(\tilde{q})^{\frac{3}{2}}+E_6(\tilde{q}) )^3}  ,
\end{eqnarray}
where we denote
\begin{eqnarray}
e_n = \frac{  E_4(\tilde{q})^{\frac{1}{4}} }{(E_4(\tilde{q})^{\frac{3}{2}}+E_6 (\tilde{q}))^{n}} [\tilde{t} (a_n E_2(\tilde{q}) + b_n) +12a_n], 
\end{eqnarray}
so that the solution $X_1 = \sum_{n=0}^{\infty} \frac{e_{3n}}{n!^3} z_2^n$.  We further denote 
\begin{eqnarray}
\frac{E_4(\tilde{q})^{\frac{1}{4}}  [A_n^{(k)} E_2(\tilde{q}) +B_n^{(k)}] } {(E_4(\tilde{q})^{\frac{3}{2}}+E_6 (\tilde{q}))^{3n+k+1}} = \prod_{i=0}^k (\theta_1 -3n-i) \alpha_n , 
\end{eqnarray} 
where $A_n^{(k)}$ and $B_n^{(k)}$ are modular forms of $\tilde{q}$ of weight $6(3n+k+1)-2$ and $6(3n+k+1)$. Then similar to the case of solution $X_1$, one can show inductively 
\begin{eqnarray}
\frac{E_4(\tilde{q})^{\frac{1}{4}}  [\tilde{t}(A_n^{(k)} E_2(\tilde{q}) +B_n^{(k)} )+ 12 A_n^{(k)} ]} {(E_4(\tilde{q})^{\frac{3}{2}}+E_6 (\tilde{q}))^{3n+k+1}} = 
\prod_{i=0}^k (\theta_1 -3n-i) (\tilde{t} \alpha_n + \frac{ 12 E_4(\tilde{q})^{\frac{1}{4}} A_n(\tilde{q})  }{(E_4(\tilde{q})^{\frac{3}{2}}+E_6 (\tilde{q}))^{3n}} ).  \nonumber 
\end{eqnarray} 
>From the recursion relation for $\alpha_n$ in  (\ref{recursion3.124}), we can write the recursion separately for $A_n$ and $B_n$ as 
\begin{eqnarray}
A_{n+1} &=& \frac{A_n^{(2)}}{(n+1)^3} -\frac{3a_{3n+3}  }{(n+1) (n+1)!^3} -\frac{9E_6(\tilde{q}) a_{3n+2}}{(n+1)!^3 }  +\frac{9(E_4(\tilde{q})^3+E_6(\tilde{q})^2) a_{3n+1}}{(n+1)!^3 }  \nonumber \\
&&
-\frac{6E_6(\tilde{q}) (3E_4(\tilde{q})^3+E_6(\tilde{q})^2) a_{3n}}{(n+1)!^3 }  , \nonumber \\
B_{n+1} &=& \frac{B_n^{(2)}}{(n+1)^3} -\frac{3b_{3n+3}  }{(n+1) (n+1)!^3} -\frac{9E_6(\tilde{q}) b_{3n+2}}{(n+1)!^3 }  +\frac{9(E_4(\tilde{q})^3+E_6(\tilde{q})^2) b_{3n+1}}{(n+1)!^3 }  \nonumber \\
&&
-\frac{6E_6(\tilde{q}) (3E_4(\tilde{q})^3+E_6(\tilde{q})^2) b_{3n}}{(n+1)!^3 }  .
\end{eqnarray}
Now it is straightforward to check that $\beta_n$ in (\ref{beta3.133}) satisfy the recursion relation 
(\ref{recursionbeta3.134}). This completes the proof of conjecture (\ref{conjecture3.86}).

\section{Expansion around a point on the conifold divisor}
\label{conifoldgap} 
We now study the expansion of the higher genus topological string amplitudes around a point on the 
conifold divisor in terms of the local flat coordinates, and use the gap condition to fix the holomorphic 
ambiguity. Since the involution symmetry exchanges the two conifold divisors, it is sufficient to 
consider the simpler one $\Delta_2 =1+27 z_2$. The simplest point to consider is the intersection 
point of the divisor $z_1=0$ and $\Delta_2 =1+27 z_2$. This is a normal intersection, 
and we can choose the local coordinates as 
\begin{eqnarray} \label{coordinate3.86}
(z_{c1}, z_{c2} ) = (z_1, z_2+\frac{1}{27})  
\end{eqnarray} 

We solve the PF equations (\ref{PF2.1}) perturbatively around this point. There should be 6 linearly independent 
solutions. We only need 3 of them with at most one logarithmic leading term in order to compute the flat 
coordinates. The three solutions are two power series solutions $X_0, X_1$ and one logarithmic solution $X_2$ as the followings 
\begin{eqnarray}
X_0 &=&  [1 + 60 z_{c1} + 13860 z_{c1}^2 + \frac{9529520 z_{c1}^3}{3} + 
 148728580 z_{c1}^4 - 569333004240 z_{c1}^5 \nonumber \\ && 
 - \frac{5035405372197200 z_{c1}^6}{9}  + \mathcal{O}(z_{c1}^7) ] + [24504480 z_{c1}^3 + 32125373280 z_{c1}^4  \nonumber \\ && +   27949074753600 z_{c1}^5   + 19093358660376000 z_{c1}^6+ \mathcal{O}(z_{c1}^7) ] z_{c2} \nonumber \\ &&  + 
 [15162373053828000 z_{c1}^6 + \mathcal{O}(z_{c1}^7) ] z_{c2}^2 + \mathcal{O}(z_{c2}^3)  \\
 X_1 &=& z_{c2} + (\frac{33}{2} + 270 z_{c1}) z_{c2}^2 + (327 + 7020 z_{c1} + 374220 z_{c1}^2) z_{c2}^3   
  + \mathcal{O}(z_{c2}^4),  \\ 
  X_2 &=& X_0\log(z_{c1}) + [372 z_{c1} + 98442 z_{c1}^2 + \frac{224976472 z_{c1}^3}{9} +  \frac{7098994823 z_{c1}^4}{3}  + \mathcal{O}(z_{c1}^5)] \nonumber \\ && 
  + [-540 z_{c1} - 124740 z_{c1}^2 + 132742656 z_{c1}^3 + 
    220660710516 z_{c1}^4+ \mathcal{O}(z_{c1}^5)]  z_{c2}  \nonumber \\ &&  
    + [27 - 4860 z_{c1} - 1683990 z_{c1}^2 - 606485880 z_{c1}^3 - 
    307198881990 z_{c1}^4+ \mathcal{O}(z_{c1}^5)] z_{c2}^2 \nonumber \\ &&  
  + \mathcal{O}(z_{c2}^3)
\end{eqnarray} 

The local flat coordinates are computed as the ratio of the solutions $t_{c1}= \frac{X_1}{X_0}$ and $t_{c2} = \frac{X_2}{X_0}$. 
For the power series solution, the flat coordinate goes like $t_{c1} \sim z_{c2} =z_2+\frac{1}{27}$, which is the expected asymptotic 
behavior of near a conifold point. For the logarithmic solution, the flat coordinate goes like $ t_{c2} \sim \log(z_{c1})$, 
which is similar to the behavior near large volume point, and we should use the exponential of the flat coordinate as the expansion 
parameter. So we define the expansion parameters in terms of flat coordinates 
\begin{eqnarray}
t_c:= t_{c1} = \frac{X_1}{X_0} , ~~~ q_c:= \exp(t_{c2}).
\end{eqnarray}

We shall compute the asymptotic expansion of the propagators $S^{i j}, S^{i}$, $S$, 
where we have written the propagators in $z_i$ coordinate and for simplicity we denote 
$S^{i j} = S^{z_iz_j}$ and  $S^{i}=S^{z_i}$. Here the propagators should simply 
transform as tensors. In the case of involution symmetry, there is a minus sign and 
shifts for $S^i$ and $S$ because the involution map the point to a different branch 
of moduli space. But this should not happen locally. The three point function also 
transforms as tensor without minus sign. The equations (\ref{propa2.6}) are 
coordinate invariant if we take the holomorphic ambiguities 
$h^{jk}_i, h^j_i, h_i, h_{ij}$ to transform as tensors, while the 
Christoffel symbol $\Gamma^k_{ij}$ and the holomorphic ambiguity 
$s^k_{ij}$ transforms with a shift according to the well known 
transformation rule of Christoffel symbol as in (\ref{Christoffeltrans}). 

We can compute the asymptotic expansion of the Christoffel symbol 
$\Gamma^{z_{ck}}_{z_{ci}z_{cj}}$ in the local coordinates (\ref{coordinate3.86}) 
by the well known formula from special geometry $\Gamma^{z_{ck}}_{z_{ci}z_{cj}} = 
\frac{\partial z_{ck}} {\partial t_{cl}} \frac{\partial^2 t_{cl}}{\partial z_{ci} \partial z_{cj}} $ 
in the holomorphic limit, and compute the propagators  $S^{z_{ci}z_{c j}}, S^{z_{ci}}$, 
$S$ according to (\ref{Christoffeltrans}). Here the transformation between the coordinates 
$z_{ci}$ and $z_i$ are particularly simple because the Jacobian transformation matrix 
$\frac{\partial z_{ci}}{\partial z_j}$ is an identity matrix. So the propagators 
$S^{i j}, S^{i}$, $S$ and holomorphic ambiguities $h^{jk}_i, h^j_i, h_i, h_{ij}, s^k_{ij}$ 
are actually invariant under the coordinate transformation (\ref{coordinate3.86}). 

We find that after including the well known scaling factor $X_0^{2g-2} $, the asymptotic 
expansion of the higher genus topological string amplitudes with the correct holomorphic 
ambiguity satisfy the gap condition near the point $(z_1,z_2)=(0,-\frac{1}{27})$. Specifically, we find 
\begin{eqnarray} \label{gap3.91} 
X_0^{2g-2} \mathcal{F}^{(g)} =  \frac{B_{2g}}{4g(g-1)3^{5(g-1)} } \frac{1}{t_c^{2(g-1)}} + \sum_{n=0}^\infty f_n(q_c) t_c^n, 
\end{eqnarray}
where $f_n(q_c)$ are some power series of $q_c$. 

Now we can study how the gap condition (\ref{gap3.91}) fixes holomorphic ambiguity 
at higher genus. First as in previous situations, suppose we have found a particular 
holomorphic ambiguity at genus $g$ such that the gap condition (\ref{gap3.91}) is satisfied. 
Then we may add an additional piece of rational function of $z_1, z_2$ to  $\mathcal{F}^{(g)}$ 
such that the gap condition (\ref{gap3.91}) is not affected. Since near the point $(z_1,z_2)=(0,-\frac{1}{27})$, 
we have the asymptotic behavior $X_0\sim 1$ and $t_c\sim z_{c2} \sim \Delta_2$, this additional 
rational function should have no pole at the conifold divisor  $\Delta_2$ in order to be 
regular at $t_c\sim 0$. Together with the involution symmetry as defined by the spaces in 
(\ref{Vspaces}), we consider the following linear space of holomorphic ambiguity
\begin{eqnarray} 
X^{(g,m)}  &:=& \{  f  ~|~  f \in \sum _{n=0}^{\infty} \oplus V_2^{(g,m,n)}  \textrm{ and has no pole at  $\Delta_2$  \}. }
\end{eqnarray}
In particular, in the case of $m=[\frac{19}{3}(g-1)]$, according to the proposition (\ref{proposition2.37}), 
the dimension of $X^{(g,[\frac{19}{3}(g-1)])}$ is exactly the remaining number of unknown constants in 
the genus $g$ holomorphic ambiguity after we impose both the fiber modularity and the gap conditions. 

Clearly if $f\in X^{(g,m)}$, then its involution transformation $\tilde{f} = (-1)^{g-1} f \in X^{(g,m)}$ as well. 
Since the involution transformation symmetry exchanges the two conifold divisors, we infer that $f$ must have no 
pole at $\Delta_1$ either, so it is actually a polynomial of $z_1, z_2$. Canceling the pole $\Delta_1^{2g-2}$ 
reduces the degree of $z_1$ in the numerator by $6(g-1)$. So the space can be determined as
\begin{eqnarray} \label{holoX3.94}
X^{(g,m)} = \left\{
\begin{array}{cl}
\sum_{n=0}^{\infty} \oplus (V_-^{(1,m-6(g-1), 2n)} \oplus V_+^{(1,m-6(g-1), 2n+1)} )   ,    & ~~  \textrm{if  $g$ is even}    ;   \\
\sum_{n=0}^{\infty} \oplus (V_+^{(1,m-6(g-1), 2n)} \oplus V_-^{(1,m-6(g-1), 2n+1)} )   ,    & ~~  \textrm{if  $g$ is odd}   .   
\end{array}    
\right.
\end{eqnarray} 
It is now easy to compute the dimension of $X^{(g,m)}$ with the formulae (\ref{V2dim2.28}). We can consider in particular the case of $m=[\frac{19}{3}(g-1)]$, and find  the number of remaining unknown constants at low genus  $\textrm{Dim}(X^{(g,[\frac{19}{3}(g-1)])}) =0, 1, 1, 1, 1, 2, 1, 2, 3, 2, 3$ for $g=2,3, \cdots, 12$ respectively. 

We can also estimate the asymptotic behavior of $\textrm{Dim}(X^{(g,[\frac{19}{3}(g-1)])})$ at large genus. Ignoring sub-leading contributions, we find 
\begin{eqnarray}
\textrm{Dim}(X^{(g,[\frac{19}{3}(g-1)])})  =\sum_{n=0}^{\frac{g-1}{9}} (\frac{g-1}{3}-3n) \sim \frac{g^2}{54} .
\end{eqnarray}

Since the elements of the space $X^{(g,m)}$ are regular at both conifold divisors $\Delta_1, \Delta_2$, 
they would not affect the gap conditions at any other points on the conifold divisors. For example, 
studying the expansion of topological string amplitudes around the intersection of 
$\Delta_1$ and $\Delta_2$ would not give more boundary conditions to fix the holomorphic 
ambiguity.  It is thus sufficient to consider  the gap condition around  only one point $(z_1,z_2)=(0,-\frac{1}{27})$ here. 

At low genus we can fix the remaining holomorphic ambiguity (\ref{holoX3.94}) by the boundary conditions 
at large volume point $(z_1,z_2)=(0,0)$, which include the known higher genus constant map contributions 
and some vanishing property of GV invariants. For example, the GV invariants vanish $n^g_{d_E,d_B}=0$ for 
base degree $d_B=0$ and $g\geq 2$ \cite{Alim:2012}. We see that only the holomorphic ambiguities with zero 
$z_2$ degree  in the space (\ref{holoX3.94}) contribute to the zero base degree GV invariants. 
So this property alone can eliminate the subspace $V_{\pm}^{(1, [\frac{g-1}{3}], 0)}$  in the 
remaining holomorphic ambiguities, which are now actually completely fixed up to genus $g\leq 9$.  
 
We consider also the local limit. The topological string amplitudes on local 
$\mathbb{P}^2$ model which has only the base class are solvable to arbitrary high genus.  
So we known the zero fiber degree GV invariants  $n^g_{d_E ,d_B}$ for $d_E=0$ from local 
model calculations, which may fix some more  holomorphic ambiguity for the compact model. 
However, at a closer look we find that the local limit actually does not help. It is easy 
to check that except one element which is the constant in  $V_{+}^{(1, m-6(g-1), 0)}$ 
or $1-864z_1$ in $V_{-}^{(1, m-6(g-1), 0)}$, all other elements in  $X^{(g,m)}$ are polynomially 
divisible by $z_1$, so they do not affect the zero fiber degree GV invariants. So the local limit 
only fix the constant or $1-864z_1$ in  $X^{(g,m)}$, which is already fixed by the 
known constant map contribution in Gromov-Witten theory. 
 
We can discuss a little more details on fixing the holomorphic ambiguity with GV invariants. 
{}From the definition of the spaces $V_{\pm}$ in (\ref{basisplus2.25}, \ref{basisplus2.26}), 
we find that the elements in   $V_{\pm}^{(1, m-6(g-1), n)}$ are polynomially 
divisible by $z_1^{3n}$, so they only affect the GV invariants $n^g_{d_E ,d_B}$ 
with fiber degree $d_E\geq 3n$. So in order to fix the elements in $V_{\pm}^{(1, m-6(g-1), n)}$ 
we need to use the GV invariants  $n^g_{d_E ,d_B}$ with fiber degree $d_E\geq 3n$ and $d_B\geq n$. 

These B-model considerations are combined with the use of weak Jacobi forms to fix the topological string amplitudes. Suppose we have the exact formulas up to a base degree $d_B$ which are valid for all genus and fiber class, then in the followings we shall show that we have sufficient boundary conditions to fix the B-model formula up to any genus no bigger than $9(d_B+1)$. To see this, we further define a vector space 
\begin{eqnarray} \label{spaceY.E.10}
Y^{(g,d_B)} = \sum_{n=d_B+1}^{\infty} \oplus V_{(-1)^{g+n+1}}^{(1,[\frac{g-1}{3}],n)},
\end{eqnarray}
where we denote the $\pm$  subscript universally as $(-1)^{g+n+1}$ depending on the sign. Here we have also set $m=[\frac{19(g-1)}{3}]$ in (\ref{holoX3.94}). The space $Y^{(g,d_B)}$ is exactly the space of remaining holomorphic ambiguities at genus $g$, after we impose the boundary conditions at the conifold and orbifold loci, and use the exact formulas for base degrees up to $d_B$.  Now we consider $g= 9(d_B+1)$, then according to (\ref{basisplus2.25}, \ref{basisplus2.26}), the dimension of the first linear space in the direct sum (\ref{spaceY.E.10}) is 
\begin{eqnarray}
\textrm{Dim}(V_{-}^{(1, 3d_B+2, d_B+1)}) =0.  
\end{eqnarray}
The dimensions of the other spaces in (\ref{spaceY.E.10}) are smaller and also vanish, so we have $\textrm{Dim}(Y^{(9(d_B+1),d_B)})=0$, i.e. all holomorphic ambiguities can be fixed at genus  $9(d_B+1)$. On the other hand, consider one genus higher at $g= 9(d_B+1)+1$, then the dimension of the first linear space is computed similarly 
\begin{eqnarray}
\textrm{Dim}(V_{+}^{(1, 3d_B+3, d_B+1)}) =1 >0.  
\end{eqnarray}
so we have $\textrm{Dim}(Y^{(9(d_B+1)+1,d_B)})\geq \textrm{Dim}(V_{+}^{(1, 3d_B+3, d_B+1)}) >0$, i.e. there are some unfixed holomorphic ambiguities.

\end{document}